\newif\ifanon
\anonfalse

\documentclass[11pt]{article}

\usepackage[T1]{fontenc}
\usepackage[utf8]{inputenc}
\usepackage[american]{babel}

\usepackage[margin=1in,paper=letterpaper]{geometry}

\usepackage{amsmath}
\allowdisplaybreaks 
\usepackage{amssymb}

\usepackage{accents}
\usepackage{physics}

\usepackage{xspace}
\usepackage[]{subcaption}

\usepackage[unicode]{hyperref}
\usepackage{doi}
\usepackage{mathtools}
\usepackage[many]{tcolorbox}
\usepackage{enumitem}
\usepackage{booktabs}
\usepackage{tikz}
\usepackage{tikz-cd}
\usepackage[mode=buildnew]{standalone}

\usepackage{csquotes}

\usepackage[numbers,sort&compress]{natbib}
\makeatletter
\let\oldcitet\citet
\renewcommand{\citet}{\@ifstar{\oldcitet*}{\oldcitet*}}  
\makeatother

\usepackage{microtype}

\usepackage{local-models}

\usepackage[capitalize]{cleveref}
\Crefname{question}{Question}{Questions}

\hypersetup{
    colorlinks=true,
    linkcolor=black,
    citecolor=black,
    filecolor=black,
    urlcolor=black,
}

\usetikzlibrary{shapes.geometric}
\usetikzlibrary{decorations,positioning}
\usetikzlibrary{backgrounds}

\usepackage{amthm}

\newtheorem{framedthm}[theorem]{Theorem}
\Crefname{framedthm}{Theorem}{Theorems}
\tcolorboxenvironment{framedthm}{
  sharp corners,boxrule=0.4pt,colback=white,
  before skip=\topsep,after skip=\topsep
}

\newtheorem{framedcor}[theorem]{Corollary}
\Crefname{framedcor}{Corollary}{Corollaries}
\tcolorboxenvironment{framedcor}{
  sharp corners,boxrule=0.4pt,colback=white,
  before skip=\topsep,after skip=\topsep
}

\Crefname{framedlem}{Lemma}{Lemmas}
\tcolorboxenvironment{framedlem}{
  sharp corners,boxrule=0.4pt,colback=white,
  before skip=\topsep,after skip=\topsep
}

\DeclarePairedDelimiter\ceil{\lceil}{\rceil} %
\DeclarePairedDelimiter\floor{\lfloor}{\rfloor} %

\makeatletter
\DeclareMathOperator{\regularlog}{log}
\renewcommand{\log}{\protect\@ifstar{\regularlog^*}{\regularlog}}
\makeatother

\renewcommand{\AA}{\mathcal{A}}
\newcommand{\BB}{\mathcal{B}}
\newcommand{\CC}{\mathcal{C}}

\newcommand{\FF}{\mathcal{F}}
\newcommand{\GG}{\mathcal{G}}
\newcommand{\HH}{\mathcal{H}}

\newcommand{\LL}{\mathcal{L}}

\newcommand{\NN}{\mathcal{N}}

\let\PP\relax  %
\newcommand{\PP}{\mathcal{P}}

\newcommand{\RR}{\mathcal{R}}
\renewcommand{\SS}{\mathcal{S}}
\newcommand{\TT}{\mathcal{T}}

\newcommand{\Vrand}{\mathsf{V}}

\newcommand{\Gbold}{\mathbf{G}}
\newcommand{\Hbold}{\mathbf{H}}

\newcommand{\Pbb}{\mathbb{P}}

\newcommand{\vvv}{\mathbf{v}}

\newcommand{\intg}{\mathbb{Z}}
\newcommand{\nat}{\mathbb{N}}
\DeclareMathOperator{\expect}{\mathbb{E}}
\newcommand{\natPos}{\nat_+}

\newcommand{\st}{\middle|} %
\DeclareMathOperator{\dist}{dist}

\mathchardef\mhyph="2D
\newcommand{\Mod}[1]{\ (\mathrm{mod}\ #1)}
\makeatletter
\@ifundefined{nequiv}{
  \newcommand{\nequiv}{\not\equiv}
}{}
\makeatother

\newcommand{\pr}[1]{\text{Pr}\left[#1\right]}

\newcommand{\myO}[1]{O\!\left(#1\right)}

\newcommand{\myTheta}[1]{\Theta\!\left(#1\right)}

\newcommand{\problem}{\Pi} %
\newcommand{\outcome}{\mathsf{O}} %

\newcommand{\inptData}{x} %
\newcommand{\id}{\text{id}} %
\newcommand{\lbl}{\lambda} %
\newcommand{\labels}{\Sigma} %
\newcommand{\inLabels}{\Sigma_{\text{in}}} %
\newcommand{\inLabel}{\lbl_{\text{in}}} %
\newcommand{\inLblRes}{\bar{\lbl}_{\text{in}}} %
\newcommand{\outLabels}{\Sigma_{\text{out}}} %
\newcommand{\outLabel}{\lbl_{\text{out}}} %
\newcommand{\indOutLabel}[1]{\lbl_{(\text{out},{#1})}} %
\newcommand{\outLblRes}{\bar{\lbl}_{\text{out}}} %
\newcommand{\outPr}{p} %
\newcommand{\lvlCR}{r} %
\newcommand{\algoA}{\AA} %
\newcommand{\algoB}{\BB} %
\newcommand{\algo}{\AA} %
\newcommand{\maxDeg}{\Delta} %
\newcommand{\maxInDeg}{\Delta_{\text{in}}} %
\newcommand{\maxOutDeg}{\Delta_{\text{out}}} %
\newcommand{\inDeg}{\text{indeg}} %
\newcommand{\outDeg}{\text{outdeg}} %

\newcommand{\fA}{\mathcal{A}}
\newcommand{\fF}{\mathcal{F}}

\newcommand{\poly}{\mathrm{poly}}

\definecolor{myblue}{HTML}{0088cc}
\definecolor{myorange}{HTML}{f26924}

\newcommand{\myaff}[1]{\,$\cdot$\, {\small #1}\par\medskip} 
\begin{document}
\newenvironment{myabstract}%
{\list{}{\listparindent 1.5em
        \itemindent    \listparindent
        \leftmargin    0cm
        \rightmargin   0cm
        \parsep        0pt}%
    \item\relax}%
{\endlist}

\newenvironment{mycover}%
{\list{}{\listparindent 0pt
        \itemindent    \listparindent
        \leftmargin    0cm
        \rightmargin   1.5cm
        \parsep        0pt}%
    \raggedright
    \item\relax}%
{\endlist}

\begin{mycover}
{\huge\bfseries\boldmath Online Locality Meets Distributed \\ Quantum Computing\par}
\bigskip
\bigskip

\ifanon
\textbf{Anonymous authors}

\else

\textbf{Amirreza Akbari}
\myaff{Aalto University, Finland}

\textbf{Xavier Coiteux-Roy}
\myaff{University of Calgary, Canada \,$\cdot$\, School of Computation, Information and Technology, Technical University of Munich, Germany \,$\cdot$\, MCQST, Germany}

\textbf{Francesco d'Amore}
\myaff{Gran Sasso Science Institute, Italy}

\textbf{Fran{\c c}ois Le Gall}
\myaff{Nagoya University, Japan}

\textbf{Henrik Lievonen}
\myaff{Aalto University, Finland}

\textbf{Darya Melnyk}
\myaff{Technische Universität Berlin, Germany}

\textbf{Augusto Modanese}
\myaff{Aalto University, Finland}

\textbf{Shreyas Pai}
\myaff{Indian Institute of Technology Madras, India}

\textbf{Marc-Olivier Renou}
\myaff{Inria, Université Paris-Saclay, Palaiseau, France \,$\cdot$\, CPHT, Ecole Polytechnique, Institut Polytechnique de Paris, Palaiseau, France \,$\cdot$\, LIX, Ecole Polytechnique, Institut Polytechnique de Paris, Palaiseau, France}

\textbf{Václav Rozhoň}
\myaff{Charles University, Czechia}

\textbf{Jukka Suomela}
\myaff{Aalto University, Finland}

\fi %

\bigskip
\end{mycover}

\begin{myabstract}
\noindent\textbf{Abstract.}
We connect three distinct lines of research that have recently explored extensions of the classical LOCAL model of distributed computing: 
\begin{enumerate}[noitemsep]
    \item[A.] distributed quantum computing and non-signaling distributions (e.g., [STOC 2024]),
    \item[B.] finitely-dependent processes, (e.g., [Forum Math.~Pi 2016]), and
    \item[C.] locality in online graph algorithms and dynamic graph algorithms
    (e.g., [ICALP 2023]).
\end{enumerate}
We prove new results on the capabilities and limitations of all of these models of computing, for \emph{locally checkable labeling problems} (LCLs). We show that all these settings can be sandwiched between the classical LOCAL model and what we call the \emph{randomized online-LOCAL model}. Our work implies limitations on the \emph{quantum advantage} in the distributed setting, and we also exhibit a new \emph{barrier} for proving tighter bounds. Our main technical results are these:
\begin{enumerate}
\item All LCL problems solvable with locality $O(\log* n)$ in the classical
deterministic LOCAL model admit a finitely-dependent distribution with locality
$O(1)$. This answers an open question by Holroyd
[Electron.~Commun.~Probab.~2024], and also presents a new barrier for proving
bounds on distributed quantum advantage using causality-based arguments.
\item In rooted trees, if we can solve an LCL problem with locality $o(\log \log \log n)$ in the randomized online-LOCAL model (or any of the weaker models, such as quantum-LOCAL), we can solve it with locality $O(\log* n)$ in the classical deterministic LOCAL model. One of many implications is that in rooted trees, $O(\log* n)$ locality in quantum-LOCAL is not stronger than $O(\log* n)$ locality in classical LOCAL.
\end{enumerate}
\end{myabstract}

\thispagestyle{empty}
\setcounter{page}{0}
\clearpage   %
\section{Introduction}
\label{sec:intro}

In this work, we connect three distinct lines of research that have recently explored extensions of the classical \local model of distributed computing:
\begin{enumerate}[noitemsep]
    \item[A.] Distributed quantum computing and non-signaling distributions \cite{gavoille2009,arfaoui2014,coiteuxroy2023}.
    \item[B.] Finitely-dependent processes \cite{holroyd2016,holroyd2018,holroyd2024}.
    \item[C.] Locality in online graph algorithms and dynamic graph algorithms \cite{akbari_et_al:LIPIcs.ICALP.2023.10,chang23_tight_arxiv}.
\end{enumerate}
We prove new results on the capabilities and limitations of all of these models of computing, for \emph{locally checkable labeling problems} (LCLs), with the help of a unifying model that we call \emph{randomized online-\local}. Our work implies limitations on the \emph{quantum advantage} in the distributed setting, and we also exhibit a new \emph{barrier} for proving tighter bounds.

\subsection{Highlights}

Our main technical results are the following two implications:
\begin{enumerate}
\item \textbf{All LCL problems} solvable with locality $O(\log* n)$ in the classical deterministic \local model admit a finitely-dependent distribution, i.e., a bounded-dependence distribution with constant locality:
\[
\begin{tikzcd}[row sep=1.0em]
    \text{classical \local algorithm, locality $O(\log^* n)$} \arrow[d] \\
    \text{bounded-dependence distribution, locality $O(1)$} \arrow[d] \\
    \text{non-signaling distribution, locality $O(1)$}
\end{tikzcd}
\]
This answers an open question by \citet{holroyd2024}. This also presents a \emph{new barrier} for proving bounds on distributed quantum advantage: all current \qlocal lower bounds are, in essence, lower bounds for non-signaling distributions, and our result shows that fundamentally different techniques will be needed to solve some of the biggest open questions in this area (e.g., showing that there is no constant-round \qlocal algorithm for coloring cycles).
\item \textbf{In rooted trees}, if we can solve an LCL problem with locality $o(\log \log \log n)$ in the randomized online-\local model (or any of the weaker models, such as quantum-\local), we can solve it with locality $O(\log* n)$ in the classical deterministic \local model:
\[
\begin{tikzcd}[row sep=1.0em]
    \text{quantum-LOCAL algorithm, locality $o(\log \log \log n)$} \arrow[d] \\
    \text{bounded-dependence distribution, locality $o(\log \log \log n)$} \arrow[d] \\
    \text{randomized online-LOCAL algorithm, locality $o(\log \log \log n)$} \arrow[d] \\
    \text{classical \local algorithm, locality $O(\log^* n)$}
\end{tikzcd}
\]
One of many implications is that in rooted trees, $O(\log* n)$ locality in quantum-\local is not stronger than $O(\log* n)$ locality in classical \local, and also finitely-dependent distributions are not stronger than $O(\log* n)$ locality in classical \local.
\end{enumerate}
We now proceed to explain what all of these terms and models mean and how they are connected with each other.

\subsection{Roadmap}

As our main goal is to unify and relate several distinct models studied in prior
work, we will need to introduce a fair number of models of computing. We
recommend the reader to keep the roadmap in \cref{fig:landscape} (final page) at
hand while reading the introduction in order to maintain a clear view as well as to consult this overview again when needed.

We start our adventure in \cref{ssec:classical} by introducing the classical
models that we have at the very top of \cref{fig:landscape} and then relate
these to the current landscape of LCLs in \cref{ssec:intro-lcls}. Next, we
gradually work our way through the quantum as well as bounded-dependence and
non-signaling models in \cref{ssec:intro-quantum}, after which we take our first
break. At this point, we are familiar with the top half of \cref{fig:landscape},
and we are ready to state our first main contributions related to symmetry
breaking with finitely-dependent processes in \cref{ssec:intro-contrib-1}.

In \cref{ssec:intro-online}, we then turn to models that at first may seem
completely unrelated. They deal with locality in sequential, dynamic, and online
settings. As we will see in \cref{ssec:intro-contrib-2}, however, we can
connect all of these models into a single hierarchy, with seemingly orthogonal
models sandwiched between deterministic \local and randomized \onlinelocal, and
we can prove various strong results that connect the complexity landscape
between these two extremes.

\subsection{Classical models}\label{ssec:classical}

Let us first recall the definitions of the classical models of distributed graph algorithms \cite{linial92,peleg00distributed} that form the foundation for our work; we keep these brief here and postpone formal details to \cref{sec:preliminaries}.
\begin{itemize}[noitemsep]
    \item \textbf{\boldmath Deterministic \local:} Our input graph $G = (V,E)$
    represents a computer network; each node $v \in V$ is a computer and each
    edge $\{u,v\} \in E$ is a communication link between two computers. Each
    node is labeled with a unique identifier from $\{1,2,\dotsc,\poly(|V|)\}$.
    All nodes follow the same distributed algorithm. Initially a node is only
    aware of its own identifier and its own degree. Computation proceeds in
    synchronous rounds, and in each round a node can send and receive a message
    to and from each neighbor and update its state.
    Eventually each node must stop and announce its local output (its part
    of the solution, e.g.\ in graph coloring its own color). The \emph{running
    time}, \emph{round complexity}, or \emph{locality} of the algorithm is the
    (worst-case) number of rounds $T(n)$ until the algorithm stops in any
    $n$-node graph.
    \item \textbf{\boldmath Randomized \local:} As above, but each node also has a private source of random bits.
\end{itemize}
We also define the following variants (see e.g.\ \cite{korman-2011-global-knowledge} for more on the impact of shared global information):
\begin{itemize}[noitemsep]
    \item \textbf{\boldmath Deterministic \local (shared):} Deterministic \local with shared global information. The set of nodes and their unique identifiers is globally known, and hence we also know $n = |V|$.
    \item \textbf{\boldmath Randomized \local (shared):} Randomized \local with shared global information and shared randomness. The set of nodes and their unique identifiers is globally known, and in addition to the private sources of random bits, there is also a shared source of random bits that all nodes can access.
\end{itemize}
We can interpret the shared versions of the models as follows: before any
computation, we get to \emph{see} the set of nodes $V$ and their unique
identifiers in advance, and we can also \emph{initialize} the nodes as we want
based on this information (and, hence, in the randomized model we can also
initialize all nodes with the same shared random string); the set of edges
$E$ is only revealed later. This interpretation will be useful especially with
our quantum models.

\subsection{Landscape of LCL problems}\label{ssec:intro-lcls}

There has been more than three decades of work on understanding the capabilities and limitations of the classical deterministic and randomized \local models, but for our purposes, the most interesting is the recent line of work that has studied distributed algorithms for \emph{locally checkable labeling problems}, or LCLs. This is a family of problems first introduced by \citet{naor1995}. LCL problems are graph problems that can be defined by specifying a \emph{finite set of valid neighborhoods}. Many natural problems belong to this family: coloring graphs of maximum degree $\maxDeg$ with $\maxDeg + 1$ colors, computing a maximal independent set, finding a maximal matching, etc.

Since 2016, we have seen a large body of work dedicated to understanding the
computational complexity of LCL problems in the deterministic and randomized
\local models
\cite{dahal_et_al:LIPIcs.DISC.2023.40,balliu18lcl-complexity,balliu20almost-global,doi:10.1137/17M1157957,dist_deran,balliu20lcl-randomness,fischer_ghaffari2017sublogarithmic,Rozhon2019,brandt16lll,chang16exponential,ghaffari17distributed,ghaffari19degree-splitting,balliu19lcl-decidability,balliu21mm,brandt17grid-lcl,balliu22rooted-trees,balliu22regular-trees},
and nowadays there are even algorithms and computer tools available for
exploring such questions
\cite{Olivetti2019,balliu22regular-trees,chang23automata-theoretic}. As a result
of this large international research effort, a landscape of the localities of
LCL problems emerges \cite{suomela-2020-landscape}. We can classify LCL problems
in discrete classes based on their locality, and we also understand how much
randomness helps in comparison with deterministic algorithms. Our main goal in this work is to extend this understanding of LCL problems far beyond the classical models, and especially explore what can be computed very fast in models that are much stronger than deterministic or randomized \local.

\subsection{\boldmath \Qlocal and finitely-dependent processes}\label{ssec:intro-quantum}

We start our exploration of stronger models with distributed quantum computation. The key question is understanding the \emph{distributed quantum advantage}: what can we solve faster if our nodes are quantum computers and our edges are quantum communication channels? There is a long line of prior work exploring this theme in different models of distributed computing \cite{legall2019,arfaoui2014,gavoille2009,elkin2014,legall2018,magniez2022,wu2022,wang2022,izumi2019,censorhillel2022,izumi2020,apeldoorn2022,coiteuxroy2023,fraigniaud2024,hasegawa2024}, but for our purposes these are the models of interest:
\begin{itemize}[noitemsep]
	\item \textbf{\boldmath \Qlocal:} This model of computing is similar to the deterministic \local model above, but now with quantum computers and quantum communication links. More precisely, the quantum computers manipulate local states consisting of an unbounded number of qubits with arbitrary unitary transformations, the communication links are quantum communication channels (adjacent nodes can exchange any number of qubits), and the local output can be the result of any quantum measurement.
    \item \textbf{\boldmath \Qlocal (shared):} \Qlocal with shared
    global information and a shared quantum state.
    As above, but now the algorithm may inspect and manipulate the set of nodes
    (before any edges are revealed). 
    In particular, it may initialize the nodes with a globally shared entangled
    state.
\end{itemize}
As quantum theory intrinsically involves randomness, \qlocal is at least as strong as randomized \local. There are some (artificial) problems that are known to be solvable much faster in \qlocal than deterministic or randomized \local \cite{legall2019}; however, whether any LCL admits such a quantum advantage is a major open question in the field.

Directly analyzing \qlocal is beyond the scope of current techniques. In essence, the only known technique for proving limitations of \qlocal 
is \emph{sandwiching} it between the classical randomized-\local model and more powerful models than \qlocal that do not explicitly refer to quantum information. These more powerful models are based on the \emph{physical causality principle} (a.k.a.\ \emph{non-signaling principle}). The idea is perhaps easiest to understand with the help of the following thought experiment:
\begin{example}
Fix a distributed algorithm $\algoA$ in the \qlocal model (with shared
global information and quantum state) that runs in $T$ rounds on graphs with $n$
nodes. Let $G = (V,E)$ be some $n$-node input graph. Apply $\algoA$ repeatedly
to $G$ to obtain some probability distribution $Y(G)$ of outputs. Now fix some
subset of nodes $U \subseteq V$, and consider the restriction of $Y(G)$ to $U$,
in notation $Y(G) {\restriction_U}$. Let $G[U,T]$ be the radius-$T$ neighborhood
of set $U$ in $G$. Now modify $G$ outside $G[U,T]$ to obtain a different
$n$-node graph $G'$ with $G[U,T] = G'[U,T]$. Apply $\algoA$ to $G'$ repeatedly,
and we obtain another probability distribution $Y(G')$ of outputs. If
$Y(G){\restriction_U} \ne Y(G'){\restriction_U}$, it would be possible to use
$\algoA$ to transmit information in $T$ time steps between two parties, Alice
and Bob, that are within distance $T+1$ from each other: Bob holds all nodes of
$U$, and he can, therefore, observe $Y(G){\restriction_U}$, while Alice controls
the graph outside $G[U,T] = G'[U,T]$, and she can, therefore, instantiate either
$G$ or $G'$. This would enable Alice to send a signal to Bob even if no physical
communication occurred from Alice to Bob (as they are at distance $T+1$ from
each other and only $T$ communication steps occurred), and thus violate the
non-signaling principle.
\end{example}

This thought experiment suggests the following definition, also known as the $\varphi$-\local model or the causal model \cite{gavoille2009,arfaoui2014}:
\begin{itemize}
    \item \textbf{\boldmath Non-signaling model:} We can produce an arbitrary output distribution as long as it does not violate the \emph{non-signaling principle}: for any set of nodes $U$, modifying the structure of the input graph at more than a distance $T(n)$ from $U$ does not affect the output distribution of $U$.
\end{itemize}
We also need to introduce the following definition to better connect our work with the study of finitely-dependent processes and in particular finitely-dependent colorings~\cite{holroyd2016,holroyd2018,holroyd2024}:
\begin{itemize}
    \item \textbf{\boldmath Bounded-dependence model:} We can produce an arbitrary output distribution as long as it does not violate the non-signaling principle and, furthermore, \emph{distant parts are independent}: if we fix any sets of nodes $U_1$ and $U_2$ such that their radius-$T(n)$ neighborhoods are disjoint, then the output labels of $U_1$ are independent of the output labels of $U_2$.
\end{itemize}
Setting $T(n) = O(1)$ in the bounded-dependence model, we recover in essence what is usually called \emph{finitely-dependent processes}.

Now we can connect all the models above as follows, sandwiching the two versions of \qlocal between other models (see \cref{app:dc-models} for details; the connection with the non-signaling model is known \cite{gavoille2009,arfaoui2014} but the connection with the bounded-dependence model is to our knowledge new):
\begin{equation}\label{eq:diagram-quantum}
\begin{tikzcd}[row sep=0.8em]
    \text{deterministic \local} \arrow[r]\arrow[d] & \text{deterministic \local (shared)} \arrow[d] \\
    \text{randomized \local} \arrow[r]\arrow[d] & \text{randomized \local (shared)} \arrow[d] \\
    \text{\qlocal} \arrow[r]\arrow[d] & \text{\qlocal (shared)} \arrow[d] \\
    \text{bounded-dependence} \arrow[r] & \text{non-signaling}
\end{tikzcd}
\end{equation}
Here an arrow $M_1 \rightarrow M_2$ indicates that an algorithm with locality (or round complexity) $T(n)$ in model $M_1$ implies an algorithm with the same locality in $M_2$. For some problems it is possible to prove near-tight bounds for \qlocal by using diagram \eqref{eq:diagram-quantum}. For example, a very recent work \cite{coiteuxroy2023} used these connections to prove limits for the distributed quantum advantage in approximate graph coloring: they prove an upper bound for the deterministic \local model and a near-matching lower bound for the non-signaling model.

\subsection{Contribution 1: symmetry breaking with finitely-dependent processes}\label{ssec:intro-contrib-1}

Now we are ready to state our first contribution. Recall the following gap result by \citet{chang_kopelowitz_pettie2019exp_separation}: all LCL problems that can be solved with locality $o(\log n)$ in deterministic \local or with locality $o(\log \log n)$ in randomized \local can also be solved with locality $O(\log* n)$ in deterministic \local. The class of problems with locality $\Theta(\log* n)$ contains in essence all \emph{symmetry-breaking problems}: these are problems that could be solved with constant locality if only we had some means of breaking symmetry (e.g.\ distance-$k$ coloring for some constant $k$ would suffice). In \cref{sec:fin-dep} we show the following result:
\begin{framedthm}{}\label{thm:intro-fin-dept}
  Let $\Pi$ be any LCL problem with locality $O(\log* n)$ in the
  deterministic \local model. Then $\Pi$ can also be solved with locality
  $O(1)$ in the bounded-dependence model.
  Furthermore, the resulting finitely-dependent processes are invariant under subgraph isomorphism.
\end{framedthm}
\noindent Put otherwise, there is a finitely-dependent distribution over valid solutions of
$\Pi$.
Here the invariance under subgraph isomorphisms implies that, for any two graphs \(G,H\) that share some isomorphic subgraphs \(G'\) and \(H'\) such that their radius-\(\myO{1}\) neighborhoods are still isomorphic, the finitely-dependent processes solving \(\problem\) restricted to \(G'\) and \(H'\) are equal.

For any constant $d$, the task of coloring $d$-regular trees with $d+1$ colors is a problem with locality $O(\log* n)$ in the deterministic \local model. Hence, we can \textbf{answer the open question} by 
\citet{holroyd2024}:
\begin{framedcor}\label{cor:intro:fin-dep-trees}
    For each $d \ge 2$, there is a finitely-dependent coloring with $d+1$ colors in $d$-regular trees. 
    Furthermore, the resulting process is invariant under subgraph isomorphisms.
\end{framedcor}
\noindent 
More specifically, there exists a finitely-dependent 4-coloring distribution of the infinite 3-regular tree that is invariant under automorphisms. 

\cref{thm:intro-fin-dept} also introduces a formal \textbf{barrier} for proving limitations on distributed quantum advantage. Recall that all current \qlocal lower bounds are, in essence, lower bounds in the non-signaling model. Before our work, there was a hope that we could discover a symmetry-breaking problem $\Pi$ with the following properties: (1)~its locality is $O(\log* n)$ in deterministic \local, and (2)~we can show that its locality is $\Omega(\log* n)$ in the non-signaling model, and therefore (3)~$\Pi$ cannot admit any distributed quantum advantage. However, our work shows that no such problem $\Pi$ can exist. In particular, \emph{arguments related to non-signaling distributions are not sufficient} to exclude distributed quantum advantage in this region.

\subsection{Locality in online and dynamic settings}\label{ssec:intro-online}

Let us now switch gears and consider a very different line of work. 
\citet{ghaffari2017} introduced a \emph{sequential} counterpart of the classical \local model:
\begin{itemize}[noitemsep]
    \item \textbf{\boldmath Deterministic \slocal model:} The nodes are processed in an adversarial order. When a node $v$ is processed, the algorithm gets to see all information in its radius-$T(n)$ neighborhood (including states and outputs of previously processed nodes). The algorithm has to label $v$ with its local output, and the algorithm can also record other information in $v$, which it can exploit when other nodes near $v$ are later processed.
    \item \textbf{\boldmath Randomized \slocal model:} As above, with access to a source of random bits.
\end{itemize}
Clearly \slocal is stronger than \local. One key feature is that the processing order naturally breaks symmetry, and all symmetry-breaking LCLs can be solved with $O(1)$ locality in \slocal. One interpretation of our first contribution is that we also establish a new, unexpected similarity between \slocal and the bounded-dependence model: \emph{both are able to solve any symmetry-breaking LCL with constant locality}.

A recent work \cite{akbari_et_al:LIPIcs.ICALP.2023.10} introduced the following models that capture the notion of locality also in the context of centralized dynamic graph algorithms and centralized online graph algorithms:
\begin{itemize}[noitemsep]
    \item \textbf{\boldmath Deterministic \dlocal model:} The adversary constructs the graph one edge at a time. The algorithm has a global view of the graph (including states and outputs of previously processed nodes), but it has to maintain a feasible solution after each update. The algorithm is restricted so that, after a modification at node $v$, it can only update the solution within distance $T(n)$ from $v$.
    \item \textbf{\boldmath Deterministic \onlinelocal model:} The adversary
    presents the input graph one node at a time. Upon arrival of a node $v$, the
    adversary also reveals the radius-$T(n)$ neighborhood of $v$. The algorithm
    then has to then choose the output label of $v$.
    Crucially, the algorithm has access to a global view of the graph revealed
    thus far (including states and outputs of previously processed nodes).
\end{itemize}
\Onlinelocal can also be seen as a stronger version of the deterministic \slocal model where all nodes have access to some \emph{global shared memory}.
Both \slocal and \dlocal can be sandwiched between \local and \onlinelocal \cite{akbari_et_al:LIPIcs.ICALP.2023.10}:
\begin{equation}\label{eq:diagram-online}
\begin{tikzcd}[row sep=1.0em]
    \text{deterministic \local} \arrow[r]\arrow[d] & \text{deterministic \slocal} \arrow[d] \\
    \text{deterministic \dlocal} \arrow[r] & \text{deterministic \onlinelocal}
\end{tikzcd}
\end{equation}
There are also some problems in which deterministic \onlinelocal is much stronger than deterministic \local: $3$-coloring in bipartite graphs has locality $\tilde{\Theta}(\sqrt{n})$ in deterministic \local \cite{brandt17grid-lcl,coiteuxroy2023} but $O(\log n)$ in \onlinelocal \cite{akbari_et_al:LIPIcs.ICALP.2023.10}; very recently \citet{chang23_tight_arxiv} also showed that this is tight for \onlinelocal.

\subsection{Contribution 2: connecting all models for LCLs}\label{ssec:intro-contrib-2}

At first sight, the models discussed in \cref{ssec:intro-quantum,ssec:intro-online} seem to have very little in common; they seem to be orthogonal extensions of the classical deterministic \local model. Furthermore, we have already seen evidence that \onlinelocal can be much stronger than deterministic \local. Nevertheless, we can connect all these models in a unified manner, and prove strong limits on their expressive power. To this end, we introduce yet another model:
\begin{itemize}
    \item \textbf{\boldmath Randomized \onlinelocal model:} Like deterministic
    \onlinelocal but the algorithm has access to a source of random bits, and we
    play against an oblivious adversary (i.e., the adversary fixes the graph and
    the order of presenting nodes before the algorithm starts to flip coins).
\end{itemize}
Trivially, this is at least as strong as all models in diagram \eqref{eq:diagram-online}. However, the big surprise is that it is also at least as strong as all models in diagram \eqref{eq:diagram-quantum}. In \cref{sec:non-signaling-to-rolocal} we prove:
\begin{framedthm}{}\label{thm:intro:ns-model-sim}
    Any labeling problem that can be solved in the non-signaling model with locality $T(n)$ can also be solved in the randomized \onlinelocal model with the same locality.
\end{framedthm}
\noindent
Then we zoom into the case of rooted trees in \cref{sec:olocal-slocal-simulation} and prove:
\begin{framedthm}{}\label{thm:intro-rolocal-simulation}
    Any LCL on rooted trees that can be solved in the randomized \onlinelocal model with locality $o(\log \log \log n)$ can also be solved in the deterministic \local model with locality $O(\log* n)$.
\end{framedthm}
\noindent Together with \cref{thm:intro-fin-dept} and results that were
previously known, we also obtain the following:
\begin{framedcor}{}
    In rooted trees, the following families of LCLs are the same:
    \begin{itemize}[noitemsep]
        \item locality $O(\log* n)$ in deterministic or randomized \local or
        \qlocal; and
        \item locality $O(1)$ in bounded-dependence model, non-signaling model, deterministic or randomized \slocal, \dlocal, or deterministic or randomized \onlinelocal.
    \end{itemize}
    Still in rooted trees, there is no LCL problem with locality between $\omega(\log* n)$ and $o(\log \log \log n)$ in any of these models: deterministic and randomized \local, \qlocal, bounded-dependence model, non-signaling model, deterministic and randomized \slocal, \dlocal, and deterministic and randomized \onlinelocal.
\end{framedcor}
\noindent In particular, when we look at LCLs in rooted trees, $O(\log* n)$-round quantum algorithms are not any stronger than $O(\log* n)$-round classical algorithms. (However, it is still possible that there are some LCLs in trees that can be solved in $O(1)$ rounds in \qlocal and that require $\Theta(\log* n)$ rounds in deterministic \local; recall the discussion in \cref{ssec:intro-contrib-1}.)

\cref{thm:intro-rolocal-simulation} can be seen as an extension of the result of \citet{akbari_et_al:LIPIcs.ICALP.2023.10} that connects \emph{deterministic} online-\local with \local for LCLs on rooted \emph{regular} trees \emph{without inputs}. Our result is applicable to the randomized online-\local (and hence we can connect it with the non-signaling and quantum models), and it holds for any LCLs on rooted trees (possibly with irregularities and inputs). Recall that the presence of inputs makes a huge difference already in the case of directed paths \cite{chang23automata-theoretic,balliu19lcl-decidability}, and we need fundamentally new ideas as we can no longer build on the classification from \cite{balliu22regular-trees,balliu22rooted-trees}.

\subsection{Big picture}

By putting together all our contributions (including some auxiliary results that we discuss later in \cref{sec:extra}), the landscape shown in \cref{fig:landscape} emerges. We can sandwich all models between deterministic \local and randomized \onlinelocal. Going downwards, we get symmetry breaking for free (as indicated by the blue arrows). In the case of rooted trees in the low-locality region $o(\log \log \log n)$, we can also navigate upwards (as indicated by the orange arrows). \Cref{tab:problem-summary} gives some concrete examples of localities for LCL problems across the landscape of models.

\begin{table}
    \centering
    \caption{Examples of localities across the models. Here \emph{symmetry-breaking problems} refer to LCL problems with locality $\Theta(\log^* n)$ in the deterministic \local model; this includes many classical problems such as maximal independent set, maximal matching, and $(\Delta+1)$-vertex coloring.}
    \begin{tabular}{@{}l@{\quad}ll@{\qquad}ll@{}}
        \toprule
        Model
            & \multicolumn{2}{@{}l@{\qquad}}{Symmetry-breaking problems}
            & \multicolumn{2}{@{}l@{}}{3-coloring in bipartite graphs}
        \\
        \midrule
        Deterministic \local
            & $\Theta(\log* n)$ & by definition
            & $\tilde{\Theta}(\sqrt{n})$ & \cite{coiteuxroy2023}
        \\
        Randomized \local
            & $\Theta(\log* n)$ & \cite{doi:10.1137/17M1157957}
            & $\tilde{\Theta}(\sqrt{n})$ & \cite{coiteuxroy2023}
        \\
        Quantum \local
            & $O(\log* n)$ & trivial
            & $\tilde{\Theta}(\sqrt{n})$ & \cite{coiteuxroy2023}
        \\        
        Bounded-dependence
            & $O(1)$ & \cref{thm:fin-dep}
            & $\tilde{\Theta}(\sqrt{n})$ & \cite{coiteuxroy2023}
        \\        
        Deterministic \slocal
            & $O(1)$ & \cite{ghaffari2017}
            & $\tilde{O}(\sqrt{n})$, $n^{\Omega(1)}$ & \cite{coiteuxroy2023,akbari_et_al:LIPIcs.ICALP.2023.10}
        \\
        Randomized \slocal
            & $O(1)$ & \cite{ghaffari2017}
            & $\tilde{O}(\sqrt{n})$, $n^{\Omega(1)}$ & \cite{coiteuxroy2023,akbari_et_al:LIPIcs.ICALP.2023.10}
        \\
        Deterministic \dlocal
            & \(O(1)\) & \cite{akbari_et_al:LIPIcs.ICALP.2023.10}
            & $\tilde{O}(\sqrt{n})$, $\Omega(\log n)$ & \cite{coiteuxroy2023,chang23_tight_arxiv}
        \\
        Deterministic \onlinelocal
            & $O(1)$ & \cite{akbari_et_al:LIPIcs.ICALP.2023.10}
            & $\Theta(\log n)$ & \cite{akbari_et_al:LIPIcs.ICALP.2023.10,chang23_tight_arxiv}
        \\
        Randomized \onlinelocal
            & $O(1)$ & \cite{akbari_et_al:LIPIcs.ICALP.2023.10}
            & $\Theta(\log n)$ & \cref{thm:three-col-ro-lcl}
        \\
        \bottomrule
    \end{tabular}
    \label{tab:problem-summary}
\end{table}
\section{Overview of techniques and key ideas}\label{sec:overview}
In this section, we give an overview of the techniques and key ideas that we use to prove our main results, and we also provide a roadmap to the rest of this paper.
We note that our first contribution is presented in \cref{sec:fin-dep}, while the second contribution comes before it in \cref{sec:non-signaling-to-rolocal}---the proofs are ordered this way since \cref{sec:non-signaling-to-rolocal} also develops definitions that will be useful in \cref{sec:fin-dep}.

\subsection[Bounded-dependence model can break symmetry]{Bounded-dependence model can break symmetry (\cref{sec:fin-dep})}\label{sec:overview:fin-dep}

Let us first present an overview of the proof of \cref{thm:intro-fin-dept} from \cref{ssec:intro-contrib-1}. We show that the bounded-dependence model can break symmetry with constant locality; that is, there is a finitely-dependent process for any symmetry-breaking LCL.

It is well known that any LCL problem \(\problem\) that has complexity \(\myO{\log* n}\) in the \local model has the following property: 
there exists a constant \(k \in \natPos\) (that depends only on the hidden constant in \(\myO{\log* n}\)) such that, if the graph is given a distance-\(k\) coloring (with sufficiently small number of colors) as input, then \(\problem\) is solvable in time \(\myO{1}\) in the \local model  (using the distance-\(k\) coloring as a local assignment of identifiers) \cite{chang16exponential}.

We prove that for each bounded-degree graph, there is a finitely-dependent process providing a distance-\(k\) coloring for constant \(k\).
Then, we can combine such a process with the \local algorithm that solves the problem with locality \(\myO{1}\) if a distance-\(k\) coloring is given, and we prove that the resulting process is still a finitely-dependent distribution.
Furthermore, we also prove that all these processes are invariant under subgraph isomorphisms (even if they do not preserve node identifiers), meaning that, for any two graphs \(G\) and \(H\) sharing two isomorphic subgraphs with isomorphic radius-\(\myO{1}\) neighborhoods, the restrictions of the finitely-dependent processes solving \(\problem\) over \(G\) and \(H\) restricted to \(G'\) and \(H'\) are equal in law.

One of the key observations that we use is that \local algorithms that do not exploit the specific assignment of node identifiers and do not depend on the size of the graph provide finitely-dependent distributions that are invariant under subgraph isomorphisms whenever the input labeling for the graphs is invariant under subgraph isomorphisms. 

\paragraph{Overview.}

The cornerstone of our proof is a surprising result by \citet{holroyd2016} and its follow-up in \cite{holroyd2018}, that state that there exist \(k\)-dependent distributions giving a \(q\)-coloring of the infinite path and of cycles for \((k,q) \in \{(1,4),(2,3)\}\) that are invariant under subgraph isomorphisms.

Recently, Holroyd has combined the finitely-dependent distributions of infinite paths to provide a finitely-dependent \(4\)-coloring of the \(d\)-dimensional lattice \cite{holroyd2024}.
Getting a translation invariant distribution is quite easy: First, use the distributions for the paths on each horizontal and vertical path obtaining a distance-\(k\) coloring (with \(k\) being a large enough constant) of the lattice with constantly many colors as shown in \cite[Corollary 20]{holroyd2016}. 
Second,
apply some \local algorithm that starts from a distance-\(k\) coloring and reduces the number of colors to \(4\) while keeping the resulting distribution symmetric (e.g., the algorithms from \cite{brandt17grid-lcl,balliu22mending}).
The major contribution of \cite{holroyd2024} is transforming such a distribution into a process that is invariant under subgraph isomorphisms.
However, this \emph{symmetrization} phase is quite specific to the considered topology.

We come up with a new approach that obtains similar results in all bounded-degree graphs through the following steps:
\begin{enumerate}[noitemsep]
    \item We show that the finitely-dependent coloring of paths and cycles can be combined to obtain finitely-dependent \(3\)-coloring distributions of rooted pseudoforests of bounded-degree that are invariant under subgraph isomorphisms.
    \item We observe that all graphs of bounded-degree admit a random decomposition in rooted pseudoforests that satisfies the required symmetry properties.
    \item We prove that such a random decomposition can be combined with the finitely-dependent \(3\)-coloring of rooted pseudoforests to obtain finitely-dependent distributions that give a \((\maxDeg+1)\)-coloring of graphs of maximum degree \(\maxDeg\) that are invariant under subgraph isomorphism.
    \item We show that we can use this finitely-dependent \((\maxDeg+1)\)-coloring distribution to provide a distance-\(k\) coloring for bounded-degree graphs, which is enough to simulate in constant time any \(\myO{\log* n}\)-round \local algorithm that solves an LCL \(\problem\).
    Such combinations result in finitely-dependent processes that are invariant under subgraph isomorphisms and solve~\(\problem\).
\end{enumerate}

Notice that, in spirit, steps 1 to 3 are similar to the steps needed to produce a \((\maxDeg + 1)\)-coloring in time \(\myO{\log* n}\) in the \local model \cite{goldberg1988,panconesi2001}:
however, the detailed way these steps are obtained in the bounded-dependence model is quite different and requires a careful analysis.

\paragraph{\boldmath Step 1: Finitely-dependent \(3\)-coloring distributions of rooted pseudoforests.}
A rooted pseudoforest is a directed graph in which each node has outdegree at most $1$.
Let us now fix any rooted pseudoforest of maximum degree \(\maxDeg\).
Consider the following process: each node \(v\) colors its in-neighbors with a uniformly random permutation of \(\{1, \dots, \inDeg(v)\}\).
The graph \(G_i\) induced by nodes colored with color \(i\) is a union of directed paths and cycles (see \cref{fig:fin-dep:rooted-pseudoforest}) and, hence, admits a finitely-dependent \(4\)-coloring given by \cite{holroyd2016} that is invariant under subgraph isomorphisms; if a node is isolated, it can deterministically join any of the \(G_i\)s, say \(G_1\).
The sequence of graphs \((G_1, \dots, G_{\maxInDeg})\) is said to be a random \(\maxInDeg\)-decomposition of the rooted pseudoforest.
Furthermore, if two graphs \(G,H\) have isomorphic subgraphs \(G',H'\) (together with some constant-radius neighborhoods), the decompositions in directed paths and cycles induced in \(G'\) and \(H'\) have the same distribution (because node colors are locally chosen uniformly).
We prove that the combination of the random decomposition and the finitely-dependent coloring yields a finitely-dependent \(4\maxDeg\)-coloring which is invariant under subgraph isomorphisms: by further combining such distribution with the Cole--Vishkin color reduction algorithm \cite{cole1986,goldberg1988,panconesi2001} (that has complexity \(\myO{\log* k}\) with \(k\) being the size of the input coloring), we can obtain a finitely-dependent \(3\)-coloring distribution for rooted pseudoforests of maximum degree \(\maxDeg\) that is invariant under subgraph isomorphisms.

\paragraph{\boldmath Steps 2--3: Finitely-dependent \((\maxDeg+1)\)-coloring distribution of bounded-degree graphs.}

First, if the input graph is undirected, make it a directed graph by duplicating all edges and assigning both orientations to duplicated edges.
Since a coloring of the nodes can be given in both cases equivalently, we focus on the directed case for simplicity.
Second, consider the following process: each node \(v\) labels its out-edges with a uniformly sampled permutation of the elements of \(\{1, \dots, \outDeg(v)\}\); this way we obtain a random decomposition of the edges of the graph into rooted pseudoforests, as each node has at most one out-edge with label \(i\).
Furthermore, if two graphs \(G,H\) have isomorphic subgraphs \(G',H'\) (together with some constant-radius neighborhoods), the decompositions induced in \(G'\) and \(H'\) have the same distribution (because edge labelings are locally chosen uniformly).
We prove that if we apply the finitely-dependent \(3\)-coloring from step 1 to each pseudoforest, we obtain a finitely-dependent \(3^\maxDeg\)-coloring of the input graph which is invariant under subgraph isomorphisms. By further combining such a distribution with a variant of the Cole--Vishkin color reduction algorithm, we obtain a finitely-dependent \((\maxDeg + 1)\)-coloring distribution for bounded-degree graphs of maximum degree \(\maxDeg\) that is invariant under subgraph isomorphisms.

\paragraph{\boldmath Step 4: Finitely-dependent distribution solving \(\problem\).}
Consider any graph \(G\) of maximum degree \(\maxDeg\), and its \(k\)-th power graph defined as follows: simply add edges to \(G\) between each pair of nodes at distance at most \(k\), where \(k\) is some large enough constant.
Observe that \(G^k\) is a graph of maximum degree \(\maxDeg^k\).
Now, step 3 implies that there is a finitely-dependent \((\maxDeg^k + 1)\)-coloring of \(G^k\) that is invariant under subgraph isomorphisms: such distribution yields a distance-\(k\) coloring of~\(G\).
For any LCL \(\problem\) that has complexity \(\myO{\log*n}\) in \local, we know that there exists a constant \(k\) such that, if given a distance-\(k\) coloring as input (with a constant number of colors), then there is an \(\myO{1}\)-round \portnum algorithm solving \(\problem\) \cite{chang16exponential}: the combination of the input distance-\(k\) coloring of \(G\) given by step 2-3 with such an algorithm yields a finitely-dependent distribution solving \(\problem\) that is invariant under subgraph isomorphisms.

\paragraph{Random decomposition of a graph.}
In steps 1 and 3 we proceed in an analogous way: First, we construct a process that induces a random decomposition of a graph.
Second, we consider finitely-dependent distributions of output labelings over the outputs of the random decomposition.
The combination of the random decomposition and the finitely-dependent distributions gives rise to a process over the whole graph.
In \cref{sec:fin-dep}, we derive a general result
(\cref{lemma:fin-dep:induced-process}) which gives sufficient conditions on the
random decomposition and the finitely-dependent distributions in order to ensure
the final process is still finitely-dependent (possibly with symmetry
properties).
\cref{lemma:fin-dep:induced-process} is then the tool used in practice in steps
1 and 3.

\paragraph{Independent related work.} 
Very recently, an independent and parallel work provided a finitely-dependent coloring of bounded-degree graphs with exponentially many colors (in the degree of the graph) \cite{timar2024}. 
The technique employed in \cite{timar2024} is very similar to ours: 
it exploits the decomposition of graphs in rooted pseudoforests, and then colors rooted pseudoforests using the finitely-dependent coloring of paths and cycles \cite{holroyd2016,holroyd2018}.
However, \cite{timar2024} stops at the mere coloring problem and does not make use of color reduction algorithms, which are the key ingredient for extending results to all symmetry-breaking LCLs.

\subsection[Simulating non-signaling in randomized \onlinelocal]{\boldmath Simulating non-signaling in randomized \onlinelocal (\cref{sec:non-signaling-to-rolocal})}

Let us now give the intuition behind the proof of \cref{thm:intro:ns-model-sim} from \cref{ssec:intro-contrib-2}: we show that the non-signaling model can be simulated in \rolcl without any loss in the locality.

A \rolcl algorithm is given as input the size of the input graph, and a distribution that is non-signaling beyond distance \(T\) and that solves some problem \(\problem\) over some graph family \(\FF\).
When the adversary picks any node \(v_1\) and shows the \rolcl algorithm its radius-\(T\) neighborhood, the \rolcl algorithm simply goes over all graphs of \(n\) nodes in \(\FF\) until it finds one, say \(H_1\), that includes the radius-\(T\) neighborhood of \(v_1\): then, it samples an output for \(v_1\) according to the restriction of the non-signaling distribution to the radius-\(T\) neighborhood of \(v_1\).
Notice that such distribution does not change if the topology of the graph changes outside the radius-\(T\) neighborhood of \(v_1\).
Recursively, when the adversary picks the \(i\)-th node \(v_i\), the \rolcl algorithm goes over all graphs of \(n\) nodes in \(\FF\) until it finds one, say \(H_i\), that includes the union of radius-\(T\) neighborhoods of \(v_1, \dots, v_i\) (it must necessarily exist as the graph chosen by the adversary is a valid input): hence, it samples an output for \(v_i\) sampling from the restriction of the non-signaling distributions to the union of the radius-\(T\) neighborhoods of \(v_1, \dots, v_{i-1}\), conditional on the outputs of \(v_1, \dots, v_{i-1}\).
We prove that the non-signaling property ensures that the algorithm described above fails with at most the same probability of failure of the non-signaling distribution.

\subsection[\Onlinelocal can be simulated in \slocal for rooted trees]{\boldmath \Onlinelocal can be simulated in \slocal for rooted trees (\cref{sec:olocal-slocal-simulation})}
\label{ssec:overview-olocal-slocal-simulation}

Next we give an overview of the proof of \cref{thm:intro-rolocal-simulation} from \cref{ssec:intro-contrib-2}: we show that a randomized online-\local algorithm that solves an LCL problem in rooted trees with locality $o(\log\log\log n)$ can be simulated in the deterministic \slocal model with locality $O(1)$, and therefore also in the deterministic \local model with locality $O(\log^* n)$.

The new ingredient we use in this section is \emph{component-wise \onlinelocal
algorithms}. Roughly speaking, a \emph{component-wise} algorithm is a
deterministic \onlinelocal algorithm that, when processing a node \(v\), uses
information only coming from the connected component of the input graph that has
been revealed so far to which \(v\) belongs, and nothing else.
(If two or more components are merged, then the algorithm may use information it knows
from any component.)

We prove \cref{thm:intro-rolocal-simulation} in three steps:
\begin{enumerate}[noitemsep]
    \item We first show that any randomized online-\local algorithm solving an LCL with locality \(T(n)\) can be turned into a deterministic component-wise online-\local algorithm with locality \(T(2^{O(2^{n^2})})\).
    \item We then prove that, for LCLs on rooted trees, we can simulate the
    component-wise algorithms in \slocal.
    \item Finally, we show that \slocal algorithms solving any LCL \(\problem\) with locality \(o(\log n)\) over rooted trees can be turned into an \(O(\log* n)\)-round \local algorithm solving \(\problem\) over rooted trees.
\end{enumerate}

\paragraph{Step 1: Constructing component-wise algorithms from deterministic online-LOCAL algorithms.}
To give some intuition, consider an LCL problem $\Pi$ on a family \(\FF\)  of
graphs that is closed under disjoint graph union and node and edge removals.
Suppose there is a \emph{deterministic} \onlinelocal algorithm $\algoA$ solving
$\Pi$ with locality $T(n)$ on $n$-node graphs.
Note that the output label $\algoA$ chooses for a node may depend arbitrarily on
everything the algorithm has seen so far.

We show how to turn algorithm~$\algoA$ into an algorithm whose output for
\emph{isolated} nodes depends only on the local topology and inputs and is
oblivious to any previously-processed nodes.
We call such algorithms $1$-amnesiac.
Here, with ``isolated node'', we mean that the node $v$ is such that all nodes belonging to the
radius \(T\)-neighborhood around \(v\) are new to algorithm~$\algoA$, that is,
\(\algoA\) has no knowledge of how $v$ is connected (if at all) to the parts of
the graph it has seen thus far.

Let \(\inLabels, \outLabels\) be the sets of input and output labels of \(\problem\), respectively.
Let \(N_1 = \abs*{\outLabels}\abs*{\inLabels}^n \cdot 2^{n^2} \cdot n^2\), and
let $\GG_1$ be the set of all possible subgraphs of any $n$-node graphs (from \(\FF\)) with inputs (up to isomorphisms) that are the radius-\(T(N_1)\) neighborhood of some node, which we call the \emph{center} of the neighborhood.
Let $g_1 = |\GG_1|$ and notice that $g_1 \le 2^{n^2} \abs*{\inLabels}^n$.
Consider now the following experiment:
\begin{enumerate}[noitemsep]
  \item Construct a simulation graph \(H_1\) that consists of $\abs*{\outLabels}
  \cdot n$ copies of all graphs in $\GG_1$.
  The size of \(H_1\) is at most \(N_1 = n^2 \abs*{\outLabels} \cdot g_1   \le
  \abs*{\outLabels}\abs*{\inLabels}^n \cdot 2^{n^2} \cdot n^2\).
  \item Reveal the center node of each of those neighborhood graphs to $\algoA$
  in an arbitrary order with locality~$T(N_1)$.
  For each \emph{type of radius-\(T(N_1)\) neighborhood} ${\TT_1}$ (i.e., any
  element of \(\GG_1\)), there exists some output label $\sigma_{\TT_1}$ that
  occurs at least $n$ times.
  We call such neighborhoods \emph{good} and such a label a \emph{canonical
  labeling} of \({\TT_1}\).
  \item Continue labelling nodes of \(H_1\) using \(\algoA\) under an arbitrary
  ordering of the nodes.
\end{enumerate}

We describe a new \onlinelocal algorithm~$\algoB$ that, using this experiment,
produces a correct labeling.
Let \(G\) be an input graph with \(n\) nodes, and let node~$v$ be revealed to \(\algoB\) along with its radius-\(T'(n)\) neighborhood, where \(T'(n) = T(N_1)\). 
Whenever \(v\) is an isolated node, algorithm~$\algoB$ finds a \enquote{fresh}
(i.e., not previously chosen) good neighborhood in the experiment graph matching
the radius-\(T(N_1)\) neighborhood of \(v\) in \(G\).
It then takes that unused good neighborhood, identifies all nodes with the
revealed input neighborhood, and labels $v$ accordingly.
An unused good neighborhood always exists since there are at least \(n\) good
neighborhoods in the experiment graph matching the radius-\(T(N_1)\)
neighborhood of \(v\) in \(G\).
Algorithm~$\algoB$ effectively cuts and pastes the neighborhood from the experiment graph to the actual input graph \emph{without algorithm~$\algoA$ noticing}.
For non-isolated nodes, \(\algoB\) just simulates what \(\algoA\) would do,
following the adversarial order of the nodes presented to \(\algoB\). 
This is always possible because the labels of isolated nodes come from valid
\onlinelocal runs of \(\algoA\).

The correctness of algorithm~$\algoB$ follows from that of $\algoA$.
Moreover, when labeling an isolated node, $\algoB$ always labels it in the same
way that depends only on the local structure and inputs of the graph; hence
\(\algoB\) is \(1\)-amnesiac.
Clearly, there is an exponential overhead in the locality: the locality of \(\algoB\) is \(T'(n) = T(N_1) = T(2^{O(n^2)})\).

Using the above, we now describe how to obtain a \(2\)-amnesiac algorithm, that
is, an algorithm that always produces the same labels for the same types of
connected components formed by the union of intersecting neighborhoods of two
distinct nodes.
We modify the previous experiment as follows:
\begin{enumerate}[noitemsep]
  \item First we must increase the size of the experiment graph.
  Let \(N_2 = \abs*{\outLabels}^2 \abs*{\inLabels}^n \cdot 2^n \cdot n^2 \) and
  redefine \(\GG_1\) with the radius-\(T(N_2)\) neighborhoods.
  \item Instead of considering \(\abs*{\outLabels} \cdot n\) many disjoint
  copies of elements of \(\GG_1\), we now take \(\abs*{\outLabels}^2 \cdot n\)
  copies.
  By the same argument as before, there are at least \(n \cdot \abs*{\outLabels}\)
  good neighborhoods.
  \item Let \(\GG_2\) be the set of all possible unions of two non-disjoint
  radius-\(T(N_2)\) neighborhoods of two different nodes of any \(n\)-node graph
  (in \(\FF\)), with all possible input labelings and orderings of the center
  nodes.
  Notice that the size of \(\GG_2\) is \(g_2 \le 2^{n^2} \abs*{\inLabels}^n\).
  We take \(\abs*{\outLabels} \cdot n\) many disjoint copies of all graphs in
  \(\GG_2\), with the catch that the neighborhood of the center node that comes
  first in the processing order is chosen arbitrarily among the good
  neighborhoods.
  The resulting graph \(H_2\) is our new experiment graph, whose size is now
  \(N_2 \le \abs*{\outLabels}\abs*{\inLabels}^n \cdot 2^{n^2} \cdot n^3 \).
  \item Use \(\algoA\) to label all nodes that come first in the input order in
  each graph, then all nodes that come second in the same respective order.
  \item Use \(\algoA\) to label the second center node of each connected
  component.
  By the pigeonhole principle, for all types of such connected components
  \(\TT_2\) there are at least \(n\) identical labelings.
  These make out the canonical labelings \(\sigma_{\TT_2}\).
\end{enumerate}

The \(2\)-amnesiac algorithm \(\algoB\) starts by running the above experiment.
When given as input an \(n\)-node graph, it outputs the canonical labeling
\(\sigma_{\TT_1}\) for isolated nodes whose radius-\(T'(n) = T(N_2)\)
neighborhood matches type \(\TT_1\).
Similarly, $\algoB$ outputs the canonical labeling \(\sigma_{\TT_2}\) for nodes
seeing a connected component of type \(\TT_2\) when we look at their
radius-\(T'(n) = T(N_2)\) neighborhood.
The correctness argument is the same as that for $1$-amnesiac algorithms. 
For the remaining nodes, \(\algoB\) just simulates \(\algoA\) using global
memory as usual.

We can continue this process all the way up to \(n\)-amnesiac algorithms, which
are simply component-wise \onlinelocal algorithms.
See \cref{lemma:olocal-amnesiac} for the formal details. 
As a remark, notice the restriction to LCLs is necessary to prove correctness:
if the amnesiac algorithm fails, it must fail locally; hence also the original
\onlinelocal algorithm fails locally, contradicting its correctness.

\paragraph{Dealing with randomness.}

Let us now turn to the setting where \(\algoA\) is randomized.
For deterministic \onlinelocal algorithms, we adaptively picked the good
neighborhoods before processing further; however, since in randomized
\onlinelocal the adversary is \emph{oblivious}, we must adapt our strategy.

Our experiment graph \(H_n\) is now \emph{random}.
It is constructed exactly as before, though now the good connected components
have to be \enquote{guessed} uniformly at random each time. 
Clearly, the probability that our guesses are good is incredibly small.
Hence we would like to amplify the success probability so that the probability
that the guesses are good enough \emph{and} the randomized \onlinelocal
algorithm works correctly is bounded away from zero.
(Once we have this, we may apply a standard derandomization argument.)
The amplification is by replicating \(k\) times the experiment graph \(H_n\), guessing the good components for each copy of \(H_n\) independently.
Since the randomness of the algorithm and the randomness of these guesses are
independent for each copy of \(H_n\), a simple union bound together with the
inclusion-exclusion principle give a positive probability that, in at least one
copy of \(H_n\), the guesses are correct and the algorithm works properly.

We now apply a standard derandomization argument:
Since there exists an assignment of a random bit string to the randomized
\onlinelocal algorithm as well as for \enquote{guessing} the good components in
the \(k\) copies of the experiment graph, there is a fixed, deterministic
realization of \(H_n\) that yields a (correct) deterministic \onlinelocal
algorithm.
Hence our \(n\)-amnesiac algorithm \(\algoB\) goes over all possible
definitions of \(H_n\) and of deterministic \onlinelocal algorithms (according
to an arbitrary order) until it finds this good combination.
It must eventually succeed because, when the size of the input graph is fixed,
there are only finitely many combinations.
The proof then reduces to the previous case. 
See \cref{lemma:rand-olocal-amnesiac} for the formal details.

\paragraph{\boldmath Step 2: From component-wise algorithms to \slocal algorithms on rooted trees.}
We heavily exploit the fact that, in rooted trees, there is a consistent
orientation of the edges towards the root.
We adapt results from \cite[Section 7]{chang23automata-theoretic} to the \slocal
model to show how to \enquote{cluster} a rooted tree in connected components
(which are also rooted trees) in time \(O(\alpha)\) and so that the following
properties are met:
\begin{enumerate}[noitemsep]
    \item All connected components have depth \(\Theta(\alpha)\).
    \item All leaves of a connected component that are not at distance \(\Theta(\alpha)\) from the root of the component are real leaves of the original rooted tree.
    \item All other leaves are either real leaves or roots of other connected components. 
\end{enumerate}
The formal details can be found in \cref{lemma:online-to-slocal:clustering}.

Suppose now we are given an \(n\)-amnesiac (equivalently, component-wise) algorithm \(\algoA\) solving an LCL \(\problem\) on rooted trees with locality \(T(n)\).
We briefly describe how to construct an \slocal algorithm \(\algoB\) that solves \(\problem\) with locality \(O(T(n))\) on rooted trees.
Recall that in \slocal we may compose two algorithms with localities \(T_1\) and
\(T_2\) and obtain an \slocal algorithm with locality \(O(T_1 + T_2)\)
\cite{ghaffari2017}.
Algorithm $\algoB$ is the composition of the following four algorithms with
locality $O(T)$:
\begin{enumerate}[noitemsep]
  \item \(\algoB_1\) constructs a clustering with properties 1-3 above with
  \(\alpha = O(T)\), where the hidden constant is large enough. 
  \item \(\algoB_2\) ensures each root of a connected component of the cluster
  precommits a solution in the neighborhood of the leaves of the component using \(\algoA\).
  This can be done independently between different components if the localities
  are appropriately chosen.
  The root of the original tree also outputs a
  solution for itself.
  \item \(\algoB_3\) outputs the precommitments on the nodes designated by \(\algoB_2\).
  \item The fact that we used \(\algoA\) independently on disjoint components of
  the graph (and $\algoA$ being correct) ensures that a solution to the LCL
  exists and the \enquote{inner part} of all connected components can be
  completed with a correct solution.
  Hence \(\algoB_4\) simply brute-forces a solution inside the clusters.
\end{enumerate}
Since \(\problem\) is an LCL, different components will have compatible
solutions as \(\algoB_2\) precommitted a solution around leaf nodes, which are
roots of other connected components.

\paragraph{Step 3: \boldmath  From \slocal algorithms to \local algorithms on rooted trees.}

Any \(o(\log n)\)-time  \slocal algorithm \(\AA\) solving an LCL over rooted trees can be turned into an \(O(1)\)-time \slocal algorithm \(\BB\) achieving the same.
This is obtained by exploiting the same tree decomposition described above and providing fake, repeating identifiers to each cluster so that two equal identifiers are ``far enough''.
Then, we use \(\AA\) as a black-box and lie to \(\AA\) by providing the size of a cluster as the input graph size.
Once we have an \(O(1)\)-time \slocal algorithm \(\BB\), how to turn it into an \(O(\log* n)\)-time \local algorithm is folklore.

\section{Additional results}\label{sec:extra}

We now extend the discussion to additional results that are needed to piece
together missing parts of the big picture in \cref{fig:landscape}. 
These results also help provide further intuition and examples on our models
and their key properties.

\subsection[Randomized \onlinelocal with adaptive adversary]{\boldmath Randomized \onlinelocal with adaptive adversary (\cref{sec:adaptive-rolocal-derandomization})}
\label{sec:extra-rolcl-adaptive}

The new model that we introduced, randomized \onlinelocal, is defined using an \emph{oblivious} adversary. In \cref{sec:adaptive-rolocal-derandomization} we show that this is also necessary: if we defined randomized \onlinelocal with an \emph{adaptive} adversary, it would be as weak as the deterministic \onlinelocal model.

We show that an adaptive adversary in \rolcl is so strong that a succeeding \rolcl algorithm of locality \(T\) would admit a single random-bit string that outputs a good solution for all possible graphs of a given size; hence, a correct \detolcl algorithm exists.
Since \detolcl algorithms of locality \(T\) for graphs of \(n\) nodes are only finitely many, the \detolcl algorithm in its initialization phase can go over all of them until it finds the one working for all graphs of \(n\) nodes; it then uses that one.

\subsection[Lower bound on 3-coloring in randomized \onlinelocal]{\boldmath Lower bound on 3-coloring in randomized \onlinelocal (\cref{sec:3-col-grid-lb})}

So far we have connected randomized \onlinelocal with other models through
simulation arguments that only work in rooted trees. Let us now put limitations
on randomized \onlinelocal in a broader setting. Recall that in deterministic
\onlinelocal we can 3-color bipartite graphs with locality $O(\log n)$
\cite{akbari_et_al:LIPIcs.ICALP.2023.10}, and this is tight
\cite{chang23_tight_arxiv}. 
Extending the work of \citet{chang23_tight_arxiv}, in \cref{sec:3-col-grid-lb}
we show that randomness does not help with solving this problem:
\begin{framedthm}{}
    \label{thm:intro:3-col-grid-lb}
    $3$-coloring in bipartite graphs is not possible with locality $o(\log n)$ in the randomized \onlinelocal model.
\end{framedthm}
\noindent This demonstrates that, even though randomized \onlinelocal is a very
strong model---indeed strong enough to simulate, e.g., any non-signaling
distribution---, it is nevertheless possible to prove strong lower bounds in
this model (which then extend to lower bounds across the entire landscape of
models).

In the proof, we use the notion of a $b$-value defined in~\cite{chang23_tight_arxiv} as a measure of the number of incompatible boundaries present in a region of a grid. We start with the assumption that a grid can be $3$-colored with locality $o(\log n)$ and derive a contradiction. The high level idea is to construct two path segments below each other, where one path segment has a large count of incompatible boundaries (a high $b$-value) and the other segment has a low boundary count (incompatible $b$-value to the upper path). This forces an algorithm to make boundaries escape on the side between the two segments. We show that the boundary count is, however, too large compared to the distance between the two segments and thus the boundaries \enquote{cannot escape}.

Two difficulties arise in the randomized case compared to the deterministic
lower bound: (i)~in order to create a path with a large $b$-value we have to use
a probabilistic construction that produces a segment with a large $b$-value with
high probability; and (ii) since this construction is probabilistic but our
adversary oblivious (recalling our previous discussion in
\cref{sec:extra-rolcl-adaptive}), we cannot \enquote{see} the large $b$-value
segment constructed in (i), that is, we can neither predict its position nor its
size. 
We therefore need to use another probabilistic construction that positions the
segment in a position that forces a contradiction (and which succeeds with
constant probability).

\subsection[Randomized \onlinelocal in paths and cycles]{\boldmath Randomized \onlinelocal in paths and cycles (\cref{sec:cycles})}

Our final technical part shows that LCL problems in paths and cycles have complexity either \(\myO{1}\) or \(\myTheta{n}\) in the \rolcl model; moreover the locality is \(\myO{1}\) in \rolcl if and only if it is \(\myO{\log* n}\) in the deterministic \local model. Together with prior work, this also shows that locality of an LCL problem in paths and cycles is \emph{decidable} across all models \cite{balliu19lcl-decidability} (with the caveat that we cannot distinguish between $O(1)$ and $\Theta(\log^* n)$ for quantum-\local).
The proof is a reworking of its deterministic variant from \cite{akbari_et_al:LIPIcs.ICALP.2023.10}.
The main take-home message of this result is the following: cycles are not a fundamental obstacle for simulating \rolcl in weaker models. Hence, there is hope for generalizing the simulation result of \cref{sec:olocal-slocal-simulation} from rooted trees to a broader class of graphs.

\subsection[Quantum, bounded-dependence, and non-signaling]{\boldmath Quantum, bounded-dependence, and non-signaling (\cref{app:dc-models})}

Finally, the material in \cref{app:dc-models} is more conceptual in nature and aims at serving a dual purpose. First, it aims at formally introducing the non-signaling model based on the non-signaling principle, and at explaining \emph{why} it is more powerful than the quantum-\local model. This is not a new result, but included for completeness and to clarify the early works of \citet{gavoille2009} and \citet{arfaoui2014}.
Second, it formally introduces the bounded-dependence model based on finitely-dependent processes and argues why the relations in diagram \eqref{eq:diagram-quantum} hold, and in particular why quantum-\local without shared quantum state is contained not only in the non-signaling model but also in the bounded-dependence model. While all the ingredients are well-known, to our knowledge this relation between the quantum-\local model and the bounded-dependence model is not made explicit in the literature before. 

\subsection{Open questions}

Our work suggests a number of open questions; here are the most prominent ones:

\begin{question}[Solved by \cite{balliu2026}]\label{question:is-stronger}
    Is \qlocal stronger than randomized \local for any LCL problem? In particular, is there any LCL problem with constant locality in \qlocal but super-constant locality in \local? We conjecture that such a problem does not exist. Our work proves that to show this conjecture, new proof techniques are needed.
\end{question}

\begin{question}[Solved by \cite{balliu2026}]\label{question:is-stronger-natural}
    Is \qlocal stronger than randomized \local for any \emph{natural} LCL problem? By \emph{natural}, we mean problems that have been widely studied in the distributed computing community, such as graph coloring, maximal independent set, maximal matching, sinkless orientation, etc.
\end{question}

\begin{question}
    Does shared global information or shared quantum state ever help with any LCL problem (beyond the fact that shared quantum state can be used to generate shared randomness, which is known to help \cite{balliu2024shared-randomness})?
\end{question}

\begin{question}
    Is it possible to simulate deterministic or randomized online-\local in \slocal and \local also in a broader graph class than rooted trees?
    If we could extend the result to \emph{unrooted} trees, it would also imply new lower bounds for the widely-studied \emph{sinkless orientation} problem \cite{brandt16lll,balliu22so-simple} across all models.
\end{question}

\subsection{Follow-up works}

Since the first versions of this work appeared on arXiv, it has already influenced several follow-up papers:

\begin{enumerate}[noitemsep]
    \item \citet{dhar24rand-online-local} study the relationship between randomized online-\local and deterministic \local in much greater depth, in particular outside the \(o(\log \log \log n)\) regime. They show, for example, that on \emph{regular rooted trees} the complexity classes of LCL problems in randomized online-\local and deterministic \local coincide exactly. As a corollary, for this graph family all models in \cref{fig:landscape} have the same power in the \(\Omega(\log^* n)\) regime.
    \item \citet{balliu2024shared-randomness} show that shared randomness can strictly help for certain LCL problems. This yields several new separations between the models in \cref{fig:landscape}. For instance, quantum-\local with a shared quantum state is strictly stronger than quantum-\local without shared quantum state; the non-signaling model is strictly stronger than the bounded-dependence model; and randomized online-\local is strictly stronger than randomized \slocal.
    \item \citet{balliu2024quantum} exhibit the first family of LCL problems \(\{\Pi_\Delta\}\), parameterized by the maximum input degree \(\Delta\), such that \(\Pi_\Delta\) can be solved with constant locality in quantum-\local but requires locality \(\Omega(\Delta)\) in randomized \local. This does not yet answer \cref{question:is-stronger}, since an asymptotic separation would require \(\Delta\) to grow with the size of the input graph, which is not allowed by the definition of LCL problems.
    \item \citet{balliu2026} give a positive answer to \cref{question:is-stronger}. By modifying the problem \(\Pi_\Delta\) of \cite{balliu2024quantum} for \(\Delta = \Theta(\log n / \log \log n)\), they construct a proper LCL problem that can be solved with locality \(O(\log n)\) in quantum-\local, but requires locality \(\Omega(\log n \log \log n / \log \log \log n)\) in randomized \local.
    We have to mention that the LCL problem constructed in \cite{balliu2026} is artificial and has specifically been designed to separate quantum-\local from randomized \local. \cref{question:is-stronger-natural} remains open, that is, determining whether more natural LCL problems, e.g., graph coloring or maximal independent set, admit similar separations.
    \item \citet{balliu2025limits} prove that the non-signaling model provides no advantage for LP problems w.r.t.\ the deterministic \local model. In particular, this implies that both the quantum-\local model and the bounded-dependence model are no stronger than deterministic \local for LP problems. They also exhibit the first LCL problem for which quantum-\local is strictly weaker than deterministic or randomized \slocal. Together with the result of \cite{balliu2024shared-randomness}, this implies that the quantum-\local and \slocal models are incomparable.
    \item \citet{boudier25_orientation_opodis} show that the lower bound for
    $3$-coloring of \cref{thm:intro:3-col-grid-lb} extends to the case where the
    algorithm is even explicitly given a globally consistent orientation of the
    grid.
    \item Recently, a survey on progress in distributed quantum computing has been published \cite{damore2025eatcs}.
\end{enumerate}

Taken together, these works demonstrate two key points.
First, they show that in restricted graph classes---such as rooted trees---it is often possible to capture the entire hierarchy of models with a single theorem about randomized online-\local, whereas in general the landscape consists of many genuinely distinct models.
Second, they show that research in this area is highly active, with new and surprising results emerging rapidly, while at the same time many fundamental open questions remain.

\subsection{Differences with the conference version}

In the conference version of this paper \cite{akbari2025stoc}, among the
additional results presented in \cref{sec:extra}, we had a derandomization
result for \local in the \dlocal model (similar to how \slocal can derandomize
\local \cite{dist_deran}).
However, the proof of that statement contained a flaw that we were not able to
remedy, and hence we removed the theorem from this version.
Nevertheless, it is true that \dlocal can derandomize \local in certain graph classes, such as rooted trees, in the \(O(\log n)\) complexity region; this follows from \cref{thm:olocal-sim:reduction-slocal,thm:olocal-slocal-simulation}.
Whether \dlocal can derandomize \local in general remains an interesting open
question.   %
\section{Preliminaries}
\label{sec:preliminaries}

We use the notation \(\nat = \{0,1,2,\dotsc\}\) and \(\nat_+ = \nat \setminus \{0\}\).
For any positive integer \(n \in \natPos\), we denote the set \(\{1, \dots, n\}\) by \([n]\).

\paragraph{Graphs.}
In this work, a graph \(G = (V,E)\) can be either directed (\(E \subseteq V^2\)) or undirected (\(E \subseteq {\binom{V}{2}}\)).
If the set of nodes and the set of edges are not specified, we refer to them by \(V(G)\) and \(E(G)\), respectively.
For any edge \(e = (u,v) \in E(G)\) we say that \(e\) is directed from \(u\) to \(v\) and is incident to \(u\) and to \(v\) (the latter holds also for undirected edges).
All graphs in this paper are simple graphs without self-loops unless otherwise specified.
The degree of a node \(v\) is the number of edges that are incident to \(v\) and is denoted by \(\deg_G(v)\), or simply by \(\deg(v)\) when \(G\) is clear from the context.
The indegree of a node \(v\) is the number of directed edges that are directed towards \(v\) and is denoted by \(\inDeg_G(v)\), while the outdegree of \(v\) is the number of directed edges that are incident to \(v\) but directed to some other vertex  and is denoted by \(\outDeg_G(v)\).
Again, we omit the suffix \(G\) if the graph is clear from the context.

If \(G\) is a subgraph of \(H\), we write \(G \subseteq H\).
For any subset of nodes \(A \subseteq V\), we denote by \(G[A]\) the subgraph
induced by the nodes in \(A\).
For any nodes $u,v \in V$, $\dist_G(u,v)$ denotes the distance between $u$ and
$v$ in $G$ (i.e., the number of edges of any shortest path between $u$ and $v$
in $G$---it doesn't need to be a directed path); if $u$ and $v$ are disconnected, then $\dist_G(u,v) = +\infty$.
If $G$ is clear from the context, we may also simply write $\dist(u,v) =
\dist_G(u,v)$.
Also, for a subset of nodes \(A,B \subseteq V\) and any node \(v \in V\), we can define \(
\dist_G(v,A) = \min_{u \in A} \dist_G(u,v) \) and, by extension, \(
  \dist_G(A,B) = \min_{u \in A} \dist_G(u,B)\).
We assume that \(\dist_G(A, \emptyset) = +\infty\).
Similarly, for subgraphs \(G_1, G_2 \subseteq G\), we define \(\dist_G(G_1,G_2) = \dist_G(V(G_1),V(G_2))\): here, we also assume
\(\dist_G(G_1,\emptyset) = +\infty \) where \(\emptyset\) is now the empty graph.
For \(T \in [n]\), the \(T\)-neighborhood of a node \(u \in V\) of a graph \(G\) is the set
\(\NN_T(u,G) = \left\{v \in V \ \st \ \dist_G(u,v) \le T \right\}\).
The \(T\)-neighborhood of a subset \(A \subseteq V\) is the set \(\NN_T(A,G) =
\left\{v \in V \ \st \ \exists u \in A : \dist_G(u,v) \le T \right\}\).
Similarly, the \(T\)-neighborhood of a subgraph \(H \subseteq G\) is the set
\(\NN_T(H,G) = \left\{v \in V \ \st \ \exists u \in V(H) : \dist_G(u,v) \le T
\right\}\).
If \(G\) is clear from the context, we just write \(\NN_T(u)\), \(\NN_T(A)\), and \(\NN_T(H)\).
If \(u \notin V\), \(A \cap V = \emptyset\), or \(V(H) \cap V(G) = \emptyset\), then the neighborhood \(\NN_T(u) = \emptyset\), \(\NN_T(A) = \emptyset\), or \(\NN_T(H) = \emptyset\), respectively.
We make use of some graph operations:
For any two graphs \(G,H\), we denote by \(G \cap H\) the intersection graph defined by \(G \cap H = (V(G) \cap V(H), E(G) \cap E(H))\).
The graph union is defined by \(G \cup H = (V(G) \cup V(H), E(G) \cup E(H))\).
Moreover, the graph difference is the graph \(G \setminus H = (V(G) \setminus V(H), E(G) \setminus E(H))\).

Finally, for any two graphs \(G\) and \(H\), we write \(G \sim_{f} H\) to denote that \(G\) and \(H\) are isomorphic and \(f:V(G) \to V(H)\) is an isomorphism.

\paragraph{Labeling problems.}

We start with the notion of labeling problem.

\begin{definition}[Labeling problem]\label{def:labeling-problem}
    Let \(\inLabels\) and \(\outLabels\) be two sets of input and output labels, respectively.
    A \emph{labeling problem} \(\problem\) is a mapping \((G,\inLabel) \mapsto \{\indOutLabel{i}\}_{i \in I}\), with \(I\) being a discrete set of indexes, that assigns to every graph \(G\) with any input labeling \(\inLabel: V(G) \to \inLabels\) a set of permissible output vectors \(\indOutLabel{i} : V(G) \to \outLabels\) that might depend on \((G, \inLabel)\).
    The mapping must be closed under graph isomorphism, i.e., if \(\varphi: V(G) \to V(G')\) is an isomorphism between \(G\) and \(G'\), and \(\indOutLabel{i} \in \problem((G', \inLabel)) \), then \(\indOutLabel{i} \circ \varphi \in \problem((G, \inLabel \circ \varphi))  \).
\end{definition}

A labeling problem can be thought of as defined for \emph{any} input graph of \emph{any} number of nodes.
If the set of permissible output vectors is empty for some input \((G,\inLabel)\), we say that the problem is not solvable on the input \((G,\inLabel)\):
accordingly, the problem is solvable on the input \((G,\inLabel)\) if \(\problem(G,\inLabel) \neq \emptyset\).

One observation on the generality of definition of labeling problem follows:
one can actually consider problems that require to output labels on edges.

We actually focus on labeling problems where, for any input graph, an output vector \(\outLabel\) is permissible if and only if the restrictions of the problem on any local neighborhoods can be solved and there exist compatible local permissible output vectors whose combination provides \(\outLabel\).
This concept is grasped by the notion of locally checkable labeling (LCL) problems, first introduced by \citet{naor1995}.
For any function \(f : A \to B\) and any subset \(A' \subseteq A\), let us denote the restriction of \(f\) to \(A'\) by \(f \restriction_{A'}\).
Furthermore, we define a centered graph to be a pair \((H,v_H)\) where \(H\) is a graph and \(v_H \in V(H)\) is a vertex of \(H\) that we name the \emph{center} of \(H\).
The  \emph{radius} of a centered graph is the maximum distance from \(v_H\) to any other node in \(H\).

\begin{definition}[Locally checkable labeling problem]\label{def:lvl-problem}
    Let \(\lvlCR, \maxDeg \in \nat\).
    Let  \(\inLabels\) and \(\outLabels\) be two finite sets of input and output labels, respectively, and \(\problem\) a labeling problem.
    \(\problem\) is \emph{locally checkable} with checking radius \(\lvlCR\) if there exists a family \( \SS = \{((H,v_H), \inLblRes, \outLblRes)_i\}_{i \in I}\) of tuples, where \((H,v_H)\) is a centered graph of radius at most \(\lvlCR\) and maximum degree at most \(\maxDeg\), \(\inLblRes : V(H) \to \inLabels\) is an input labeling for \(H\), \(\outLblRes : V(H) \to \outLabels\) is an output labeling for \(H\) (which can depend on \(\inLblRes\)) with the following property
    \begin{itemize}
        \item for any input \((G, \inLabel)\) to \(\problem\) with \(\deg(G) \le \maxDeg\), an output vector \(\outLabel : V(G) \to \outLabels\) is permissible (i.e., \(\outLabel \in \problem((G,\inLabel))\)) if and only if, for each node \(v \in V(G)\), the tuple \(((G[\NN_\lvlCR(v)]), \inLabel \restriction_{\NN_{\lvlCR}(v)}, \outLabel \restriction_{\NN_{\lvlCR}(v)} )\) belongs to \(\SS\) (up to graph isomorphisms).
    \end{itemize}
\end{definition}
Notice that the family \(\SS\) can be always thought to be finite up to graph isomorphisms, as \(r\) and \(\maxDeg\) are fixed and the set of input/output labels are finite.
We now define the computational models we work in.

\paragraph{The \portnum model.}
A port-numbered network is a triple \(N = (V, P, p)\) where \(V\) is the set of nodes, \(P\) is the set of \emph{ports}, and \(p: P \to P\) is a function specifying connections between ports.
Each element \(x\in P\) is a pair \((v,i)\) where \(v \in V\), \(i \in \nat_+\).
The connection function \(p\) between ports is an involution, that is, \(p(p(x)) = x\) for all \(x \in P\).
If \((v,i) \in P\), we say that \((v,i)\) is port number \(i\) in node \(v\).
The degree of a node \(v\) in the network \(N\) is \(\deg_N(v)\), the number of ports in \(v\), that is, \(\deg_N(v) = \abs{\{i \in \nat: (v,i) \in P\}}\).
Unless otherwise mentioned, we assume that port numbers are consecutive, i.e., the ports of any node \(v \in V \) are \((v,1), \dots, (v, \deg_N(v))\).
Clearly, a port-numbered network identifies an \emph{underlying graph} \(G = (V,E)\) where, for any two nodes \(u,v \in V\), \(\{u,v\} \in E\) if and only if there exists ports \(x_u,x_v \in P\) such that \(p(x_u) = x_v\).
Clearly, the degree of a node \(\deg_N(v)\) corresponds to \(\deg_G(v)\).

In the \portnum model we are given a distributed system consisting of a port-numbered network of \( \abs{V} =
n\) \emph{processors} (or \emph{nodes}) that operates in a sequence of
synchronous rounds.
In each round the processors may perform unbounded computations on their
respective local state variables and subsequently exchange messages of
arbitrary size along the links given by the underlying input graph.
Nodes identify their neighbors by using ports as defined before, where port assignment may be done adversarially.
Barring their degree, all nodes are identical and operate according to the
same local computation procedures.
Initially all local state variables have the same value for all processors;
the sole exception is a distinguished local variable \(\inptData(v)\) of
each processor \(v\) that encodes input data.

Let \(\inLabels\) be a set of input labels.
The input of a problem is defined in the form of a labeled graph \((G,
\inptData)\) where \(G = (V, E)\) is the system graph, \(V\) is the set of
processors (hence it is specified as part of the input), and \(\inptData\colon V
\to \inLabels\) is an assignment of an input label \(\inLabel(v) \in \inLabels\) to
each processor \(v\).
The output of the algorithm is given in the form of a vector of local output
labels \(\outLabel\colon V \to \outLabels\), and the algorithm is assumed to
terminate once all labels \(\outLabel(v)\) are definitely fixed.
We assume that nodes and their links are fault-free.
The local computation procedures may be randomized by giving each processor
access to its own set of random variables; in this case, we are in the
\emph{randomized} \portnum model as opposed to the \emph{deterministic}
\portnum model.

The running time of an algorithm is the number of synchronous rounds required by all nodes to produce output labels.
If an algorithm's running time is \(T\), we also say that the algorithm has locality \(T\).
Notice that \(T\) can be a function of the size of the input graph. 
We say that a problem \(\problem\) over some graph family \(\FF\) has complexity \(T\) in the \portnum model if there is a \portnum algorithm running in time \(T\) that solves \(\problem\) over \(\FF\), and \(T=T(n)\) is the minimum running time (among all possible algorithms that solve \(\problem\) over \(\FF\)) in the worst case instance of size \(n\).
If the algorithm is randomized, we also require that the failure probability is at most \(1/n\), where \(n\) is the size of the input graph.

We remark that the notion of an (LCL) problem is a graph problem, and does not depend on the specific model of computation we consider (hence, the problem cannot depend on, e.g., port numbers).

\paragraph{\boldmath The \local model.}

The \local model was first introduced by \citet{linial87}: it is just the \portnum model augmented with an assignment of unique identifiers to nodes.
Let \(c \ge 1\) be a constant, and let \(\inLabels\) be a set of input labels.
The input of a problem is defined in the form of a labeled graph \((G,
\inptData)\) where \(G = (V, E)\) is the system graph, \(V\) is the set of
processors (hence it is specified as part of the input), and \(\inptData\colon V
\to [n^c] \times \inLabels\) is an assignment of a \emph{unique} identifier
\(\id(v) \in [n^c]\) and of an input label \(\inLabel(v) \in \inLabels\) to
each processor \(v\).
The output of the algorithm is given in the form of a vector of local output
labels \(\outLabel\colon V \to \outLabels\), and the algorithm is assumed to
terminate once all labels \(\outLabel(v)\) are definitely fixed.
We assume that nodes and their links are fault-free.
The local computation procedures may be randomized by giving each processor
access to its own set of random variables; in this case, we are in the
\emph{randomized} \local (\randlcl) model as opposed to \emph{deterministic}
\local (\detlcl).
Notice that the knowledge of \(n\) makes the randomized \portnum model roughly equivalent to the randomized \local model, as unique identifiers can be produced with high probability.
We say that a problem \(\problem\) over some graph family \(\FF\) has complexity \(T\) in the \local model if there is a \local algorithm running in time \(T\) that solves \(\problem\) over \(\FF\), and \(T=T(n)\) is the minimum running time (among all possible algorithms that solve \(\problem\) over \(\FF\)) in the worst case instance of size \(n\).
If the algorithm is randomized, we also require that the failure probability is at most \(1/n\), where \(n\) is the size of the input graph.

\paragraph{\boldmath The sequential \local model.}
The sequential \local model was first introduced by \cite{ghaffari2017}: it is a sequential version of the \local model.
Nodes are processed according to an adversarial order \(\sigma = v_1, \dots, v_n\).
When processing a node \(v_i\), a \(T\)-round algorithm collects all inputs in the radius-\(T\) neighborhood of \(v_i\) (including the states and the outputs of previously processed nodes in \(\NN_T(v_i)\), i.e., \(v_j \in \NN_T(v_i)\) for \(j < i\)):
we say that such an algorithm has complexity \(T\).
Note that the algorithm might store all inputs in \(\NN_T(v_i)\) in the state of \(v_i\): hence, when processing \(v_i\), it can see the input of \(v_j\), \(j < i\), if and only if there is a subsequence of nodes \(\{v_{h_\ell}\}_{\ell \in [m]}\) with \(j = h_1 < h_2 < \dots < h_m = i\) such that \(v_{h_\ell} \in \NN_T(v_{h_{\ell+1}})\) for all \(\ell \in [m-1]\).

If the algorithm is given an infinite random bit string, we talk about the randomized \slocal model, as opposed to the deterministic \slocal model.
We assume that the adversarial order according to which nodes are processed is oblivious to the random bit string, as in the original definition of the model.
We say that a problem \(\problem\) over some graph family \(\FF\) has complexity \(T\) in the \slocal model if there is an \slocal algorithm running in time \(T\) that solves \(\problem\) over \(\FF\), and \(T=T(n)\) is the minimum running time (among all possible algorithms that solve \(\problem\) over \(\FF\)) in the worst case instance of size \(n\).
If the algorithm is randomized, we also require that the failure probability is at most \(1/n\), where \(n\) is the size of the input graph.

\paragraph{\boldmath The \dlocal model.}
The deterministic \dlocal model was introduced in \cite{akbari_et_al:LIPIcs.ICALP.2023.10}. It is a centralized model of computing where the adversary constructs the graph one edge at a time.
An adversary constructs the input graph by adding one edge at a time (with an ordering of the nodes). 
The algorithm has a global view of the current state of the graph and has to commit for the newly added nodes (according to the ordering), but it has to maintain a feasible solution after each update. 
The algorithm is restricted so that after a modification at node $v$, it can only update the solution within distance $T(n)$ from $v$.
We say that a problem \(\problem\) over some graph family \(\FF\) has complexity \(T(n)\) in the \dlocal model if \(T(n)\) is the minimum function such that there is a \dlocal algorithm running in time \(T(n)\) that solves \(\problem\) over \(\FF\), and \(T=T(n)\) is the minimum running time (among all possible algorithms that solve \(\problem\) over \(\FF\)) in the worst case instance of size \(n\).

\paragraph{\boldmath The \onlinelocal model.}
The (deterministic) \onlinelocal model was introduced in~\cite{akbari_et_al:LIPIcs.ICALP.2023.10}. It is basically equivalent to the \slocal model with global memory. 
More specifically, the \onlinelocal model is a centralized model of computing where the algorithm initially knows only the set of nodes of the input graph $G$. The nodes are processed with respect to an adversarial input sequence $\sigma = v_{1}, v_{2}, \dots, v_{n}$.
The output of $v_{i}$ depends on $G_{i} = G[\bigcup_{j=1}^{i} \NN_T(v_{j})]$, i.e., the subgraph induced by the radius-$T$ neighborhoods of $v_{1}, v_{2}, \dots, v_{i}$ (including all input data), plus all the outputs of previously processed nodes (in order).

We define the \rolcl model as a randomized variant of the \onlinelocal model where the label assigned by the algorithm to $v_{i}$ is a random outcome. Note that this model is oblivious to the randomness used by the algorithm. In particular this means that the graph $G \setminus G_{i}$ cannot be changed depending on the label assigned to $v_{i}$. 
One could also define the \rolcl model in an adaptive manner, but it turns out that this is equivalent to the deterministic \onlinelocal model as we show in \cref{sec:adaptive-rolocal-derandomization}.
We say that a problem \(\problem\) over some graph family \(\FF\) has complexity \(T\) in the \olcl (\rolcl) model if there is an \olcl (\rolcl) algorithm running in time \(T\) that solves \(\problem\) over \(\FF\), and \(T=T(n)\) is the minimum running time (among all possible algorithms that solve \(\problem\) over \(\FF\)) in the worst case instance of size \(n\).
If the algorithm is randomized, we also require that the failure probability is at most \(1/n\), where \(n\) is the size of the input graph.

\section{\boldmath Simulating non-signaling in randomized \onlinelocal}\label{sec:non-signaling-to-rolocal}

\subsection{Framework}\label{ssec:philcl-framework}
In this section we give the necessary framework to define the non-signaling model, and it is largely inspired by the definitions given in \cite{gavoille2009,coiteuxroy2023}. The next definition introduces the concept of outcome.

\begin{definition}[Outcome]\label{def:sim-phi-lcl:outcome}
    Let \(\inLabels\) and \(\outLabels\) be two sets of input and output labels, respectively, and let \(\FF\) be a family of graphs.
    An \emph{outcome} \(\outcome\) over \(\FF\) is a mapping \((G,\inptData) \mapsto \{(\indOutLabel{i}, \outPr_i)\}_{i \in I}\), with \(I\) being a discrete set of indexes, assigning to every input graph \(G \in \FF\) with any input data \(\inptData = (\id:V(G) \to [\abs{V(G)}^c], \inLabel:V(G) \to \inLabels)\), a discrete probability distribution \(\{\outPr_i\}_{i \in I}\) over output vectors \(\indOutLabel{i} : V(G) \to \outLabels \) such that:
    \begin{enumerate}[noitemsep]
        \item for all \(i \in I\), \(\outPr_i > 0\);
        \item \(\sum_{i \in I} \outPr_i = 1 \);
        \item \(\outPr_i\) represents the probability of obtaining \(\indOutLabel{i}\) as the output vector of the distributed system.
    \end{enumerate}
\end{definition}

We say that an outcome \(\outcome\) over some graph family \(\FF\) \emph{solves} problem \(\problem\) over \(\FF\) \emph{with probability \(p\)} if, 
for every \(G \in \FF\) and any input data \(\inptData = (\id, \inLabel)\), it holds that
\[
    \sum_{\substack{(\indOutLabel{i}, \outPr_i) \in \outcome((G,\inptData)) \ : \\ \indOutLabel{i} \in \problem((G,\inLabel))}} \outPr_i \ge p.  
\]
When \(p = 1\), we will just say that \(\outcome\) \emph{solves} problem \(\problem\) over the graph family \(\FF\).

The next computational model tries to capture the fundamental properties of any
\emph{physical} computational model (in which one can run either deterministic,
random, or quantum algorithms) that respects causality. 
The defining property of such a model is that, for any two (labeled) graphs
$(G_1,\inptData_{1})$ and $(G_2,\inptData_{2})$ that share some identical subgraph $(H,y)$, every
node $u$ in $H$ must exhibit identical behavior in $G_1$ and $G_2$ as long as
its \emph{local view}, that is, the set of nodes up to distance $T$ away from
$u$ together with input data and port numbering, is fully contained in $H$.
As the port numbering can be computed with one round of communication through a fixed procedure (e.g., assigning port numbers \(1,2, \dots, \deg(v)\) based on neighbor identifiers in ascending order) and we care about asymptotic bounds, we will omit port numbering from the definition of local view.

The model we consider has been introduced by \cite{gavoille2009}.
In order to proceed, we first define the \emph{non-signaling} property of an outcome.
Let \(T \ge 0\) be an integer, and \(I\) a set of indices.
For any set of nodes \(V\), subset \(S \subseteq V\), and for any input \((G = (V,E), \inptData)\), we define its \emph{\(T\)-local view} as the set 
\[
    \vvv_T(G, \inptData, S) = \left\{(u, \inptData(u)) \ \st \ \exists \ u \in V, \ v \in S \text{ such that } \dist_G(u,v) \le T \right\},
\]
where \(\dist_G(u,v)\) is the distance in \(G\).
Furthermore, for any subset of nodes \(S \subseteq V\) and any output distribution \(\{\indOutLabel{i}, \outPr_i\}_{i \in I}\), we define the \emph{marginal distribution} of \(\{\indOutLabel{i}, \outPr_i\}_{i \in I}\) on set \(S\) as the unique output distribution defined as follows:
for any \(\outLabel : V(G) \to \outLabels\), the probability of \(\outLabel \restriction_{S}\) on \(S\) is given by
\[
  \outPr(\outLabel, S) = \sum_{\substack{i \in I \ :\ \outLabel \restriction _S = \indOutLabel{i} \restriction _ S }} p_i,  
\]
where \(\outLabel \restriction_{S}\) and \(\indOutLabel{i} \restriction _ S \) are the restrictions of \(\outLabel\) and \(\indOutLabel{i}\) on \(S\), respectively. 

\begin{definition}[Non-signaling outcome]\label{def:ns-outcome}
    Let \(\FF\) be a family of graphs.
    An outcome \(\outcome: (G, \inptData) \mapsto \{(\indOutLabel{i}, \outPr_i )\}_{i \in I}\) over \(\FF\) is \emph{non-signaling} beyond distance \(T = T(G, \inptData)\) if for 
    any pair of inputs \((G_1 = (V_1,E_1), \inptData_{1})\), \((G_2 = (V_2,E_2), \inptData_{2})\), with the same number of nodes, such that \(\vvv_{T(G_1, \inptData_1)} (G_1, \inptData_{1}, S)\) is isomorphic to \( \vvv_{T(G_1, \inptData_1)} (G_2, \inptData_{2}, S)\) and \(G_1,G_2 \in \FF\), the output distributions corresponding to these inputs have identical marginal distributions on the set \(S\).
\end{definition}
\cref{def:ns-outcome} is also the more general definition for the locality of an outcome: an outcome \(\outcome\) has locality \(T\) if it is non-signaling beyond distance \(T\).

\paragraph{\boldmath The \philcl model.} 

The \philcl model is a computational model that produces non-signaling outcomes over some family of graphs \(\FF\). 
Let \(p \in [0,1]\).
A problem \(\problem\) over some graph family \(\FF\) has complexity \(T\) (and success probability \(p\)) if there exists an outcome \(\outcome\) that is non-signaling beyond distance \(T\) which solves \(\problem\) over \(\FF\) (with probability at least \(p\)), and \(T=T(n)\) is the minimum ``non-signaling distance'' (among all possible outcomes that solve \(\problem\) over \(\FF\)) in the worst case instance of size \(n\).

As every (deterministic or randomized) algorithm running in time at most \(T\) in the \local model produces an outcome which has locality \(T\), we can provide lower bounds for the \local model by proving them in the \philcl model.

Notice that algorithms in the \local model can always be thought of as producing outputs for \emph{any} input graph: when the computation at any round is not defined for some node, we can make the node output some garbage label, say \(\perp\). 
If we require that also outcomes are defined for every possible graph, then we are restricting the power of \philcl because outcomes must be defined accordingly.
This gives rise to a slightly weaker model, the \nslclFull, which was considered in other works such as \cite{coiteuxroy2023} and is still stronger than any classical or quantum variation of the \local model.

\begin{theorem}
  Let \(\problem\) be any labeling problem over any family of graphs \(\FF\). 
  Let \( \outcome: (G, \inLabel) \mapsto \{(\indOutLabel{h}, \outPr_i)\}_{h \in H} \) be an outcome over \(\FF\) which solves \(\problem\) with failure probability at most \(\varepsilon\) and is non-signaling at distance greater than \(T\).
  There exists a \rolcl algorithm \(\algo\) with complexity \(T\) that solves \(\problem\) over \(\FF\) and has failure probability at most \(\varepsilon\).
\end{theorem}

\begin{proof}
We want to design a \rolcl algorithm \(\algo\) with complexity \(T\) that solves \(\problem\) over \(\FF\) and has failure probability at most \(\varepsilon\) that somehow simulates \(\outcome\). 
We need to assume that the number of nodes of the input graph is known by \(\algo\).

Fix any graph \( G \in \FF \) of \(n\) nodes and any input labeling \(\inLabel\), and fix any sequence of nodes \( v_1, \dots, v_n \) the adversary might choose to reveal to the algorithm. 

The adversary initially shows \( G[\NN_T(v_1)] \) to the algorithm (including
input labels and identifiers) and asks it to label \( v_1 \).
In order for the algorithm to choose an appropriate output, it can arbitrarily
choose a graph \(H_1 \in \FF\) with \(n\) nodes that contains a subgraph
isomorphic to \( G[\NN_T(v_1)] \) (again, including input labels and
identifiers). 
Notice that \(H_1\) necessarily exists since \(G\) itself is such a graph.
We stress that the arbitrariness in the choice of $H_1$ is thanks to the
non-signaling property, which ensures that the restriction of the output
distribution \( \outcome(H_1, \inLabel) \) over \({\NN_T(v_1)} \) does not
change if the topology of the graph outside \( G[\NN_T(v_1)] \) differs (i.e.,
no matter how the adversary chooses it).
Indeed, if $\outcome$ did not have the non-signaling property, then it would be
impossible to sample from it correctly since the marginal output distribution on
$G[\NN_T(v_1)]$ would depend on the topology of the graph beyond it, which the
algorithm still has no knowledge of.

For any \(\outLabel: V(G) \to \outLabels\), the probability that \(\algo\) labels \(v_1\) with \(\outLabel(v_1)\) is
\begin{align*}
  \outPr({\outLabel},v_1) = \sum_{k: \ \outLabel \restriction_{\{v_1\}} = \indOutLabel{k}\restriction_{\{v_1\}}} \outPr_k,
\end{align*}
where \(\{(\indOutLabel{k},\outPr_k)\}_{k \in K_1} = \outcome(H_1, \inLabel)\).

Let \(\Lambda_i\) be the random variable yielding the label assigned to \(v_i\) by \(\algo\).
In general, for any \(\outLabel : V(G) \to \outLabels\) let \(p(\outLabel, v_i)\) denote the probability that \(\Lambda_i = \outLabel(v_i)\). 
Assume now that \(G[\NN_T(v_{i+1})]\) is shown to the algorithm.
Conditional on \(\Lambda_1, \dots, \Lambda_i\), for any \(\outLabel : V(G) \to \outLabels\) such that \(\outLabel(v_j) = \Lambda_j\) for all \(j = 1, \dots, i\), the probability that \(\Lambda_{i+1} = \outLabel(v_{i+1})\) is
\begin{align*}
  \outPr({\outLabel},v_{i+1}) = \sum_{\substack{k: \ \outLabel \restriction_{\{v_{j}\}} = \indOutLabel{k}\restriction_{\{v_{j}\}} \\ \forall \ 1 \le j \le i+1}} \frac{\outPr_k}{ \prod_{1 \le l \le i}\outPr({\outLabel},v_l)} ,
\end{align*}
where \(\{(\indOutLabel{k},\outPr_k)\}_{k \in K_{i+1}} = \outcome(H_{i+1}, \inLabel)\) for any graph \(H_{i+1} \in \FF\) of \(n\) nodes that contains a subgraph isomorphic to \(G[\cup_{i=1}^{i+1} \NN_T(v_j)]\) (including input labels and identifiers).
As before, \(H_{i+1}\) may be chosen arbitrarily due to the non-signaling
property. 

Now consider \(\outcome(G, \inLabel) = \{(\indOutLabel{h}, p_h)\}_{h \in H}\) for the right input graph \(G\), and take any \[(\indOutLabel{h^\star}, p_{h^\star}) \in \outcome(G, \inLabel).\] 

The probability that the outcomes of \(\Lambda_1, \dots, \Lambda_n\) give exactly \(\indOutLabel{h^\star}\) is
\begin{align*}
  & \pr{\Lambda_1 = \indOutLabel{h^\star}(v_1), \dots, \Lambda_n = \indOutLabel{h^\star}(v_n)} \\
  = \ & \pr{\Lambda_n = \indOutLabel{h^\star}(v_n) \st \Lambda_1 = \indOutLabel{h^\star}(v_1), \dots, \Lambda_{n-1} = \indOutLabel{h^\star}(v_{n-1})}  \\ 
  \cdot \ & \pr{\Lambda_{n-1} = \indOutLabel{h^\star}(v_{n-1}) \st \Lambda_1 = \indOutLabel{h^\star}(v_1), \dots, \Lambda_{n-2} = \indOutLabel{h^\star}(v_{n-2})}  \\
  & \dotsc \\
  \cdot \ & \pr{\Lambda_1 = \indOutLabel{h^\star}(v_1)} \\
  = \ & \outPr({\indOutLabel{h^\star}},v_{1})  \cdot \dotsc \cdot \outPr({\indOutLabel{h^\star}},v_{n}) \\
  = \ & \prod_{1 \le l \le n-1}\outPr({\outLabel},v_l) \cdot \sum_{\substack{k_n: \ \indOutLabel{h^\star} \restriction_{\{v_{j}\}} = \indOutLabel{k_j}\restriction_{\{v_{j}\}} \\ \forall \ 1 \le j \le n}} \frac{\outPr_{k_n}}{\prod_{1 \le l \le n-1}\outPr({\outLabel},v_l)} \\
  = \ & \sum_{\substack{k_n: \ \indOutLabel{h^\star} \restriction_{\{v_{j}\}} = \indOutLabel{k_j}\restriction_{\{v_{j}\}} \\ \forall \ 1 \le j \le n}} {\outPr_{k_n}}.
\end{align*}
Notice that 
\[
  \sum_{\substack{k_n: \ \indOutLabel{h^\star} \restriction_{\{v_{j}\}} = \indOutLabel{k_j}\restriction_{\{v_{j}\}} \\ \forall \ 1 \le j \le n}} {\outPr_{k_n}} = \outPr_{h^\star}
\]
as \(H_n\) is necessarily isomorphic to \(G\).
By the hypothesis,
\[
  \sum_{h: \ \indOutLabel{h} \text{ is valid for } (G, \inLabel)} \outPr_h \ge 1 - \varepsilon,
\]
implying that \(\algo\) succeeds with probability at least \(1 - \varepsilon\).
\end{proof}

\section{Bounded-dependence model can break symmetry}\label{sec:fin-dep}

For any graph \(G = (V,E)\), a \emph{random process} (or \emph{distribution}) \emph{on the vertices} of \(G\) is a family of random variables \(\{X_v\}_{v\in V}\), indexed by \(V\), while a \emph{random process on the edges} of \(G\) is a family of random variables \(\{X_e\}_{e\in E}\) indexed by \(E\).
More generally, a random process on \(G\) is a family of random variables \(\{X_y\}_{y \in V \cup E}\) indexed by \(V \cup E\).
The variables of a random process live in the same probability space and take values in some label set \(\labels\).
In general, we will consider random processes over vertices of graphs unless otherwise specified.

We now introduce the notion of \(T\)-dependent distribution. 
To do so, we extend the definition of distance to edges as follows: 
For any two edges \(e = \{v_1,v_2\},e'=\{u_1,u_2\} \in E\), \(\dist_G(e,e') = \min_{i,j \in [2]} \dist_G(v_i,u_j)\). 
Similarly, the distance between any edge \(e = (v_1,v_2)\) and a vertex \(v\) is \(\dist_G(e,v) = \min_{i \in [2]} \dist_G(v_i,v)\).
The definition extends easily to subsets containing vertices and edges.

\begin{definition}[\(T\)-dependent distribution]
  Let \(T \in \nat\) be a natural number and \(G = (V,E)\) be any graph.
  A random process \(\{X_v\}_{v\in V}\) on the vertices of \(G\) is said to be a \(T\)-dependent distribution if, for all subsets \(S,S' \subseteq V\) such that \(\dist_G(S,S') > T\), the two processes \(\{X_v\}_{v \in S}\) and \(\{X_v\}_{v \in S'}\) are independent.
  Analogous definitions hold for random processes on the edges of a graph and for random processes on the whole graph.
\end{definition}

A way to define \(T\)-dependent distributions that uses the same notation as \cref{ssec:philcl-framework} is by describing the output probability of global labelings: Let \(I\) be a discrete set of indices, and \(G = (V,E)\) some graph.
A distribution \( \{(\lbl_{i}, \outPr_i)\}_{i\in I}\) over output labelings, where \(\lbl_{i} : V \to \labels\) is an output labeling, is \emph{\(T\)-dependent} if the following holds:
for all output labelings \(\lbl\) in \( \{(\lbl_{i}, \outPr_i)\}\), for every two subsets of nodes \(S_1,S_2 \subseteq V(G)\) such that \(\dist_G(S_1,S_2) > T\), we have that 
\[
  \outPr(\lbl, S_1 \cup S_2) = \outPr(\lbl, S_1) \cdot \outPr(\lbl, S_2).
\]
Notice that outcomes, as defined in \cref{def:sim-phi-lcl:outcome}, output a random process for every input. 
Often, we have a family of graphs of arbitrarily large size \(n\) on which a \(T(n)\)-dependent distribution is defined.\footnote{Note that the dependency \(T\) might depend on other graph parameters such as the maximum degree \(\maxDeg\), the chromatic number \(\chi\), etc.}

Now we define the hypothetical computational model that outputs \(T\)-dependent distributions on some input graph.

\paragraph{The \boundep.}

The \boundep is a computational model that, for a given family of graphs \(\FF\), produces an outcome (as defined in \cref{def:sim-phi-lcl:outcome}) over \(\FF\) that is non-signaling beyond distance \(T = T(G,\inptData)\) (as defined in \cref{def:ns-outcome}), where \(G \in \FF\) and \(\inptData\) represents the input, which, in turn, produces \(T(G,\inptData)\)-dependent distributions. 
The random processes produced by an outcome are said to be \emph{finitely-dependent} if \(T = 
\myO{1}\) for all graphs in \(\FF\) and all input data.
If an outcome with the aforementioned properties solves a problem \(\problem\) over a graph family \(\FF\) with probability at least \(p\), we say that the pair \((\problem, \FF)\) has complexity \(T\) (with success probability \(p\)), and \(T=T(n)\) is the minimum dependence (among all possible distributions solving \(\problem\) over \(\FF\)) in the worst case instance of size \(n\).
We remark that \(T\) is also called the \emph{locality} or the \emph{complexity} of the corresponding output distributions.

We are particularly interested in \(T\)-dependent distributions  that satisfy invariance properties: We say that a random process \(\{X_v\}_{v \in V}\) over vertices of a graph \(G=(V,E)\) is \emph{invariant under automorphisms} if, for all automorphisms \(f: V \to V\) of \(G=(V,E)\), the two processes \(\{X_v\}_{v \in V}\) and \(\{Y_v = X_{f(v)}\}_{v \in V}\) are equal in law.
The definition is easily extendable to random processes on edges and random processes on the whole graph.

A stronger requirement is the \emph{invariance under subgraph isomorphisms}: 
Suppose we have an outcome \(\outcome: (G,\inptData) \mapsto \{X_v\}_{v \in V(G)}\) that maps each input graph \(G\) from a family of graphs \(\FF\) and any input data \(\inptData\) to a \(T=T(G,\inptData)\)-dependent distribution over \(G\) from a family of random processes \(\RR\).
We say that the random processes over vertices in 
\(\RR\) are invariant under subgraph isomorphisms if, given any two graphs \(G_1, G_2 \in \FF\) of size \(n_1,n_2\) with associated process 
\(\{X^{(1)}_v\}_{v \in V(G_1)}\) and \(\{X^{(2)}_v\}_{v \in V(G_2)}\), 
and any two subgraphs \(H_1 \subseteq G_1, H_2 \subseteq G_2\) such that 
\(G_1[\NN_{T(n_1)}(H_1)]\) and \(G_2[\NN_{T(n_2)}(H_2)]\) are isomorphic (with the isomorphism that brings \(H_1\) into \(H_2\)), then 
\(\{X^{(1)}_v\}_{v \in V(H_1)}\) and 
\(\{X^{(2)}_v\}_{v \in V(H_2)}\) are equal in law.\footnote{We remark that, as opposed to \cref{def:ns-outcome}, the isomorphism must preserve the input but \emph{not} the node identifiers.}
Trivially, invariance under subgraph isomorphisms implies invariance under automorphisms and the non-signaling property.
This definition is again easily extendable to the case of families of random processes over edges or over graphs in general.

The baseline of our result is a finitely-dependent distribution provided by \cite{holroyd2016,holroyd2018} on paths and cycles.

\begin{theorem}[Finitely-dependent coloring of the integers and of cycles \cite{holroyd2016,holroyd2018}]\label{thm:fin-dep:baseline}
  Let \(G = (V,E)\) be a graph that is either a cycle with at least 2 nodes or has \(V = \intg\) and \(E = \{\{i, i+1\} : i \in \intg\}\).
  For \((k,q) \in \{(1,4),(2,3)\}\), there exists a \(k\)-dependent distribution \(\{X^{(G)}_v\}_{v\in V(G)}\) that gives a \(q\)-coloring of \(G\).
  Furthermore, such distributions can be chosen to meet the following properties:
  If \(\{X_i\}_{i \in [n]}\) is the \(k\)-dependent \(q\)-coloring of the \(n\)-cycle and \(\{Y_i\}_{i \in \intg}\) is the \(k\)-dependent \(q\)-coloring of the integers, then \(\{X_i\}_{i \in [n-k]}\) is equal in law to \(\{Y_i\}_{i \in [n-k]}\).
  If \(n=2\), then the finitely-dependent colorings on the \(2\)-cycle and on a path of 2 nodes are identical in distribution. 
\end{theorem}

It is immediate that the disjoint union of any number of paths and cycles (even countably many) admits such distributions as well.

\begin{corollary}\label{cor:fin-dep:baseline}
  Let \(\FF\) be the family of graphs formed by the disjoint union (possibly uncountably many) paths with countably many nodes and cycles of any finite length.
  For all \(G \in \FF\) and for \((k,q) \in \{(1,4),(2,3)\}\), there exists a \(k\)-dependent distribution \(\{X^{(G)}_v\}_{v\in V(G)}\) that gives a \(q\)-coloring of \(G\).
  Furthermore, such distributions can be chosen to meet all the following properties:
  \begin{enumerate}[noitemsep]
    \item On each connected component \(H \subseteq G\), \(\{X^{(G)}_v\}_{v\in V(H)}\) is  given by \cref{thm:fin-dep:baseline}.
    \item \(\{\{X^{(G)}_v\}_{v\in V(G)}\}_{G \in \FF}\) is invariant under subgraph isomorphisms.
  \end{enumerate}
\end{corollary}

We will use this result to provide ``fake local identifiers'' to the nodes of the graph to simulate an \(\myO{\log* n}\)-round LOCAL algorithm through a random process with constant dependency.
In order to do that, we introduce some results on the composition of \(T\)-dependent distributions.

Trivially, every \(T\)-round (deterministic or randomized) \PN algorithm defines a \(2T\)-dependent distribution over the input graph.
Furthermore, if the underlying \PN model has access to random bits, the distribution can be made invariant under subgraph isomorphisms (provided that, whenever a distribution over input labelings is given together with the input graph, such distribution is also invariant under subgraph isomorphisms). 
The composition of the \(T_2\)-dependent distribution obtained by a  \(T_2\)-round \PN algorithm and any \(T_1\)-dependent distribution yields a \((2T_2 + T_1)\)-dependent distribution.

\begin{lemma}\label{lemma:fin-dep:combination}
  Let \(\labels^{(1)}\) and \(\labels^{(2)}\) be two label sets with countably many labels.
  Let \(\FF\) be any family of graphs, and let \(\RR = \{\{X_v\}_{v \in V(G)} : G \in \FF\}\) be a family of \(T_1\)-dependent distributions taking values in \(\labels^{(1)}\), where \(T_1\) might depend on parameters of \(G\). 
  Consider any \(T_2\)-round \PN algorithm \(\AA\) that takes as an input \(G \in \FF\) labelled by \(\{X_v\}_{v\in V(G)}\): it defines another distribution \(\{Y_v\}_{v\in V(G)}\) taking values in some label set \(\labels^{(2)}\).
  Then, the following properties hold:
  \begin{enumerate}
    \item  \(\{Y_v\}_{v \in V(G)}\) is a \((2T_2 + T_1)\)-dependent distribution on \(G\).
    \item  If the processes in \(\RR\) are invariant under subgraph isomorphisms,
    \(T_2\) is constant, and \(\AA\) does not depend on the size of the input
    graph and permutes port numbers locally u.a.r.\ at round 0, then the
    processes in \(\{\{Y_v\}_{v \in V(G)} : G \in \FF\}\) are invariant under
    subgraph isomorphisms.
  \end{enumerate}
\end{lemma}
\begin{proof}
  We prove the claim 1 first.
  Fix any \(G = (V,E) \in \FF\), and the corresponding \(T_2\)-dependent distribution \(\{X_v\}_{v \in V} \in \RR\).
  Fix any two subsets \(S,S' \subseteq V\) such that \(\dist_G(S,S') > 2T_2 + T_1\).
  Consider any output labeling \(\lbl^{(2)} : V \to \labels^{(2)}\).
  Then,
  \begin{align}
    & \pr{\cap_{v \in S \cup S'} \{Y_v = \lbl^{(2)}(v)\}} \nonumber \\
    = \ &  
    \sum_{\lbl^{(1)}\, : \, V \to \labels^{(1)}} 
    \pr{\cap_{v \in S \cup S'} \{Y_v = \lbl^{(2)}(v)\} \st \cap_{v \in V} \{X_v = \lbl^{(1)}(v)\}} \nonumber \\
    & \cdot \pr{\cap_{v \in V} \{X_v = \lbl^{(1)}(v)\}} \nonumber \\
    = \ & \sum_{\lbl^{(1)}\, : \, V \to \labels^{(1)}} 
    \pr{\cap_{v \in S \cup S'} \{Y_v = \lbl^{(2)}(v)\} \st \cap_{v \in \NN_{T_2(S \cup S')}} \{X_v = \lbl^{(1)}(v)\}} \label{eq:fin-dep:indep-far-away} \\
    & \cdot \pr{\cap_{v \in V} \{X_v = \lbl^{(1)}(v)\}} \nonumber \\
    = \ & \sum_{\lbl^{(1)}\, : \, \NN_{T_2}(S \cup S') \to \labels^{(1)}} 
    \pr{\cap_{v \in S \cup S'} \{Y_v = \lbl^{(2)}(v)\} \st \cap_{v \in \NN_{T_2(S \cup S')}} \{X_v = \lbl^{(1)}(v)\}} \nonumber \\
    & \cdot \pr{\cap_{v \in \NN_{T_2}(S \cup S')} \{X_v = \lbl^{(1)}(v)\}}\nonumber\nonumber ,
  \end{align}
  where  \cref{eq:fin-dep:indep-far-away} holds because the output of \(\{Y_v\}_{v \in S \cup S'}\) is independent of \(\{X_v\}_{v \notin \NN_{T_2}(S \cup S')}\).
  Since \(\dist_G(S,S') > 2T_2\), 
  \begin{align*}
    & \sum_{\lbl^{(1)}\, : \, \NN_{T_2}(S \cup S') \to \labels^{(1)}} 
    \pr{\cap_{v \in S \cup S'} \{Y_v = \lbl^{(2)}(v)\} \st \cap_{v \in \NN_{T_2(S \cup S')}} \{X_v = \lbl^{(1)}(v)\}} \\
    & \cdot \pr{\cap_{v \in \NN_{T_2}(S \cup S')} \{X_v = \lbl^{(1)}(v)\}} \\
    = \ & \sum_{\lbl^{(1)}\, : \, \NN_{T_2}(S \cup S') \to \labels^{(1)}} 
    \pr{\cap_{v \in S} \{Y_v = \lbl^{(2)}(v)\} \st \cap_{v \in \NN_{T_2(S \cup S')}} \{X_v = \lbl^{(1)}(v)\}} \\
    & \cdot \pr{\cap_{v \in S'} \{Y_v = \lbl^{(2)}(v)\} \st \cap_{v \in \NN_{T_2(S \cup S')}} \{X_v = \lbl^{(1)}(v)\}} \\
    & \cdot \pr{\cap_{v \in \NN_{T_2}(S \cup S')} \{X_v = \lbl^{(1)}(v)\}} \\
    = \ & \sum_{\lbl^{(1)}\, : \, \NN_{T_2}(S \cup S') \to \labels^{(1)}} 
    \pr{\cap_{v \in S} \{Y_v = \lbl^{(2)}(v)\} \st \cap_{v \in \NN_{T_2(S)}} \{X_v = \lbl^{(1)}(v)\}} \\
    & \cdot \pr{\cap_{v \in S'} \{Y_v = \lbl^{(2)}(v)\} \st \cap_{v \in \NN_{T_2(S')}} \{X_v = \lbl^{(1)}(v)\}} \\
    & \cdot \pr{\cap_{v \in \NN_{T_2}(S \cup S')} \{X_v = \lbl^{(1)}(v)\}}.
  \end{align*}
  Now observe that \(\dist_G(\NN_{T_2}(S),\NN_{T_2}(S')) > T_1\).
  Since \(\{X_v\}_{v \in V}\) is a \(T_1\)-dependent distribution, it holds that
  \begin{align*}
    & \sum_{\lbl^{(1)}\, : \, \NN_{T_2}(S \cup S') \to \labels^{(1)}} 
    \pr{\cap_{v \in S} \{Y_v = \lbl^{(2)}(v)\} \st \cap_{v \in \NN_{T_2(S)}} \{X_v = \lbl^{(1)}(v)\}} \\
    & \cdot \pr{\cap_{v \in S'} \{Y_v = \lbl^{(2)}(v)\} \st \cap_{v \in \NN_{T_2(S')}} \{X_v = \lbl^{(1)}(v)\}} \\
    & \cdot \pr{\cap_{v \in \NN_{T_2}(S \cup S')} \{X_v = \lbl^{(1)}(v)\}} \\ 
    = \ & \sum_{\lbl^{(1)}\, : \, \NN_{T_2}(S \cup S') \to \labels^{(1)}} 
    \pr{\cap_{v \in S} \{Y_v = \lbl^{(2)}(v)\} \st \cap_{v \in \NN_{T_2(S)}} \{X_v = \lbl^{(1)}(v)\}} \\
    & \cdot \pr{\cap_{v \in S'} \{Y_v = \lbl^{(2)}(v)\} \st \cap_{v \in \NN_{T_2(S')}} \{X_v = \lbl^{(1)}(v)\}} \\
    & \cdot \pr{\cap_{v \in \NN_{T_2}(S)} \{X_v = \lbl^{(1)}(v)\}} \pr{\cap_{v \in \NN_{T_2}(S')} \{X_v = \lbl^{(1)}(v)\}} \\
    = \ &  \pr{\cap_{v \in S} \{Y_v = \lbl^{(2)}(v)\} } \cdot
    \pr{\cap_{v \in S'} \{Y_v = \lbl^{(2)}(v)\}}.
  \end{align*}

  We now prove claim 2.
  Given  any \(G,H \in \FF\) of size \(n_G\) and \(n_H\), respectively, consider any subgraph isomorphism \(f: \NN_{2T_2 + T_1(n_G)}(G',G) \to \NN_{2T_2 + T_1(n_H)}(H',H)\) for any two \(G' \subseteq G, H' \subseteq H\) such that \(f\) restricted to \(G'\) is an isomorphism to \(H'\).
  Let \(\{X^{(G)}_{v}\}_{v \in V(G)}\) and \(\{X^{(H)}_{v}\}_{v \in V(H)}\), and \(\{Y^{(G)}_{v}\}_{v \in V(G)}\) and \(\{Y^{(H)}_{v}\}_{v \in V(H)}\) be the corresponding distributions of interest before and after the combination with the \PN algorithm, respectively.
  Fix any subset of nodes \(U \subseteq V(G')\) and consider any family of labels \(\{\lbl_u\}_{u \in U}\) indexed by \(U\).
  In order to prove property 2, it is sufficient to show that 
  \[
    \pr{\cap_{u \in U} \{Y^{(G)}_u = \lbl_u\}} = \pr{\cap_{u \in U} \{Y^{(H)}_{f(u)} = \lbl_u\}}.
  \]
  Notice that \(\pr{\cap_{u \in U} \{Y_u = \lbl_u\}}\) depends solely on the graph \(G[\NN_{T_2}(U,G)] = \cup_{u\in U} G[\NN_{T_2}(u,G)]\), the distribution of the port numbers in \(G[\NN_{T_2}(U,G)]\), and the random process \(\{X_u\}_{u \in \NN_{T_2}(U)}\).
  By hypothesis, \(\{X^{(G)}_u\}_{u \in \NN_{T_2}(U,G)}\) is equal in law to \(\{X^{(H)}_{f(u)}\}_{u \in \NN_{T_2}(U,G)}\).
  Furthermore, the restriction of \(f\) to \(G[\NN_{T_2}(U,G)]\) defines an isomorphism from \(G[\NN_{T_2}(U,G)]\)  to \( H[\NN_{T_2}(f(U),H)]\), hence the distribution of the port numbers in \(G[\NN_{T_2}(U,G)]\) is the same as that in \(H[\NN_{T_2}(f(U),H)]\) because each node permutes the port numbers locally u.a.r.
  Thus, \(\pr{\cap_{u \in U} \{Y^{(G)}_u = \lbl_u\}}\) must be the same as \(\pr{\cap_{u \in U} \{Y^{(H)}_{f(u)} = \lbl_u\}}\).
\end{proof}

\(T\)-dependent distributions over different graphs can be combined to obtain a \(T\)-dependent distribution over the graph union.

\begin{lemma}\label{lemma:fin-dep:composition}
  Let \( \{X_v\}_{v\in V_1}\) and \(\{Y_v\}_{v \in V_2}\) be a \(T_1\)-dependent distribution over a graph \(G_1 = (V_1, E_1)\) taking values in \(\labels^{(1)}\) and a \(T_2\)-dependent distribution over a graph \(G_2 = (V_2,E_2)\) taking values in \(\labels^{(2)}\), respectively.
  Assume \(\{X_v\}_{v \in V_1}\) and \(\{Y_v\}_{v\in V_2}\) to be independent processes.
  Consider the graph \(H = (V = V_1 \cup V_2, E = E_1 \cup E_2)\) and a distribution \(\{Z_v\}_{v \in V}\) over \(H\) taking values in \(\labels = (\labels^{(1)} \cup \{0\})\times(\labels^{(2)} \cup \{0\})\) defined by
  \begin{align*}
    Z_v = 
    \begin{cases}
      (X_v, 0) & \text{ if } v \in V_1 \setminus V_2, \\ 
      (X_v, Y_v) & \text{ if } v \in V_1 \cap V_2, \\ 
      (0, Y_v) & \text{ if } v \in V_2 \setminus V_1.
    \end{cases}
  \end{align*}
  Then, \(\{ Z_v\}_{v \in V} \) is \(\max(T_1,T_2)\)-dependent.
\end{lemma}
\begin{proof}
  For any vector \(\vvv \in \Sigma_1 \times \dots \times \Sigma_n\), we write \(\vvv[i]\) to denote its \(i\)-th entry.
  Fix any output labeling \(\lbl: V \to \labels\) such that \(\lbl(v)[2] = 0\) for all \(v \in V_1 \setminus V_2\) and \(\lbl(v)[1] = 0\) for all \(v \in V_2 \setminus V_1\).
  Observe that for any subset \(S \subseteq V\) it holds that
  \begin{align}
    & \pr{\cap_{v \in S} \{Z_v = \lbl(v)\}} \nonumber \\
    = \ & \pr{\left(\cap_{v \in S \cap V_1} \{X_v = \lbl(v)[1]\}\right) \cap \left(\cap_{v \in S \cap V_2} \{Y_v = \lbl(v)[2]\}\right) } \nonumber \\
    = \ &  \pr{\cap_{v \in S \cap V_1} \{X_v = \lbl(v)[1]\}} \cdot \pr{\cap_{v \in S \cap V_2} \{Y_v = \lbl(v)[2]\} }, \label{eq:fin-dep:indep-processes1}
  \end{align}
	  where the latter equality follows by independence between \(\{X_v\}_{v \in V_1}\) and \(\{Y_v\}_{v \in V_2}\).
  W.l.o.g., suppose \(T_1 \ge T_2\).
  Consider two subsets \(S,S' \subseteq V\) such that \(\dist_G(S,S') > T_1\).
  Fix any output labeling \(\lbl: V \to \labels\) such that \(\lbl(v)[2] = 0\) for all \(v \in V_1 \setminus V_2\) and \(\lbl(v)[1] = 0\) for all \(v \in V_2 \setminus V_1\).
  Using \cref{eq:fin-dep:indep-processes1}, we have that
  \begin{align}
    & \pr{\cap_{v \in S\cup S'} \{Z_v = \lbl(v)\}} \nonumber \\
    = \ & \pr{\cap_{v \in (S \cup S') \cap V_1} \{X_v = \lbl(v)[1]\}} \cdot \pr{\cap_{v \in (S \cup S')  \cap V_2} \{Y_v = \lbl(v)[2]\} } \nonumber \\
    = \ & \pr{\cap_{v \in (S \cap V_1) \cup (S' \cap V_1 )} \{X_v = \lbl(v)[1]\}} \cdot \pr{\cap_{v \in (S \cap V_2) \cup (S' \cap V_2 )} \{Y_v = \lbl(v)[2]\} } \nonumber \\
    = \ & \pr{\cap_{S \cap V_1} \{X_v = \lbl(v)[1] \}} \cdot \pr{\cap_{S' \cap V_1} \{X_v = \lbl(v)[1] \}} \label{eq:fin-dep:T-dependence} \\
    & \cdot \pr{\cap_{S \cap V_2} \{Y_v = \lbl(v)[2] \}} \cdot  \pr{\cap_{S' \cap V_2} \{Y_v = \lbl(v)[2] \}}  \nonumber \\
    = \ & \pr{\left(\cap_{v \in S\cap V_1 } \{X_v = \lbl(v)[1]\}\right) \cap\left(\cap_{v \in S\cap V_2 } \{Y_v = \lbl(v)[2]\}\right) }   \label{eq:fin-dep:indep-processes2} \\
    & \cdot \pr{\left(\cap_{v \in S' \cap V_1 } \{X_v = \lbl(v)[1]\}\right) \cap\left(\cap_{v \in S' \cap V_2 } \{Y_v = \lbl(v)[2]\}\right) }   \nonumber \\
    = \ & \pr{\cap_{v \in S} \{Z_v = \lbl(v)\}} \cdot \pr{\cap_{v \in S'} \{Z_v = \lbl(v)\}}   \nonumber,
  \end{align}
  where \cref{eq:fin-dep:T-dependence} holds because \(\{X_v\}_{v \in V_1}\) is \(T_1\)-dependent and \(\{Y_v\}_{v \in V_2}\) is \(T_2\)-dependent, while \cref{eq:fin-dep:indep-processes2} holds because \(\{X_v\}_{v \in V_1}\) and \(\{Y_v\}_{v \in V_2}\) are independent.
\end{proof}

Now we present a final lemma on the composition of random processes.
To do so, we first introduce the notation of \emph{random decomposition} of a graph.

\begin{definition}[Random decomposition]\label{def:fin-dep:random-decomposition}
  Let \(G = (V,E)\) be any graph and \(\PP\) a family of subgraphs of \(G\).
  For any \(k \in \nat\),
  let \(\Gamma(G)\) be a random variable taking values in \(\PP^k\) 
  that is sampled according to any probability distribution.
  We say that \(\Gamma(G)\) is a \emph{random \(k\)-decomposition} of \(G\) in \(\PP\).
\end{definition}

Given a random \(k\)-decomposition \(\Gamma(G)\) of \(G\) in \(\PP\), for any \(y \in V \cup E\), we define the random variable \(\Gamma(G)_y \in \{0,1\}^k\) as follows: \(\Gamma(G)_y[i] = 1\) if \(y\) belongs to \(\Gamma(G)[i]\) and 0 otherwise.
Notice that \(\{\Gamma(G)_y\}_{y \in V \cup E}\) is a random process on \(G\).
If \(\{\Gamma(G)_y\}_{y \in V \cup E}\) is invariant under automorphisms, then we say that the random \(k\)-decomposition \(\Gamma(G)\) is invariant under automorphisms.
If, for a family of graphs \(\FF\) and any graph \(G \in \FF\), \(\{\Gamma(G)_y\}_{y \in V(G) \cup E(G)}\) is \(T\)-dependent (with \(T\) being a function of the size of \(G\)) and the processes in \(\{\{\Gamma(G)_y\}_{y \in V(G) \cup E(G)} : G \in \FF \}\) are invariant under subgraph isomorphisms, then we say that the random \(k\)-decompositions in \(\{\Gamma(G) : G \in \FF\}\) are \(T\)-dependent and invariant under subgraph isomorphisms.

For a random decomposition, we define the notion of induced random process.

\begin{definition}[Induced process]\label{def:fin-dep:induced-process}
  Let \(G = (V,E)\) be a graph that admits a family of subgraphs \(\PP\).
  Suppose there exists a random process \(\{X_v\}_{v \in V(H(\PP))}\) with \(H(\PP)\) being the graph obtained by the disjoint union of all elements of \(\PP\).
  Let \(\Gamma(G)\) be a random \(k\)-decomposition of \(G\) in \(\PP\).
  For all \(v\in V\) and \(G' \in \PP\), define the random process \(\{X_v^{(G')}\}_{v \in V}\) by setting \(X_v^{(G')} = X_{f_{G'}(v)}\) for all \(v \in V(G')\), where \(f_{G'}: V(G') \to V(H(\PP)) \) is the natural immersion of \(G'\) into \(H(\PP)\) otherwise set \(X_v^{(G')} =  0\).
  Let \(\{Y^{(\Gamma(G))}_v\}_{v\in V}\) be a random process that we define conditional on the output of \(\Gamma(G)\): 
  for all \(\Gbold \in \PP^k\), conditional on \(\Gamma(G) = \Gbold\),
  \begin{align*}
    Y^{(\Gamma(G))}_v = Y^{(\Gbold)}_v = (X_v^{(\Gbold[1])}, \dots, X_v^{(\Gbold[k])}).
  \end{align*}
  The random process \(\{Y_v^{(\Gamma(G))}\}_{v\in V}\) on \(G\) is said to be induced by the action of \(\{X_v\}_{v \in V(H(\PP))}\) over the random \(k\)-decomposition \(\Gamma(G)\).
\end{definition}

Now we present a result on the induced random process whenever the random decomposition and the family of random processes that acts on the random decomposition meet some invariance properties.

\begin{lemma}\label{lemma:fin-dep:induced-process}
  Let \(\FF\) be a family of graphs.
  For any \(G = (V,E) \in \FF\), let \(\PP_G\) be any family of subgraphs of \(G\) that is closed under node removal and disjoint graph union, that is, if \(G_1,G_2 \in \PP_G\) and \(V(G_1) \cap V(G_2) = \emptyset\), then \(G_1 \cup G_2 \in \PP_G\).
  Furthermore, suppose that, for each pair of isomorphic subgraphs \(G_1,G_2 \subseteq G\),  \(G_1 \in \PP_G \implies G_2 \in \PP_G\). 
  Moreover, let \(\Gamma(G)\) be a random \(k\)-decomposition of \(G\) in \(\PP_G\) that is \(T_1\)-dependent, and suppose that the random decompositions in \(\{\Gamma(G) : G \in \FF\}\) are invariant under subgraph isomorphism.
  Suppose there is a \(T_2\)-dependent distribution \(\{X_v\}_{v \in  V(H(\PP_G))}\), taking values in a finite set \(\labels\), such that \(\{\{X_v\}_{v \in  V(H(\PP_G))} : G \in \FF\}\) is invariant under subgraph isomorphism, where  \(H(\PP_G)\) is obtained by the disjoint union of a copy of each element of \(\PP_G\).   
  Let \(\{Y^{(\Gamma(G))}_v\}_{v \in V(G)}\) be the random process induced by the action of \(\{X_v\}_{v \in V(H(\PP_G))}\) over \(\Gamma(G)\).
  Then, \(\{Y^{(\Gamma(G))}_v\}_{v\in V(G)}\) is a \((T_1 + 2T_2)\)-dependent distribution.
  Furthermore, the processes in \(\{\{X_v\}_{v \in V(H(\PP_G))} : G \in \FF\}\) are invariant under subgraph isomorphism.
\end{lemma}
\begin{proof}
  Fix \(G \in \FF\) and let \(\Gamma = \Gamma(G)\), \(\PP = \PP_G\).
  Since \(\PP^k\) might be uncountable, we consider the density function \(f_{\Gamma}\) and the probability measure \(\Pbb\) associated to \(\Gamma\).
  Consider two subsets of nodes \(S, S' \subseteq V\) at distance at least \(\max(T_1,T_2) + 1\).
  Fix any labeling \(\lbl: V \to \labels^k\).
  It holds that
  \begin{align}
    & \pr{\cap_{v \in S \cup S'} \{Y_v^{(\Gamma)} = \lbl(v)\}} \nonumber \\
    = \ & \int_{\PP^k} \pr{\cap_{v \in S \cup S'} \{Y_v^{(\Gamma)} = \lbl(v)\} \ \st \ \Gamma = \Gbold }  f_{\Gamma}(\Gbold) \dd{\Pbb} \label{eq:fin-dep:random-decomp:law-total-prob} \\
    = \ &  \int_{\PP^k} \pr{\cap_{v \in S \cup S'} \{Y^{(\Gbold)}_v = \lbl(v)\}} f_{\Gamma}(\Gbold) \dd{\Pbb} \nonumber \\
    = \ &  \int_{\PP^k} \pr{\cap_{v \in S } \{Y^{\Gbold}_v = \lbl(v)\} } \pr{\cap_{v \in  S'} \{Y^{\Gbold} = \lbl(v)\} } f_{\Gamma}(\Gbold) \dd{\Pbb} \label{eq:fin-dep:random-decomp:indep} \\
    = \ &  \int_{\PP^k}  \prod_{i \in [k]} \pr{\cap_{v \in S  } \{X^{\Gbold[i]}_v = \lbl(v)[i]\} }  \pr{\cap_{v \in  S' } \{X^{\Gbold[i]}_v = \lbl(v)[i]\} } f_{\Gamma}(\Gbold) \dd{\Pbb} , \label{eq:fin-dep:random-decomp:non-zero}
  \end{align}
  where \cref{eq:fin-dep:random-decomp:law-total-prob} holds by the law of total probability, \cref{eq:fin-dep:random-decomp:indep} holds since \(\{Y_v^{(\Gbold)}\}_{v\in V}\) is \(T_2\)-dependent for all \(\Gbold \in \PP^k\) by \cref{lemma:fin-dep:composition}, and \cref{eq:fin-dep:random-decomp:non-zero} holds since \(\{X^{\Gbold[i]}_v\}_{v \in V}\) and \(\{X^{\Gbold[j]}_v\}_{v \in V}\) are independent if \(i \neq j\).
  Notice that, for any set \(S \subseteq V\), \(G' = \Gbold[i][\NN_{T_2} (S, \Gbold[i])] \in \PP\) and the event \(\cap_{v \in S } \{X^{\Gbold[i]}_v = \lbl(v)[i]\} \) has the same probability as \(\cap_{v \in S } \{X^{G'}_v = \lbl(v)[i]\} \) because \(\{X_v\}_{v \in V(H(\PP))}\) is invariant under subgraph isomorphisms.

  For any subset \(U \subseteq V\) and any integer \(T \ge 0\), let 
  \[
    \HH(T,U) = \left\{(\Gbold[1][\NN_{T} (U, \Gbold[1])], \dots, \Gbold[k][\NN_{T} (U, \Gbold[k])])  \, \st \, \Gbold \in \PP^k \right\} \subseteq \PP^k,
  \]
  and let \(\Pbb[\HH(T,U)]\) be the restriction of \(\Pbb\) to \(\HH(T,U)\) defined as follows: for any event \(E\), 
  \[
    \Pbb[\HH(T,U)](E) = \Pbb(E\cap \HH(T,U)) / \Pbb(\HH(T,U)).
  \]
  Finally, define the random variable \(\Gamma^{(T,U)}\) to be a \(k\)-dimensional vector, taking values in \(\HH(T,U)\),  whose \(i\)-th entry is \(\Gamma[i][\NN_{T}(U, \Gamma[i])]\), and let \(f_{\Gamma^{(T,U)}}\) be its density function.
  Define \(U = S \cup S'\):  \cref{eq:fin-dep:random-decomp:non-zero} becomes
  \begin{align}
    &  \int_{\PP^k}  \prod_{i \in [k]} \pr{\cap_{v \in S  } \{X^{\Gbold[i]}_v = \lbl(v)[i]\} }  \pr{\cap_{v \in  S' } \{X^{\Gbold[i]}_v = \lbl(v)[i]\} } f_{\Gamma}(\Gbold) \dd{\Pbb} \label{eq:fin-dep:random-decomp:graph-inv} \\
    = \ & \int_{\HH(T_2,U)} \prod_{i \in [k]} \pr{\cap_{v \in S  } \{X^{\Hbold[i]}_v = \lbl(v)[i]\} }  \pr{\cap_{v \in  S' } \{X^{\Hbold[i]}_v = \lbl(v)[i]\} } f_{\Gamma^{(T_2,U)}}(\Hbold) \dd{\Pbb[\HH(T_2,U)]}. \nonumber
  \end{align}
  Now notice that, since the random decomposition \(\Gamma\) is \(T_1\)-dependent and \(\dist_G(S,S') > T_1 + 2T_2\), the probability space \(\HH(T_2,U)\) (with measure \(\Pbb[\HH(T_2,U)]\))  is isomorphic to the  product space \(\HH(T_2,S) \times \HH(T_2,S')\) (with product measure \(\Pbb[\HH(T_2,S)] \times \Pbb[\HH(T_2,S')]\)).
  The isomorphism brings any element
  \[
    (\Gbold[1][\NN_{T_2} (U, \Gbold[1])], \dots, \Gbold[k][\NN_{T_2} (U, \Gbold[k])])
  \]
  of \(\HH(T_2,U)\)
  into
  \[
    \left((\Gbold[1][\NN_{T_2} (S, \Gbold[1])], \dots, \Gbold[k][\NN_{T_2} (S, \Gbold[k])]) ,(\Gbold[1][\NN_{T_2} (S', \Gbold[1])], \dots, \Gbold[k][\NN_{T_2} (S', \Gbold[k])]) \right)
  \]
  which belongs to \(\HH(T_2,S) \times \HH(T_2,S')\).
  Hence, we get
  \begin{align}
    & \int_{\HH(T_2,U)} \prod_{i \in [k]} \pr{\cap_{v \in S  } \{X^{\Hbold[i]}_v = \lbl(v)[i]\} }  \pr{\cap_{v \in  S' } \{X^{\Hbold[i]}_v = \lbl(v)[i]\} } f_{\Gamma^{(T_2,U)}}(\Hbold) \dd{\Pbb[\HH(T_2,S)]} \nonumber \\
    = \ & \int_{\HH(T_2,S)} \int_{\HH(T_2,S')} \prod_{i \in [k]} \pr{\cap_{v \in S  } \{X^{\Hbold[i]}_v = \lbl(v)[i]\} }  \pr{\cap_{v \in  S' } \{X^{\Hbold[i]}_v = \lbl(v)[i]\} } \nonumber \\
    & \cdot f_{\Gamma^{(T_2,S)}}(\Hbold_1)  f_{\Gamma^{(T_2,S')}}(\Hbold_2) \dd{\Pbb[\HH(T_2,S)]} \dd{\Pbb[\HH(T_2,S')]} \nonumber. 
  \end{align}
  The latter becomes
  \begin{align}
    & \int_{\HH(T_2,S)} \int_{\HH(T_2,S')} \prod_{i \in [k]} \pr{\cap_{v \in S  } \{X^{\Hbold[i]}_v = \lbl(v)[i]\} }  \pr{\cap_{v \in  S' } \{X^{\Hbold[i]}_v = \lbl(v)[i]\} } \nonumber \\
    & \cdot f_{\Gamma^{(T_2,S)}}(\Hbold_1)  f_{\Gamma^{(T_2,S')}}(\Hbold_2) \dd{\Pbb[\HH(T_2,S)]} \dd{\Pbb[\HH(T_2,S')]} \nonumber \\
    = \ & \int_{\HH(T_2,S)} \prod_{i \in [k]} \pr{\cap_{v \in S  } \{X^{\Hbold[i]}_v = \lbl(v)[i]\} }  f_{\Gamma^{(T_2,S)}}(\Hbold_1)   \dd{\Pbb[\HH(T_2,S)]}  \nonumber \\
    & \cdot \int_{\HH(T_2,S')} \prod_{i \in [k]} \pr{\cap_{v \in S'  } \{X^{\Hbold[i]}_v = \lbl(v)[i]\} }  f_{\Gamma^{(T_2,S')}}(\Hbold_2)   \dd{\Pbb[\HH(T_2,S')]}  \nonumber \\
    = \ & \int_{\PP^k} \prod_{i \in [k]} \pr{\cap_{v \in S  } \{X^{\Gbold[i]}_v = \lbl(v)[i]\} }  f_{\Gamma}(\Gbold)   \dd{\Pbb} \label{eq:fin-dep:random-decomp:reassembling} \\
    & \cdot \int_{\PP^k} \prod_{i \in [k]} \pr{\cap_{v \in S'  } \{X^{\Gbold[i]}_v = \lbl(v)[i]\} }  f_{\Gamma}(\Gbold)   \dd{\Pbb}  \nonumber \\
    = \ & \pr{\cap_{v \in S} \{Y_v^{(\Gamma(G))} = \lbl(v)\}} \pr{\cap_{v \in S'} \{Y_v^{(\Gamma(G))} = \lbl(v)\}} \nonumber,
  \end{align}
  where \cref{eq:fin-dep:random-decomp:reassembling} follows by the same reasoning as for \cref{eq:fin-dep:random-decomp:graph-inv} but in reverse.

  Now we want to prove that \(\{\{Y^{(\Gamma(G))}_v\}_{v \in V} : G \in \FF\}\) is invariant under subgraph isomorphism.
  Fix any two graphs \(G, H \in \FF\) of sizes \(n_G\) and \(n_H\), respectively.
  Consider any isomorphism \(\alpha\) between the radius-\((T_1(n_G) + 2T_2(n_G))\) neighborhoods of any subgraph \(G' \subseteq G\) and the radius-\((T_1(n_H) + 2T_2(n_H))\) neighborhoods of any subgraph \(H' \subseteq H\), such that the restriction of \(\alpha\) to \(G'\) is an isomorphism to \(H'\). 
  With an abuse of notation, for any subgraph \(K \subseteq G'\), let us denote \(\alpha(K) \subseteq H'\) its isomorphic image in \(H'\) through \(\alpha\).
  Let \(T_G = T_1(n_G) + 2T_2(n_G)\) and  \(T_H = T_1(n_H) + 2T_2(n_H)\).
  Fix any labeling \(\lbl: V \to \labels ^{k}\).
  We have that
  \begin{align}
    & \pr{\cap_{v \in V(G')} \{Y_v^{(\Gamma(G))} = \lbl(v)\}} \nonumber \\
    = \ & \int_{\PP_G^k} \pr{\cap_{v \in V(G')} \{Y_v^{(\Gamma(G))} = \lbl(v)\} \ \st \ \Gamma(G) = \Gbold }  f_{\Gamma(G)}(\Gbold) \dd{\Pbb_G}  \nonumber. \\
    = \ &  \int_{\PP_G^k} \pr{\cap_{v \in V(G')} \{Y^{(\Gbold)}_v = \lbl(v)\}} f_{\Gamma(G)}(\Gbold) \dd{\Pbb_G} \nonumber \\
    = \ &  \int_{\PP_G^k}  \prod_{i \in [k]} \pr{\cap_{v \in V(G')  } \{X^{(\Gbold[i])}_{v} = \lbl(v)[i]\} } f_{\Gamma(G)}(\Gbold) \dd{\Pbb_G} \nonumber \\
    = \ & \int_{\HH_G(T_G,V(G'))}  \prod_{i \in [k]} \pr{\cap_{v \in V(G')  } \{X^{(\Hbold[i])}_{v} = \lbl(v)[i]\} } f_{\Gamma(G)^{(T_G,V(G'))}}(\Hbold) \dd{\Pbb_G(\HH_G(T_G,V(G')))} \nonumber \\
    = \ & \int_{\HH_G(T_G,V(G'))}  \prod_{i \in [k]} \pr{\cap_{v \in V(G')  } \{X^{(\alpha(\Hbold[i]))}_{\alpha(v)} = \lbl(v)[i]\} } f_{\Gamma(G)^{(T_G,V(G'))}}(\Hbold) \dd{\Pbb_G(\HH_G(T_G,V(G')))},\nonumber   
  \end{align}
  where the latter holds because 
  \[ 
    \pr{\cap_{v \in V(G')  } \{X^{(\Hbold[i])}_{v} = \lbl(v)[i]\} }=\pr{\cap_{v \in V(G')  } \{X^{(\alpha(\Hbold[i]))}_{\alpha(v)} = \lbl(v)[i]\} }
  \]
  as \(\{\{X_v\}_{v \in V(H(\PP_G))}: G \in \FF\}\) is \(T_2\)-dependent and invariant under subgraph isomorphism.
  Furthermore, since the processes in \(\{\Gamma(G) : G \in \FF\}\) are \(T_1\)-dependent and invariant under subgraph isomorphism, we have that \(f_{\Gamma(G)^{(T_G,V(G'))}}(\Hbold) = f_{\Gamma(H)^{(T_H,V(H'))}}((\alpha(\Hbold[1]), \dots, \alpha(\Hbold[k])))\) almost everywhere in \(\HH_G(T_G,V(G'))\), and that the probability space \(\HH_G(T_G,V(G'))\) with measure \( \dd{\Pbb_G(\HH_G(T_G,V(G')))}\) is isomorphic to  \(\HH_H(T_H,V(H'))\) with measure \( \dd{\Pbb_H(\HH_H(T_H,V(H')))}\), where the isomorphism brings \(\Hbold\) into  \((\alpha(\Hbold[1]), \dots, \alpha(\Hbold[k]))\).
  Hence,
  \begin{align}
    & \int_{\HH_G(T_G,V(G'))}  \prod_{i \in [k]} \pr{\cap_{v \in V(G')  } \{X^{(\alpha(\Hbold[i]))}_{\alpha(v)} = \lbl(v)[i]\} } f_{\Gamma(G)^{(T_G,V(G'))}}(\Hbold) \dd{\Pbb_G(\HH_G(T_G,V(G')))}\nonumber  \\
    = \ & \int_{\HH_H(T_H,V(H'))}  \prod_{i \in [k]} \pr{\cap_{v \in V(H')  } \{X^{(\Hbold[i])}_{v} = \lbl(v)[i]\} } f_{\Gamma(H)^{(T_H,V(H'))}}(\Hbold) \dd{\Pbb_H(\HH_H(T_H,V(H')))}\nonumber  \\
    = \ & \pr{\cap_{v \in V(H')} \{Y_v^{(\Gamma(H))} = \lbl(v)\}}, \nonumber 
  \end{align}
  concluding the proof.
\end{proof}

\begin{remark}
  When the underlying graph is a directed graph, \cref{lemma:fin-dep:induced-process} guarantees invariance under subgraph isomorphisms that keep edge orientation.
\end{remark}

We are going to work on rooted pseudotrees and pseudoforests.
A \emph{pseudotree} is a graph that is connected and contains at most one cycle. 
A \emph{pseudoforest} is a graph obtained by the disjoint union of pseudotrees; an equivalent definition of pseudoforest is a graph in which each connected component has no more edges than vertices.
Note that a pseudotree might contain multiple edges: however, we assume it does not contain self-loops as self-loops are useless communication links in the \local model.
A \emph{rooted tree} is a tree where each edge is oriented and all nodes have outdegree at most 1: it follows that all but one node have outdegree exactly 1 and one node (the \emph{root}) has outdegree 0.
Trivially, a tree can be rooted by selecting one node and orienting all edges towards it.
A \emph{rooted pseudotree} is a pseudotree where each edge is oriented and each node has outdegree at most 1: if the pseudotree contains a cycle, then all nodes necessarily have outdegree exactly 1. 
Any pseudotree can be oriented so that it becomes rooted: just orient the cycle first (if it exists) in a consistent way, then remove it, and make the remaining trees rooted at nodes that belonged to the cycle.
A \emph{rooted pseudoforest} is the union of rooted pseudotrees.

We will show that pseudoforests of maximum degree \(\maxDeg\) admit an \(\myO{\log* \maxDeg}\)-dependent \(3\)-coloring distribution.
In order to do so, we use a color reduction technique that follows by known revisions of the Cole--Vishkin technique \cite{cole1986,goldberg1988}.

\begin{lemma}[\PN algorithm for color reduction in pseudoforests \cite{cole1986,goldberg1988}]\label{lemma:fin-dep:cole-vishkin}
  Let \(G \) be a pseudoforest with countably many nodes.
  Assume \(G\) is given as an input a \(k\)-coloring for some \(k \ge 3\).
  There exists a deterministic \PN algorithm that does not depend on the size of \(G\) and outputs a \(3\)-coloring of \(G\) in time \(\myO{\log* k}\).  
\end{lemma}

\begin{figure}[t]
  \centering
  \includegraphics[scale=0.4]{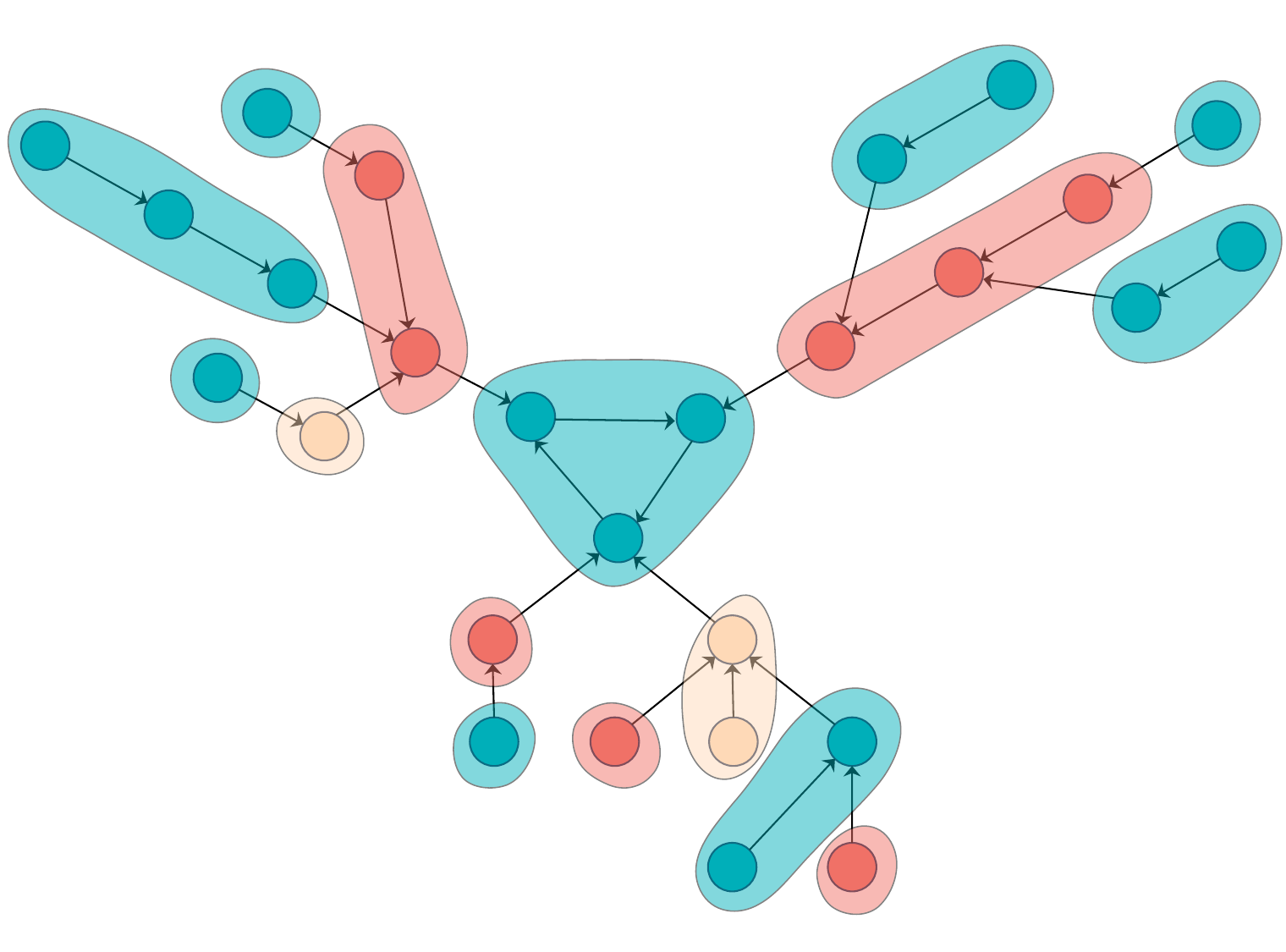}
  \caption{A decomposition of rooted pseudoforests in directed paths and cycles: each node \(v\) colors its in-neighbors with a uniformly sampled permutation of the elements of \([\inDeg(v)]\).
  The graph induced by nodes colored with color \(i\) is a disjoint union of directed paths and cycles.}
  \label{fig:fin-dep:rooted-pseudoforest}
\end{figure}

\begin{figure}[tp]
  \centering
  \begin{subfigure}{\textwidth}
    \centering
    \includegraphics[scale=0.33]{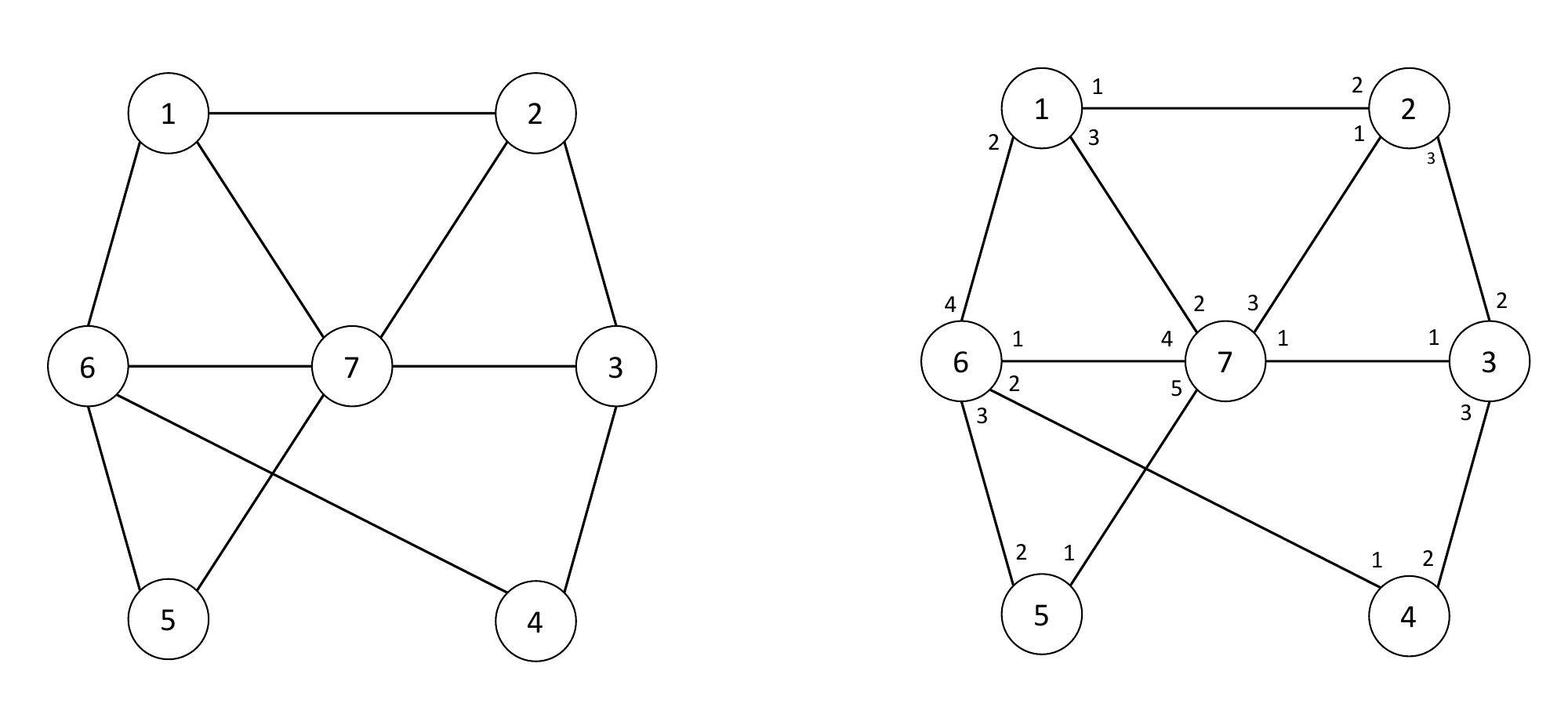}
    \caption{Each node \(v\) rearranges its port numbers uniformly at random.}
    \label{fig:fin-dep:bounded-graphs-1}
  \end{subfigure}
  \\[5mm]
  \begin{subfigure}{\textwidth}
    \centering
    \includegraphics[scale=0.45]{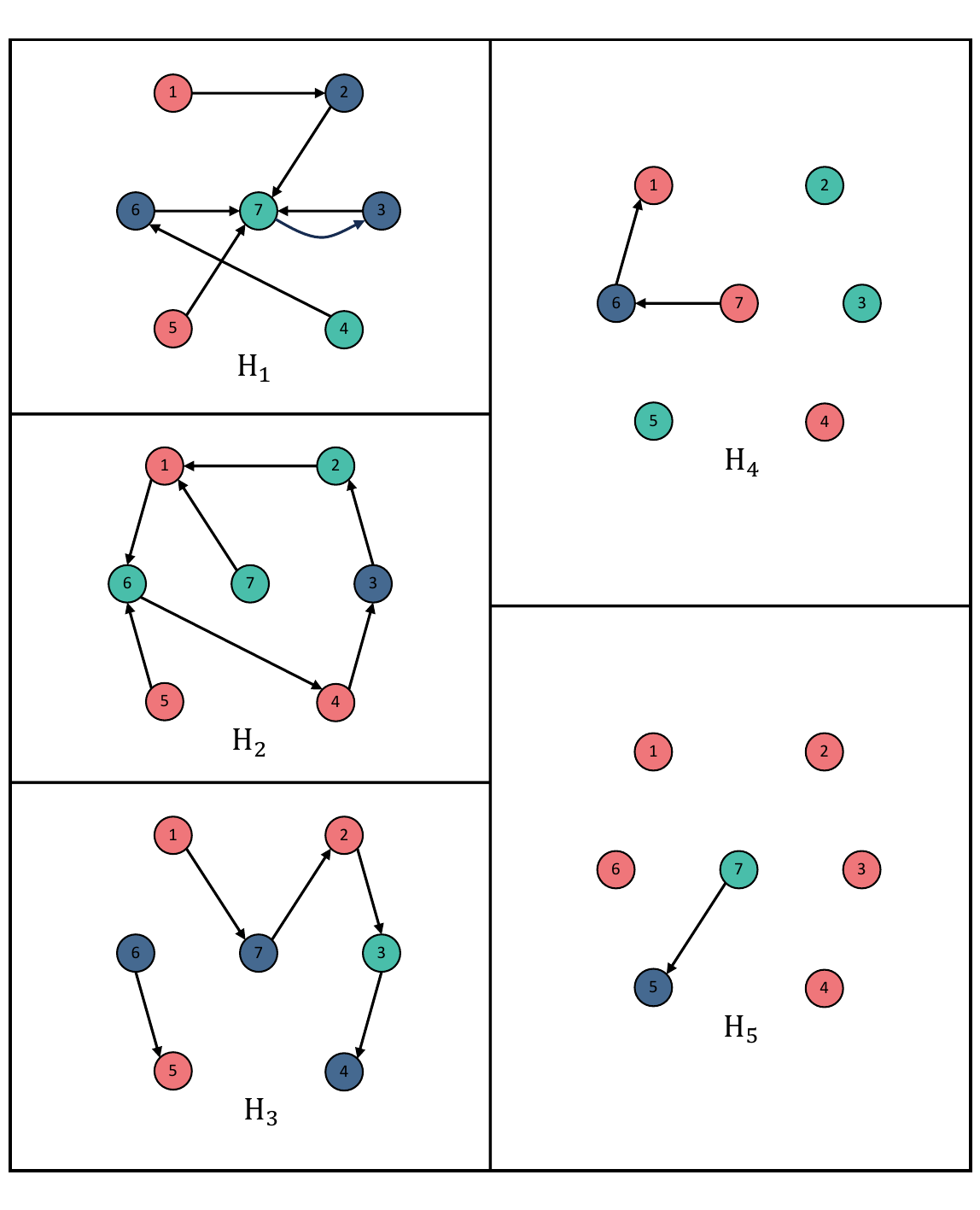}
    \caption{The \(\maxDeg = 5\) rooted pseudoforests that we take by considering port numbers \(i \in [\maxDeg]\).}
    \label{fig:fin-dep:bounded-graphs-2}
  \end{subfigure}
  \caption{A decomposition of a graph of maximum degree \(\maxDeg = 5\) in rooted pseudoforests: for the sake of image clarity, we focus on the undirected case.
  In \cref{fig:fin-dep:bounded-graphs-1}, each node \(v\) rearranges its port numbers with a uniformly sampled permutation of the elements of \([\deg(v)]\).
  As shown in \cref{fig:fin-dep:bounded-graphs-2}, edges hosting port number \(i\) at some endpoint are oriented away from that port (in case both endpoints host port number \(i\), the edge is duplicated) and form a rooted pseudoforest.}
  \label{fig:fin-dep:bounded-graphs}
\end{figure}

Now we are ready to prove our result on pseudoforests.

\begin{lemma}[Finitely-dependent coloring of rooted pseudoforests]\label{lemma:fin-dep:col-rooted-pseudoforest}
  Let \(\FF\) be a family of rooted pseudoforests of countably many nodes of maximum degree \(\maxDeg\).
  Then, there exists an outcome that associates to each graph \(G \in \FF\) a \(\myO{\log* \maxDeg}\)-dependent distribution on the vertices of \(G\) that gives a \(3\)-coloring of \(G\).
  Furthermore, the family of distributions output by the outcome is invariant under subgraph isomorphisms.
\end{lemma}
\begin{proof}
  Let us fix the rooted pseudoforest \(G \in \FF\).
  Let \(\PP_G\) be the family of all subgraphs of \(G\) formed by the disjoint union of directed paths and cycles.
  Notice that \(\PP_G\) is closed under node removal and disjoint graph union.
  Furthermore, for any pair of isomorphic subgraphs \(G_1,G_2 \subseteq G\), \(G_1 \in \PP_G \implies G_2 \in \PP_G\).
  Let \(H(\PP_G)\) be the graph formed by the disjoint union of a copy of each element of \(\PP_G\).
  By \cref{cor:fin-dep:baseline}, \(H(\PP_G)\)
  admits a \(1\)-dependent \(4\)-coloring distribution \(\{X_v\}_{v\in H(\PP_G)}\) (with colors in \([4]\)), such that \(\{\{X_v\}_{v\in H(\PP_G)} : G \in \FF\}\) is invariant under subgraph isomorphism.

  Consider now a (non-proper) coloring of the pseudoforest in which each node \(u\) colors its in-neighbors with a permutation of \(\{1, \dots, \inDeg(u)\}\) sampled uniformly at random:
  if a node has outdegree zero, then it is deterministically colored with color 1.
  Such a coloring is described by a 2-dependent distribution \(\{Z_v\}_{v \in V}\) that is (trivially) invariant under subgraph isomorphisms that keep edge orientation. 
  Also,
  \(\{Z_v\}_{v \in V}\) identifies \(\maxInDeg\) disjoint random subsets of nodes \(\Vrand_1, \dots, \Vrand_{\maxInDeg}\), where nodes in \(\Vrand_i\) are colored with the color \(i\), and \(\maxInDeg\) is the maximum indegree of the graph.
  Let \(\Gamma(G) = (G[\Vrand_1], \dots, G[\Vrand_{\maxInDeg}])\), where \(\Gamma(G)[i]\) 
  is the random graph induced by \(\Vrand_i\).
  Furthermore, observe that the output of \(\Gamma(G)[i]\) is the disjoint union of oriented paths and/or oriented cycles, with \(\Gamma(G)[i]\) being the \(i\)-th entry of the \(\maxInDeg\)-tuple \(\Gamma(G)\) (see \cref{fig:fin-dep:rooted-pseudoforest}).
  Notice that the process \(\Gamma(G)\) is 2-dependent and is a random \(\maxInDeg\)-decomposition of \(G\) in \(\PP_G\) (according to \cref{def:fin-dep:random-decomposition}), such that the random decompositions in \(\{\Gamma(G) : G \in \FF \}\) are invariant under subgraph isomorphisms.
  By \cref{lemma:fin-dep:induced-process}, the random process \(\{Y^{(\Gamma(G))}_v\}_{v \in V}\), that is induced by the action of \(\{X_v\}_{v \in H(\PP_G)}\) over the random \(\maxInDeg\)-decomposition \(\Gamma(G)\) (according to \cref{def:fin-dep:induced-process}), is a \(4\)-dependent distribution that gives a \(4\maxInDeg\)-coloring of \(G\): in fact, \(\Gamma(G)[i]\) and \(\Gamma(G)[j]\) are disjoint if \(i\neq j\), hence only one entry of \(Y^{(\Gamma(G))}_v\) is non-zero, for all \(v \in V\).
   
  We then combine \(\{Y^{(\Gamma(G))}_v\}_{v\in V}\) and a modified version of the \PN algorithm from \cref{lemma:fin-dep:cole-vishkin} where at round 0 each node permutes port numbers locally u.a.r.:
  by \cref{lemma:fin-dep:combination}, we obtain an \(\myO{\log* \maxDeg}\)-dependent 3-coloring distribution \(\{Q^{(G)}_v\}_{v \in V(G)}\) of \(G\) such that processes in \(\{\{Q^{(G)}_v\}_{v \in V(G)} : G \in \FF\}\) are invariant under subgraph isomorphisms.
\end{proof}

Our finitely-dependent coloring of pseudoforests can be used as a baseline to provide a \((\maxDeg + 1)\)-coloring of graphs with maximum degree \(\maxDeg\).
The tool we use is again an application of the Cole--Vishkin color reduction technique  \cite{maus2022}.

\begin{lemma}[\PN algorithm for color reduction of general graphs \cite{maus2022}]\label{lemma:fin-dep:color-reduction-general-graphs}
  Let \(G = (V,E)\) be a graph with maximum degree \(\maxDeg\) and countably many nodes.
  Suppose \(G\) is given in input a \(k\)-coloring for some \(k \ge \maxDeg + 1\).
  There exists a deterministic \PN algorithm that does not depend on the size of the input graph and outputs a \((\maxDeg + 1)\)-coloring of \(G\) in time \(\myO{\log* k + \sqrt{\maxDeg \log \maxDeg}}\).
\end{lemma}

We first present a corollary of \cref{lemma:fin-dep:color-reduction-general-graphs} where we characterize the combination of an input finitely-dependent coloring distribution and the \PN color-reduction algorithm. 

\begin{corollary}\label{cor:fin-dep:color-reduction-from-fin-dep-coloring}
  Let \(G = (V,E)\) be a graph with maximum degree \(\maxDeg\) and countably many nodes.
  Let \(\{X_v\}_{v \in V}\) be a \(T\)-dependent distribution that gives a \(k\)-coloring of \(G\), with \(k \ge \maxDeg + 1\).
  Then there exists a distribution \({Y_v}_{v \in V}\) that gives a \((\maxDeg + 1)\)-coloring of \(G\) and is \(\myO{\log* k + \sqrt{\maxDeg \log \maxDeg} + T}\)-dependent.
  Furthermore, if \(\{X_v\}_{v \in V}\) is invariant under subgraph isomorphisms, then so is \(\{Y_v\}_{v \in V}\).
\end{corollary}
\begin{proof}
  We combine the distribution \(\{X_v\}_{v \in V}\) with a modified version of the \PN algorithm from \cref{lemma:fin-dep:color-reduction-general-graphs} where at round 0 each node permutes port numbers locally u.a.r.:
  \cref{lemma:fin-dep:combination} implies the existence of an \(\myO{\log* k + \sqrt{\maxDeg \log \maxDeg} + T}\)-dependent \((\maxDeg + 1)\)-coloring distribution of \(G\).
  If \(\{X_v\}_{v \in V}\) is invariant under subgraph isomorphism, \cref{lemma:fin-dep:combination} implies that \(\{Y_v\}_{v \in V}\) has the same property.
\end{proof}

\begin{lemma}[Finitely-dependent coloring of bounded-degree graphs]\label{lemma:fin-dep:coloring-bounded-deg-graphs}
  Let \(\FF\) be a family of graphs of countably many nodes and maximum degree \(\maxDeg\).
  Then, there exists an outcome that associates to each graph \(G \in \FF\) an \(O({ \sqrt{\maxDeg \log \maxDeg}})\)-dependent distribution on the vertices of \(G\) that gives a \((\maxDeg + 1)\)-coloring of \(G\).
  Furthermore, the family of distributions outputted by the outcome are invariant under subgraph isomorphisms.
\end{lemma}
\begin{proof}
  Let us fix \(G = (V,E) \in \FF\).
  First, if \(G\) is not directed, then duplicate each edge and give to each pair of duplicates different orientations.
  Since a coloring of the original graph is a proper coloring if and only if the same coloring is proper in the directed version, w.l.o.g., we can assume \(G\) to be directed.

  Let \(\PP_G\) be a family of all subgraphs of \(G\) formed by rooted pseudotrees and disjoint union of rooted pseudotrees.
  Notice that \(\PP_G\) is closed under node removal and disjoint graph union.
  Furthermore, for any two pair of isomorphic subgraphs \(G_1,G_2 \subseteq G\), \(G_1 \in \PP_G \implies G_2 \in \PP_G\).
  The graph \(H(\PP_G)\) that is the disjoint union of all elements of copies of each \(\PP_G\) is a rooted pseudoforest of maximum degree \(\maxDeg\) and, by \cref{lemma:fin-dep:col-rooted-pseudoforest}, admits a \(3\)-coloring \(\myO{\log* \maxDeg}\)-dependent distribution \(\{X_v\}_{v \in H(\PP_G)}\) such that processes in \(\{\{X_v\}_{v \in H(\PP_G)} : G \in \FF\}\) are invariant under subgraph isomorphism.
  
  Now, consider a process in which each node \(v\) samples uniformly at random a permutation of a port numbering from \(\{1, \dots, \outDeg(v)\}\) for its outgoing edges.
  For each \(i \in [\maxOutDeg]\), consider the graph \(G_i\) induced by edges that host port \(i\):
  notice that \(G_i\) is a rooted pseudoforest as  each node has at most one out-edge with port \( i \). 
  If a node has degree \(0\), it deterministically joins \(G_1\), which remains a pseudoforest.
  The random choice of port numbering defines a random variable \(\Gamma \in \PP_G^k\), where \(\Gamma[i]\) is the graph induced by port number \(i\): according to \cref{def:fin-dep:random-decomposition}, we obtain a random \(\maxOutDeg\)-decomposition \(\Gamma(G)\) of \(G\) in \(\PP_G\) which is 2-dependent (by construction): also, the random decompositions in \(\{\Gamma(G) : G \in \FF\}\) are invariant under subgraph isomorphisms.
  For an example of a possible output of the random decomposition, see \cref{fig:fin-dep:bounded-graphs}.

  Hence, the random process \(\{Y_v^{(\Gamma(G))}\}_{v \in V}\) from \cref{def:fin-dep:induced-process} is well defined, and provides a proper \(3^\maxDeg\)-coloring of \(G\).
  By \cref{lemma:fin-dep:induced-process}, the random process \(\{Y_v^{(\Gamma(G))}\}_{v \in V}\) is \(\myO{\log* \maxDeg}\)-dependent and, when \(G \in \FF\) varies, the processes \(\{Y_v^{(\Gamma(G))}\}_{v \in V}\) are  invariant under subgraph isomorphisms.

  By \cref{cor:fin-dep:color-reduction-from-fin-dep-coloring} we obtain an \(O({ \sqrt{\maxDeg \log \maxDeg}})\)-dependent \((\maxDeg + 1)\)-coloring distribution on \(G\) that is invariant under subgraph isomorphism (as \(G \in \FF\) varies).
\end{proof}

\cref{lemma:fin-dep:coloring-bounded-deg-graphs} answers an open question posed
by Holroyd \cite{holroyd2024}.
(See also \cref{cor:intro:fin-dep-trees} and the discussion around it.)

\begin{corollary}\label{cor:fin-dep:trees}
  Let \(G = (V,E)\) be the infinite \(d\)-regular tree.
  There exists a finitely-dependent distribution giving a \((d+1)\)-coloring of \(G\) that is invariant under automorphisms.
\end{corollary}
\begin{proof}
  Finding a \((d+1)\)-coloring of any \(d\)-regular tree has complexity \(O(1)\) in \slocal;
  the coloring can just be performed greedily.
  \cref{lemma:fin-dep:coloring-bounded-deg-graphs} yields the desired result.
\end{proof}

Consider any graph \(G = (V,E)\).
For any \(k \in \natPos\), a distance-\(k\) coloring of \(G\) is an assignment of colors \(c: V \to \labels\) such that, for each node \(v \in V\), all nodes in \(\NN_k(v) \setminus \{v\} = \{u \in V: \dist_G(u,v) \le k\}  \setminus \{v\}\) have colors that are different from \(c(v)\).

It is well known that any LCL problem \(\problem\) that has complexity \(\myO{\log* n}\) in the \local model has the following property: 
there exists a constant \(k \in \natPos\) (that depends only on the hidden constant in \(\myO{\log* n}\)) such that, if the input graph is given a distance-\(k\) coloring, then \(\problem\) is solvable in time \(\myO{1}\) in the \portnum model \cite{chang16exponential}.
Furthermore, no knowledge of the size of the input graph is required.

\begin{theorem}\label{thm:fin-dep}
  Consider any LCL problem \(\problem\) with checking radius \(r\) that has complexity \(T \ge r\) in the \local model over a family \(\FF\) of graphs with maximum degree \(\maxDeg\), where \(T = \myO{\log* n}\) for input graphs of size \(n\).
  For each \(G \in \FF\), there exists an \(\myO{f(\maxDeg)}\)-dependent distribution \(\{Y_v^{(G)}\}_{v \in V(G)}\) that solves \(\problem\) over \(G\).
  Furthermore, 
  the processes in \(\{\{Y_v^{(G)}\}_{v \in V(G)} : G \in \FF\}\) are invariant under subgraph isomorphisms.  
\end{theorem}
\begin{proof}
  Fix \(G \in \FF\) of size \(n\).
  As said before, there exists a constant \(k \in \natPos\) (that depends only on the description of the \(\problem\)) such that, if the input graph is given a distance-\(k\) coloring, then \(\problem\) is solvable in time \(\myO{1}\) (hiding some dependence on \(\maxDeg\)) in the \portnum model \cite{chang16exponential}.
  Furthermore, no knowledge of the size of the input graph is required.
  Consider the power graph \(G^{k}\), where \(k \ge T\), defined by \(G^k = (V, E^k)\) with \(E^k = \{\{u,v\} : u,v \in V, \dist_G(u,v) \le k\}\).
  The maximum degree of \(G^k\) is \(\maxDeg^k\).
  By \cref{lemma:fin-dep:coloring-bounded-deg-graphs}, \(G^k\) admits an \(\myO{\sqrt{k \maxDeg ^k \log \maxDeg}}\)-dependent \((\maxDeg^k + 1)\)-coloring distribution \(\{X_v\}_{v \in V}\) that is invariant under subgraph isomorphisms (as \(G \in \FF\) varies).
  Notice that such a coloring provides a distance-\(k\) coloring for \(G\).

  Let \(\algo\) be the algorithm in the \portnum model that solves \(\problem\) in time \(O(1)\) when given the distance-\(k\) coloring as input.
  Consider a \PN algorithm \(\AA'\) that simulates \(\algo\) and is defined as follows:
  At round 0, \(\AA'\) permutes ports locally u.a.r.
  Then, \(\AA'\) simply simulates \(\algo\) using identifiers given by the distance-\(k\) coloring and properly solves \(\problem\) by the hypotheses.
  Notice that \(\{Y_v\}_{v \in V}\), that is, the random process induced by combining \(\{X_v\}_{v \in V}\) and \(\AA'\) is a \(\myO{f(\maxDeg)}\)-dependent distribution on \(G\) by \cref{lemma:fin-dep:combination} with the required invariance properties, and we get the thesis by choosing \(k = T\).
\end{proof}   %
\section{\boldmath Simulation of \onlinelocal in \slocal for rooted trees}
\label{sec:olocal-slocal-simulation}

In this section, we show how to turn an \onlinelocal algorithm solving an LCL problem in rooted forests into a deterministic \slocal algorithm solving the same problem.
More concretely, we prove the following theorems:
\begin{restatable}{theorem}{restateThmOlocalSlocalSim}
  \label{thm:olocal-slocal-simulation}
  Let $\Pi$ be an LCL problem with degree constraint $\maxDeg$, input label set $\inLabels$, output label set $\outLabels$, and checking radius \(r > 0\).
  In addition, let $\algoA$ be an \onlinelocal algorithm solving $\Pi$ with
  locality $T(n)$ over rooted forests.
  Then the following holds:
  \begin{enumerate}
    \item If \(\algoA\) is deterministic, then there exists a deterministic
    \slocal algorithm solving $\Pi$ with locality $O(r) + T(2^{O(n^3)})$.
    \item If \(\algoA\) is randomized and has success probability \(p(n) > 0\),
    then there exists a deterministic \slocal algorithm solving $\Pi$ with
    locality $O(r) + T(2^{O(n^3)} + 2^{O(2^{n^2})} \cdot \log \frac{1}{p(n)})$.
  \end{enumerate}
\end{restatable}

\cref{thm:olocal-slocal-simulation} implies that any randomized \olocal algorithm with locality \(o(\log \log \log n)\) solving an LCL \(\problem\) over rooted trees can be turned into an \(\slocal\) algorithm solving \(\problem\) with locality \(o(\log n)\).
Next theorem shows that, over rooted trees, the class of LCL problems with complexity \(o(\log n)\) is the same as that of LCL problems with complexity \(O(1)\).

\begin{restatable}{theorem}{thmolocalsimreductionslocal}\label{thm:olocal-sim:reduction-slocal}
  Let \(\AA\) be an \slocal algorithm solving an LCL problem \(\problem\) that runs with locality \(o(\log_{\maxDeg} n)\) over rooted forests of maximum degree \(\maxDeg\) and \(n\) nodes.
  Then, there exists an \slocal algorithm \(\BB\) solving \(\problem\) with locality \(O(1)\).
\end{restatable}

It is folklore that any \(O(1)\)-round \slocal algorithm solving an LCL can be turned into an \(O(\log* n)\)-round \local algorithm solving the same LCL, implying the thesis of \cref{thm:intro-rolocal-simulation}.

\subsection{\boldmath Amnesiac \onlinelocal algorithms}

We start by formalizing what an \onlinelocal algorithm sees when run for a fixed number of steps on a graph.
We then define formally what we mean by amnesiac algorithms:
\begin{definition}[Partial \onlinelocal run of length $\ell$]
  Let $G$ be a graph with an ordering of nodes $v_1, v_2, \dots, v_n$.
  Consider the subgraph~$G_\ell \subseteq G$ induced by the radius-$T$ neighborhoods of the first $\ell$ nodes $v_1, \dots, v_\ell$.
  We call $(G_\ell, (v_1, \dots, v_\ell))$ the \emph{partial \onlinelocal run of length~$\ell$ of $G$} as this is exactly the information that an \onlinelocal algorithm would know about $G$ when deciding the output for node~$v_\ell$.

  We denote by \(({G}_\ell, (v_1, \dots, v_\ell))[v_\ell]\)
  the pair
  \((\bar{G}_\ell, (w_1, \dots, w_k))\) where \(\bar{G}_\ell \subseteq G_\ell\)
  is the connected component containing \(v_\ell\), and \((w_1,\dots, w_k)\) is
  the maximal subsequence of \((v_1, \dots, v_\ell)\) of nodes that belong to
  \(\bar{G}_\ell\). In such case, \(k\) is the length of
  {
    \((G_\ell, (v_1, \dots, v_\ell))[v_\ell]\).
  }
  Trivially, \(w_k = v_\ell\).
  
  If $({G}_\ell, (v_1, \dots, v_\ell))[v_\ell] = ({G}_\ell, (v_1, \dots,
  v_\ell))$, then we say that $(G_\ell, (v_1, \dots, v_\ell))$ is a
  \emph{partial one-component \onlinelocal run of length~$\ell$ of $G$}.
  In such case, we call $v_\ell$ the \emph{center} of the neighborhood.

  Finally, we denote by $G_\ell^-$ the disjoint set of partial one-component \onlinelocal runs formed by $v_1, \dots, v_{\ell-1}$, that is, one step shorter than $G_\ell$; this is exactly the information that an \onlinelocal algorithm knows before labeling node $v_\ell$ when $(G_\ell, (v_1, \dots, v_\ell))$ is a partial one-component \onlinelocal run.
\end{definition}

\begin{definition}[Amnesiac algorithm]
  Let $\algoA$ be a deterministic \onlinelocal algorithm, and let us fix the number of nodes to be \(n\).
  We say that \(\algoA\) is \emph{amnesiac} if the following condition is met:
  \begin{displayquote}
    Consider any two graphs \(G,H\) of \(n\) nodes and any adversarial orderings of the nodes \((v_1, \dots, v_n)\) (for \(G\)) and \((u_1, \dots, u_n)\) (for \(H\)). 
    Fix any pair \((i,j) \in [n]^2\) of indices such that
    \begin{align*}
      (G_i, (v_1, \dots, v_i))[v_i] &= (\bar{G}_i, (v_1', \dots, v_k')) \text{ and } \\
      (H_j, (u_1,\dots, u_j))[u_j] &= (\bar{H}_i, (u_1', \dots, u_k'))
    \end{align*}
    and are isomorphic (with the isomorphism being
    order-preserving, i.e., bringing \(v_h'\) into \(u_h'\), for all \(h \in
    [k]\)).
    Then the outputs of \(u_i\) and \(v_j\) coincide.
  \end{displayquote}
  
  Intuitively, $\algoA$ is amnesiac if the output of $\algoA$ for some
  node $v$ depends only on the local connected component of the partial online
  local run and not anything else the algorithm has seen. 
\end{definition}

In general, an \onlinelocal algorithm uses global memory and hence it is not amnesiac.

\subsection{Online to amnesiac}

We are now ready to generalize the intuition we provided in
\cref{ssec:overview-olocal-slocal-simulation} for turning any \onlinelocal
algorithm into an amnesiac algorithm with the cost of blowing up its locality.

\begin{lemma}
  \label{lemma:olocal-amnesiac}
  Let $\Pi$ be an LCL problem with degree constraint $\maxDeg$, input label set $\inLabels$ and output label set $\outLabels$.
  Let $\algoA$ be a deterministic \onlinelocal algorithm solving $\Pi$ over a family \(\FF\) of graphs that is closed under disjoint graph union and node/edge removal.
  Assume \(\algoA\) has locality $T(n)$ for an \(n\)-node graph.
  Then there exists an amnesiac \onlinelocal algorithm $\algoA'$ that solves $\Pi$ on an \(n\)-node graph in \(\FF\) with locality $T(2^{O(n^3)})$.
\end{lemma}
\begin{proof}
  Let $T = T(N)$ be the locality of the following experiment; we determine the size~$N$ of the experiment graph later.
  Let $\GG_\ell$ be the set of all partial one-component \onlinelocal runs of length~$\ell$ for $n$-node graphs in \(\FF\) (up to isomorphisms).
  We denote $g_\ell = |\GG_\ell|$ and remark that $g_\ell \le g = 2^{n^2} \abs*{\inLabels}^n$. 
  Notice that the term \(2^{2^n}\) takes into account not only all possible topologies (up to isomorphisms), but also all possible orderings of the nodes:
  indeed, there are at most \(2^{n^2}\) directed graphs of \(n\) nodes where the vertices are labelled with elements of \([n]\) (which defines an ordering).
  We set $N_\ell = (3\abs*{\outLabels} n g)^{n - \ell + 1} n$.

  We now adaptively construct an experiment graph~$H$ in $n$ phases.
  In the first phase, we construct $N_1$ copies of all graphs in $\GG_1$, that is, we construct $N_1g_1$ disjoint components describing all possible isolated neighborhoods an \onlinelocal algorithm might see.
  We then show the centers of these regions to $\algoA$ and let it label them.
  As there are {$N_1$} copies of each partial online local
  run of length 1, for each of these, say $\TT$, there exists a canonical label
  $\sigma_\TT$ that appears at least $N_1 / \abs*{\outLabels}$ times, by the
  pigeonhole principle.
  We say that such \onlinelocal runs of length 1 are \emph{good}, and ignore the rest.

  In each subsequent phase~$\ell$, we again construct $N_\ell$ copies of all graphs in
  $\GG_\ell$, with the catch that, for each graph $(G, (v_1, \dots, v_\ell)) \in
  \GG_\ell$ and for each component of $G^-$, we take the corresponding \emph{good} partial one-component \onlinelocal run from a previous phase; we defer arguing why such good runs exist for now.
  Next we show node $v_\ell$ to $\algoA$ and let it label the node.
  Finally, again by the pigeonhole principle, for each partial one-component
  \onlinelocal run $\TT$ (of length $\ell$) we find a canonical label
  $\sigma_\TT$ that appears at least $N_\ell / \abs*{\outLabels}$ times, and mark
  the corresponding runs as good.
  
  Now our new \onlinelocal algorithm~$\algoB$ works as follows:
  Before even processing the first node of the input, $\algoB$ runs the
  above-described experiment by simulating $\algoA$ on $H$ with locality $T$.
  Notice that the nodes of the experiment graph will be processed according to a \emph{global ordering} which is \emph{locally consistent} with that of the nodes in the single partial one-component \onlinelocal runs.
  We do not care about the specific global ordering of the nodes.

  When processing node~$v$, algorithm~$\algoB$ looks at partial one-component \onlinelocal run around $v$ and finds a corresponding good and unused run from the experiment graph~$H$.
  Algorithm~$\algoB$ labels $v$ with the good label that it found in $H$ and marks the neighborhood in $H$ \emph{used}.

  In each step, algorithm~$\algoB$ can find such unused good neighborhoods as after the whole input has been processed, what the algorithm has seen corresponds to a partial \onlinelocal run, and in particular each component of this run corresponds to a partial one-component \onlinelocal run in $H$.
  There are at most $n$ such components.
  By construction, there are at least $N_n$ copies of each type of partial one-component runs in $H$.
  At least $N_n / \abs*{\outLabels} \gg n$ of those are good, and hence the algorithm can find good and unused neighborhoods for each of the final partial one-component \onlinelocal runs.
  But as each good run in $H$ is formed by combining only good components in $H$, there must also be a good and unused partial one-component \onlinelocal run corresponding to each processed node~$v$.
  This construction ensures that $\algoB$ is amnesiac.

  Now, by contradiction, assume that \(\algoB\) doesn't solve \(\problem\).
  Since \(\problem\) is an LCL, there must be a node \(v\) of the input graph such that the whole output in its radius-\(r\) neighborhood, with \(r\) being the checking radius, is not admissible.
  However, the output of \(\algoB\) in this neighborhood is the output of \(\algoA\) from the experiment graph in a partial one-component \onlinelocal run, which must be correct by the hypothesis.

  The only thing left to do is to argue that we can always find a good \onlinelocal run in the construction and to compute the final size of the experiment graph, and hence the locality of $\algoB$.
  We start with the former:
  Consider phase~$\ell$.
  All subsequent phases contain $\sum_{k = \ell+1}^n N_k g_k$ graphs, the simulation uses at most $n$ graphs, and each of those may use at most $n$ previously-used components.
  Hence, nodes from phase~$\ell$ can be used at most
  \begin{equation*}
    n + n\sum_{k = \ell+1}^n N_k g_k
    \le n + ng \sum_{k = \ell+1}^n N_k
    \le n + \frac{N_\ell}{2\abs*{\outLabels}}
  \end{equation*}
  As $n \ll N_\ell$, we get that
  \begin{equation*}
    \frac{N_\ell}{\abs*{\outLabels}} \ge n + n\sum_{k = \ell+1}^n N_k g_k ,
  \end{equation*}
  that is, the future phases will not run out of components.

  Finally, the size of the experiment graph is
  \begin{equation*}
    N
    = n \sum_{\ell=1}^n N_\ell g_\ell
    \le N_0
    = n (3\abs*{\outLabels}gn)^{n+1}
    = n 2^{O(n^3)} = 2^{O(n^3)},
  \end{equation*}
  and it belongs to \(\FF\) as \(\FF\) is closed under disjoint graph
  union and node/edge removal.
\end{proof}

The next lemma shows how to generalize the idea behind \cref{lemma:olocal-amnesiac} all the way up to randomized \onlinelocal.

\begin{lemma}\label{lemma:rand-olocal-amnesiac}
  Let $\Pi$ be an LCL problem with degree constraint $\maxDeg$, input label set $\inLabels$ and output label set $\outLabels$.
  Let $\algoA$ be a randomized \onlinelocal algorithm solving $\Pi$ over a
  family \(\FF\) of graphs that is closed under disjoint graph union and
  node and edge removals.
  Assume \(\algoA\) has locality $T(n)$ and success probability \(p(n) > 0\) for an \(n\)-node graph.
  Then there exists an amnesiac \onlinelocal
  algorithm $\algoA'$ that solves $\Pi$ on a graph in \(\FF\) with \(n\) nodes
  with locality $T( 2^{O(n^3)} + 2^{O(2^{n^2})} \cdot \log \frac{1}{p(n)} )$.
\end{lemma}

Note the locality is such that, if $p(n)$ converges to $1$ fast enough, then we
only have a $2^{\poly(n)}$ blowup in the locality (instead of a
doubly-exponential one).

\begin{proof}
  The main idea is to adapt the proof of \cref{lemma:olocal-amnesiac}.
  The main challenge is that the construction of the simulation graph in the
  proof of \cref{lemma:olocal-amnesiac} was adaptive to the output of the
  deterministic \onlinelocal algorithm; however, the adversary in the randomized
  \onlinelocal model is oblivious to the randomness of the algorithm.
  In fact only a $q' = 1/\abs*{\outLabels}$ fraction of components in each phase
  of the proof of \cref{lemma:olocal-amnesiac} are good, and we cannot
  adaptively choose to use only those.

  To combat this, we modify the experiment: instead of adaptively choosing only the good components and discarding the rest, we sample the components from the previous phases uniformly at random without replacement.
  We later show that the sampled component is good with probability at least $q = q'/2$.
  In the worst case, each component at layer \(N_\ell\) may use up to $n$ partial one-component \onlinelocal runs from the previous layers \(N_1, \dots, N_{\ell-1}\), each of which is sampled to be good with probability only $q$.
  As the construction is $n$ levels deep (i.e., there are \(n\) layers \(N_1,
  \dots, N_n\)) and the whole construction needs to work for all up to $g =
  2^{n^2} \abs*{\inLabels}^n$ different partial one-component \onlinelocal runs, the
  total probability of finding a good label in all \(N_\ell\), \(\ell \in [n]\),
  is at least $q^{g n^2}$.

  This probability is minuscule, but we can boost it by increasing the size of the
  construction by considering
  \[
    k = 1 + \ceil*{q^{-{2^{n^2} \abs*{\inLabels}^n n^2}} \log\frac{1}{p(n)}}
  \] 
  disjoint copies of the simulation graph independently.
  Letting \(N\) be the size of a single simulation graph, \(kN\) is then the
  size of the whole experiment.
  Since $N = 2^{O(n^3)}$ in \cref{lemma:olocal-amnesiac}, we have
  \begin{equation*}
    kN = 2^{O(n^3)} + 2^{O(2^{n^2})} \log \frac{1}{p(n)}.
  \end{equation*}
  The probability that each simulation graph does not contain all the necessary good labels in all layers is at most
  \[
    \left( 1 - q^{{gn^2}} \right)^k
    \le e ^{-kq^{{g n^2}}}
    = e ^{-kq^{{2^{n^2} \abs*{\inLabels}^n n^2}}}
    < p(n).
  \]
  Since the success probability of $\algoA$ is $p(n)$, by the inclusion-exclusion principle, we get positive probability that
  at least one simulation graph in the experiment contains all the necessary good
  labels and, in addition, $\algoA$ does not fail on it.

  We show that the probability of sampling a good component is indeed at
  least $q = 1/(2\abs*{\outLabels})$.
  Consider the components in phase~$\ell$.
  For each component type, there are at least $N_\ell / \abs*{\outLabels}$ that are good.
  By the calculations in the proof of \cref{lemma:olocal-amnesiac}, future layers will use at most $N_\ell / (2\abs*{\outLabels})$ of them.
  Hence, after each step, at least $1 / (2\abs*{\outLabels}) = q$ fraction of
  the components are still good and unused, as desired.

  We finish the proof by showing that there indeed exists an amnesiac algorithm $\algoA'$ as in the claim.
  As argued above, our experiment succeeds with positive probability.
  After fixing the randomness \(f: V(H) \to \{0,1\}^{\nat}\) of the algorithm
  \(\algoA\), we obtain a uniquely defined deterministic \onlinelocal algorithm
  \(\algoA[f]\) that behaves exactly as \(\algoA\) when given the random bit
  string \(f\) as input.
  In particular, for such \(f\), the following holds:
  \begin{enumerate}
    \item There exists a simulation graph with the necessary number of good labels in all layers.
    \item \(\algoA[f]\) does not fail on the aforementioned simulation graph.
  \end{enumerate}
  Now the amnesiac algorithm \(\algoB\) works as follows:
  Since there are finitely many possible simulation graphs and finitely many
  behaviors of \(\algoA[f]\) (for all possible \(f\)), \(\algoB\) executes a
  preprocessing phase in which it goes over all possible simulation graphs and
  tries all possible deterministic \onlinelocal algorithms (according to some
  ordering, e.g., lexicographical in the description) until it finds the pair
  with the properties 1 and 2 above.
  We take such pair and use the labeling for that component as the base for our
  amnesiac algorithm.
  Proceeding as in the proof of \cref{thm:olocal-slocal-simulation}, we then
  obtain an amnesiac algorithm that has locality 
  \[
    T\left(
      2^{O(n^3)} + 2^{O(2^{n^2})} \log\frac{1}{p(n)}
    \right). \qedhere
  \]
\end{proof}

\begin{remark}
  Note that the proofs of \cref{lemma:olocal-amnesiac,lemma:rand-olocal-amnesiac} hold also if we consider labeling problems that are checkable in connected components (and not only checkable within constant-radius neighborhoods), such as component-wise leader election.
  However, for our purposes we focus just on LCLs.
\end{remark}

\subsection{\boldmath From \olocal to \slocal}

We first show how the \slocal model can nicely ``partition'' a rooted forest.
The idea is not new and comes from \cite[Section 7]{chang23automata-theoretic} where it is directly applied to the \local model; 
we present here an adaptation to \slocal.
The definition of rooted forest is given in \cref{sec:fin-dep}.
For a rooted forest \(G = (V,E)\) and each node \(v \in V\), we denote by \(G_v\) the subtree of \(G\) that is rooted at \(v\).
A root-to-leaf path \(v_0v_1\dots v_k\), where \(v_i \in V\) for all \(i \in \{0,1,\dots,k\}\), is a path of a rooted tree such that \((v_i, v_{i-1})\) is an edge for all \(i \in [k]\): we say that the path \emph{starts} at \(v_0\) and \emph{ends} at \(v_k\).

\begin{definition}[\((\alpha,\beta)\)-clustering of rooted trees]\label{def:online-to-slocal:clustering}
  Let \(G = (V,E)\) be a rooted forest. 
  An \((\alpha,\beta)\)-clustering of \(G\) is a subset \(\LL \subseteq V\) of \emph{leader nodes} that contains the root and is such that, for each \(v \in \LL\), the following properties are met:
  \begin{enumerate}
    \item \(G_v\) does not contain elements of \(\LL\) at levels \(1, \dots, \alpha - 1\).
    \item Each maximal oriented path in \(G_v\) contains exactly one element \(u
    \in \LL\) such that \(\dist_{G_v} (u,v) \in [\alpha, \beta]\), unless the
    length of the maximal oriented path is at most \(\beta - 1\), in which case
    there is at most one element of \(\LL\) in the path (and possibly none).
  \end{enumerate}
  A \emph{cluster} is a maximally connected component of non-leader nodes.
  A \emph{closed cluster} is a cluster together with its adjacent leader nodes.
  See \cref{fig:slocal-clustering} for an example.
\end{definition}

\begin{figure}
  \centering
  \resizebox{.7\textwidth}{!}{%
    \includegraphics{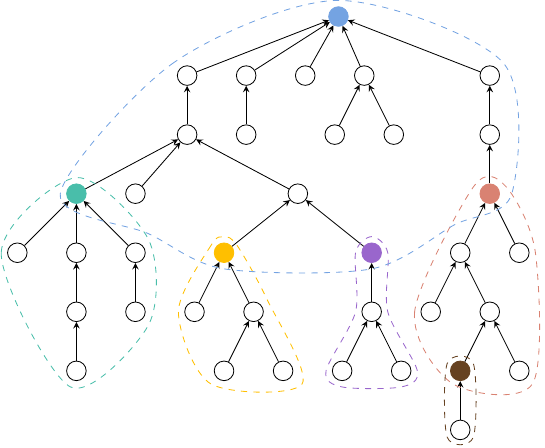}
  }
  \caption{A $(3,4)$-clustering of a rooted tree.
  The leader nodes are colored and their closed clusters marked with their respective
  color.}
  \label{fig:slocal-clustering}
\end{figure}

We will combine many \slocal algorithms:
it is easy to prove that the combination of two \(\slocal\) algorithms with localities \(T_1,T_2\) gives an \slocal algorithm with locality \(O(T_1 + T_2)\) \cite[Lemma 2.3]{ghaffari2017}.

\begin{lemma}\label{lemma:online-to-slocal:clustering}
  Let \(\alpha \in \natPos\), and let \(G = (V,E)\) be a rooted forest.
  There is an \slocal algorithm with locality $O(\alpha)$ that produces an \((\alpha-1,2\alpha + 1)\)-clustering of \(G\).
\end{lemma}
\begin{proof}
  We combine some \slocal algorithms.
  Consider first the following algorithm \(\algoA_1\):
  Suppose a node \(v \in V\) is picked by the adversary and asked to commit to
  something.
  For each root-to-leaf path \(v_0v_1\dots v_{2\alpha-1}\) (with \(v_0 = v\))
  that \(v\) can distinguish using locality \(2\alpha-1\), in parallel, \(v\)
  precommits that the node \( v_{\alpha-1}\) is the
  \emph{leader} of the path \(v_{\alpha-1} v_{\alpha} \dots v_{2\alpha-1}\)  unless
  there is another precommitment on the path within distance \(\alpha\) from~\(v\).
  
  Next we define an algorithm \(\algoA_2\) that takes as input a rooted forest
  labelled by \(\algoA_1\):{
  Each node \(v\) looks at its radius-\(\alpha\) neighborhood and checks if there is some other (unique) node that precommitted for \(v\) in the neighborhood. 
  If so, it becomes a leader and stores the path that is under its leadership, otherwise it does nothing.
  }

  The final algorithm \(\algoA_3\) takes as input a rooted forest labelled by \(\algoA_2\) and is defined as follows:
  Each node \(v\) becomes a leader
	  if and only if it belongs to an oriented root-to-leaf path \(v_0v_1\dots v_{\alpha-1}\), with \(v_{\alpha-1} = v\) led by \(v_0\). 
  All other nodes do not output anything, except the root, which becomes a leader.

  We now prove that the \slocal algorithm \(\algoA\) providing the clustering is the composition of \(\algoA_1\), \(\algoA_2\), and \(\algoA_3\), and has locality \(O(\alpha)\).
  
  Notice that the root is a leader thanks to \(\algoA_3\).
  Property 1 in \cref{def:online-to-slocal:clustering} is satisfied due to \(\algoA_3\) as well.
  Furthermore, for each maximal root-to-leaf path, consecutive leader nodes must be within distance at most \(2\alpha + 1\) between each other, unless the path ends with a leaf at distance at most \(2\alpha\) from the last leader.
  In fact, if the distance between consecutive leaders is at least \(2\alpha + 2\), or there is only one leader and the path is longer than \(2\alpha\) from the last leader, there would be at least one non-leader node that does not see any leader within distance  \(\alpha\) in the path, which is impossible due to how precommitments are done in  \(\algoA_1\).
\end{proof}

We can now prove
\cref{thm:olocal-slocal-simulation,thm:olocal-slocal-simulation}, which we
restate here for the reader's convenience:
\restateThmOlocalSlocalSim*

\begin{proof}
  Let $G$ denote the input graph, which is a rooted forest, and assume \(\algoA\) is a deterministic \onlinelocal algorithm.
  We can apply \cref{lemma:olocal-amnesiac} to get an amnesiac algorithm~$\algoA'$ that solves $\Pi$ with locality $T'(n) = T(2^{O(n^3)})$.
  We now show how to get an \slocal algorithm \(\algoB\) solving the problem
  with roughly the same locality. 
  \(\algoB\) is the composition of four
	  different \slocal algorithms $\algoB_1, \algoB_2,\algoB_3, \algoB_4$, each of them having
  locality \(O(T')\):
  \begin{enumerate}
    \item We first use \cref{lemma:online-to-slocal:clustering} with \(\alpha =
    10T' + 4r\), where \(r\) is the checking radius of \(\Pi\), and obtain an
    algorithm \(\algoB_1\) that outputs an \((\alpha-1,2\alpha+1)\)-clustering
    of $G$ with locality~\(O(\alpha)\).
    \item Now consider an algorithm \(\algoB_2\) with locality \(2\alpha+4r\)
    that takes as input $G$ as labelled by \(\algoB_1\) and works as follows:
    When processing a leader node \(v\), it collects the topology of the
    radius-\(2r\) neighborhood of the set composed of the closest leaders \(v\) sees
    in each root-to-leaf path.
    Then, \(v\) runs \(\algoA'\) in this neighborhood and
    {precommits} a solution to the LCL for the nodes in such neighborhood.
    If \(v\) is the root, it also presents its own radius-\(2r\) neighborhood to \(\algoA'\) and precommits a solution for the whole neighborhood.
    Observe that all such neighborhoods are disjoint by construction of the \((\alpha-1,2\alpha+1)\)-clustering.
    \item \(\algoB_3\) also has locality $2\alpha+4r$ and just makes all nodes
    whose label has been precommitted by some other node actually output the
    precommitted label. 
    (The output of $\algoB_2$ guarantees there are no conflicts to be resolved.)
    \item Finally, \(\algoB_4\) has again locality \(2\alpha+4r\) and just
    brute-forces a solution in each cluster.
    The solution is guaranteed to exist because \(\algoA'\) was run in all
    disjoint neighborhoods and works correctly.
    In fact, \(\algoB_4\) can just continue calling \(\algoA'\) on each cluster it sees.
  \end{enumerate}
  The combination of \(\algoB_1\), \(\algoB_2\), \(\algoB_3\), and \(\algoB_4\)
  yields a deterministic \slocal algorithm \(\algoB\) that has locality~\(O(\alpha) = O(r) + O(T'(n))\).

  The same argument applies to a randomized \onlinelocal algorithm, the only change is that now \(T'\) is \(T'(n) = T(2^{O(n^3)} + 2^{O(2^{n^2})}  \cdot \log \frac{1}{p(n)}   )\).
\end{proof}

\subsection{\boldmath From \slocal to \local}

It is folklore that all LCL problems that have complexity \(O(1)\) in \slocal belong to the complexity class \(O(\log* n)\) in \local. 
Then, to obtain \cref{thm:intro-rolocal-simulation}, it suffices to show that all LCLs with complexity \(o(\log n)\) in \slocal in rooted trees actually belong to the complexity class \(O(1)\).
We restate \cref{thm:olocal-sim:reduction-slocal}.

\thmolocalsimreductionslocal*

\begin{proof}
  Assume \(\problem\) has checking radius \(r \ge 0\), and \(T(n) = o(\log_{\maxDeg} n)\) is the locality of \(\AA\).
  Furthermore, w.l.o.g., suppose \(T(n) \ge r\) for all \(n\) large enough.
  Let \(N\) be a large enough integer such that \(T(N) = k \ge r\).
  Let \(G\) be a rooted tree of maximum degree \(\maxDeg\) with \(n\) nodes.
  We now construct a new \slocal algorithm \(\BB\) which is the composition of many \slocal algorithms.

  For the first \slocal algorithm \(\BB_1\), we make use of \cref{lemma:online-to-slocal:clustering} where \(\alpha = \floor{(\log_{\maxDeg} N - 2)/4}\).
  Hence, \(\BB_1\) yields an \((\alpha - 1, 2\alpha + 1)\)-clustering of \(G\) in time \(O(\alpha)\), with \(\alpha - 1 \ge (\log_{\maxDeg} N - 10)/4\) and \(2\alpha + 1 \le (\log_{\maxDeg} N)/2\).
  Hence, any closed cluster will have at most \(\maxDeg^{2\alpha + 1} \le \sqrt{N}\) nodes (including the adjacent leader nodes).

  Then, we consider a second \slocal algorithm \(\BB_2\) that takes \(G\) and the \((\alpha - 1, 2\alpha + 1)\)-clustering of \(G\) as input and reassigns identifiers from the set \([N]\) ``locally''.
  To better describe how \(\BB_2\) works, let us define some notation.
  For each leader node \(v\), let \(\CC_v\) denote the closed cluster where \(v\) is the leader node of minimum level, and let \(\LL_v\) be the set of leader nodes in \(\CC_v\) except for \(v\).
  Consider a partition of the nodes in \((\CC_v\setminus\NN_k(v)) \cup \NN_k(\LL_v)\) according to their distance from \(v\). 
  More specifically, \(\CC_v^{(1)}\) contains all nodes that have distance between \(k+1\) and \(\floor{(\alpha - 1)/4}\) from \(v\), while, for \(i \in \{2,3\}\), \(\CC_v^{(i)}\) contains nodes that have distance  between \(\floor{(\alpha-1)/(6 - i)} + 1\) and \(\floor{(\alpha -1)/(5-i)}\) from \(v\). 
  Finally, \(\CC_v^{(4)}\) contains all the other nodes in \(\CC_v\), which have distance at least \(\floor{(\alpha-1)/2} + 1\) from \(v\).
  Notice that \(\abs{(\CC_v\setminus\NN_k(v)) \cup \NN_k(\LL_v)} \le \maxDeg^{2\alpha + 1 + k} \le \maxDeg^{(\log_\maxDeg N)/2 + o(\log_\maxDeg N)} \le N^{2/3}\) for \(N\) large enough.
  Then, \(\BB_2\) works as follows:
  When a non-leader node is picked by the adversary, nothing happens. 
  When a leader node \(v\) is selected, it precommits identifiers for the nodes in \((\CC_v\setminus\NN_k(v)) \cup \NN_k(\LL_v)\). 
  In detail, it assigns identifiers from \(1\) to \(\floor{N/4}\) to nodes in \(\CC_v^{(1)}\), from \(\floor{N/4}+1\) to \(\floor{N/2}\) to \(\CC_v^{(3)}\), from \(\floor{N/2}+1\) to \(\floor{3N/4}\) to \(\CC_v^{(2)}\),
  from \(\floor{3N/4}+1\) to \(N\) to \(\CC_v^{(4)}\).
  Notice that, since \(\abs{\CC_v^{(i)}} \le N^{2/4}\), even \(\floor{N/5}\) distinct identifiers are enough to cover the whole region \(\CC_v^{(i)}\) when \(N\) is large enough.
  Furthermore, if \(v\) is the root of the whole graph, it precommits distinct identifiers from the set \(\{\floor{3N/4}+1, \dots, N\}\) for the nodes that have distance at most \(k\) from \(v\) (\(v\) itself included). 

  The third and last \slocal algorithm \(\BB_3\) takes as input the whole rooted tree with the clustering given by \(\BB_1\) and the output of leader nodes given by \(\BB_2\), and computes the solution to the problem as follows:
  Every node \(u\) (that does not belong to \(\CC_v\) where \(v\) is the root of the graph) looks at its two closest leader nodes \(v_u^{(1)}\) (the closest one) and \(v_u^{(2)}\) (the second closest one) that are ancestors (if \(u\) is a leader, it still looks up in the tree for two different leader ancestors), and collects all the information regarding identifiers for \(G[\NN_k(\CC_{v_u^{(1)}})]\).
  In this way, \(u\) can locally simulate \(\AA\) on a ``virtual graph'' \(G_u = G[\NN_k(\CC_{v_u^{(1)}})]\) with \(\abs{\NN_k(\CC_{v_u^{(1)}})} \le N\) nodes.
  The adversarial ordering according to which \(\AA\) processes \(G_u = G[\NN_k(\CC_{v_u^{(1)}})]\) is induced by the ordering of the identifiers.
  Notice that it might still be that \(\AA\) processes two nodes at the same time but, due to the assignment of identifiers done by \(\BB_2\), the outputs of these nodes can be determined independently of each other.
  Then, \(u\) outputs whatever \(\AA\) outputs on \(u\) according to this simulation.
  If \(u\) is the root of the whole graph, then it already knows all the identifiers needed to run \(\AA\) in \(G[\NN_k(\CC_u)]\).
  If \(u\) is in \(\CC_v\) and \(\CC_v\) is the cluster of the root node of the graph, then by accessing the memory of \(v\) it gets access to all the information it needs to compute its output.

  Notice that each \(\BB_i\) has locality \(O(\alpha)\), hence the whole algorithm \(\BB\) obtained by the composition of \(\BB_1\), \(\BB_2\), and \(\BB_3\) has locality \(O(\alpha)\) by \cite[Lemma 2.3]{ghaffari2017}.
  Observe also that \(O(\alpha)  = O(\log_\maxDeg N) = O(1)\).
  
  Also, the final output of \(\BB\) is correct: Suppose, by contradiction, that there is a failure somewhere, that is, a node \(v\) could detect its neighborhood \(\NN_r(v)\) as non-admissible.
  Suppose \(v\) does not belong to \(\CC_u\) where \(u\) is the root of the graph, and denote its two closest ancestors that are leader nodes by \(v_u^{(1)}\) (the closest one) and \(v_u^{(2)}\) (the second closest one).
  We have two cases. 
  If \(u \in \CC^{(4)}_{v_u^{(2)}} \cup \CC^{(1)}_{v_u^{(1)}} \cup \CC^{(2)}_{v_u^{(1)}}\), then its output must be correct: it is completely determined by \(G[\NN_k(\CC_{v_u^{(1)}} \cap (\CC^{(4)}_{v_u^{(2)}} \cup \CC^{(1)}_{v_u^{(1)}} \cup \CC^{(2)}_{v_u^{(1)}}))]\) which contains only distinct identifiers (because it is a subgraph of \(G[\CC^{(4)}_{v_u^{(2)}} \cup \CC^{(1)}_{v_u^{(1)}} \cup \CC^{(2)}_{v_u^{(1)}} \cup \CC^{(3)}_{v_u^{(1)}}]\)) and hence can be extended to a valid input for \(\AA\) with \(N\) nodes (choosing identifiers arbitrarily in \([N^2]\setminus [N]\) and extending the adversarial processing order respecting the ordering induced by identifiers). 
  In the second case, \(u \in \CC^{(3)}_{v_u^{(1)}} \cup \CC^{(4)}_{v_u^{(1)}}\). 
  Also in this case the output of \(u\) must be correct: it is completely determined by \(G[\NN_k(\CC_{v_u^{(1)}} \cap (\CC^{(3)}_{v_u^{(1)}} \cup \CC^{(4)}_{v_u^{(1)}}))]\) which contains only distinct identifiers (because it is a subgraph of \(G[\CC^{(2)}_{v_u^{(1)}} \cup \CC^{(3)}_{v_u^{(1)}} \cup \CC^{(4)}_{v_u^{(1)}}]\)) and hence can be extended to a valid input for \(\AA\) with \(N\) nodes (choosing identifiers arbitrarily in \([N^2]\) and extending the adversarial processing order respecting the ordering induced by identifiers). 
  An analogous argument holds for nodes in \(\CC_v\) where \(v\) is the root node of the graph.
\end{proof}   %

\section{\boldmath Randomized \onlinelocal with an adaptive adversary}
\label{sec:adaptive-rolocal-derandomization}

This section aims at showing that \rolcl with an adaptive adversary is as strong as its deterministic counterpart.
The proof is very similar to the proof of \citet{dist_deran} for derandomizing
\local in \slocal and bases on the method of conditional expectations.
The results can be extended to any subclass of graphs.

\begin{theorem}
  \label{thm:olocal_adversary_derandomization}
  Let $\Pi$ be any labeling problem, and let $\algoA$ be an \onlinelocal algorithm with an adaptive adversary solving $\Pi$ with locality $T(n)$ and probability $p > 1 - 1/n$.
  Then there exists a deterministic \onlinelocal algorithm~$\algoA'$ solving $\Pi$ with locality $T(n)$.
\end{theorem}
\begin{proof}
  Let $\fA$ be an algorithm in the \onlinelocal model with an adaptive adversary and with round complexity $T(n)$.
  We construct a deterministic \onlinelocal algorithm $\fA'$ with the same round complexity $T(n)$ as follows:
  The algorithm $\fA'$ simulates $\fA$ with certain fixed random-bit-strings.
  In particular, whenever new nodes arrive at the input, the algorithm $\fA'$ fixes the random bits in all newly arrived nodes in the way we describe below.
  
  We view a run of an \onlinelocal algorithm with an adaptive adversary as a game between an algorithm~$\fA$ and an adversary.
  In each step of that game, first the adversary chooses a set of nodes that arrive at the input (one special node and its neighborhood).
  Then the algorithm $\fA$ samples the random bits of the newly arrived nodes and fixes the label of the arrived node.
  After $i$ steps of this game, we use $\fF$ to denote the event that the adversary wins the game, that is, algorithm~$\fA$ produces an invalid output according to problem~$\Pi$.
  By definition, we have $\pr{\fF} \le 1/n$.
  Here and from now on, the probability expressions assume that the adversary maximizes the probability of $\fF$ happening. 

  Our algorithm $\fA'$ in each $i$-th step fixes the random bits of the newly arrived nodes in a way that guarantees that 
  \[
    \pr{\fF \mid B_1 = b_1, \dots, B_i = b_i} \le \pr{\fF \mid B_1 = b_1, \dots, B_{i-1} = b_{i-1}}.
  \]
  Here, $B_j = b_j$ stands for fixing the values of random bits of nodes that arrived in the $j$-th step to $b_j$. We note that this can always be done since, by definition, we have 
  \[
    \pr{\fF \mid B_1 = b_1, \dots, B_{i-1} = b_{i-1}} = \expect_{B_i}\left[  \pr{\fF \mid B_1 = b_1, \dots, B_{i} = b_{i}}\right] .
  \]
  Moreover, the \onlinelocal algorithm importantly both knows all the past arrived nodes and their randomness, and can also consider all the different choices of the adversary in the future, and hence can compute the above conditional probabilities.
  
  At the end of the simulation of $\fA$, we have $\pr{\fF \mid B_1 = b_1, \dots, B_{n} = b_{n}} < 1$, and hence the bad event $\fF$ does not occur. This implies that the algorithm $\fA'$ is always correct. 
\end{proof}

\begin{remark}
  This proof assumes that we have white-box access to the \rolcl algorithm and that it uses finitely many random bits.
  If either of these assumptions does not hold, we can use a more inefficient derandomization along the lines of \cite{dahal_et_al:LIPIcs.DISC.2023.40}.
  Here, we just observe that an \onlinelocal algorithm working against an adaptive adversary implies the existence of a function $f$ that maps all possible input sequences on all possible $n$-node graphs that can be produced by the adversary to a valid output of the problem $\Pi$.
  The deterministic \onlinelocal algorithm finds such a function $f$ at the beginning of its execution and produces the output according to $f$.
\end{remark}
\section{\boldmath Lower bound for 3-coloring grids in randomized \onlinelocal}\label{sec:3-col-grid-lb}

In this section, we will show that $3$-coloring $(\sqrt{n}\times\sqrt{n})$-grids
in \rolcl requires $\Omega(\log n)$-locality:

\begin{theorem}
  \label{thm:three-col-ro-lcl}
  The locality of a \rolcl algorithm for solving $3$-coloring in $(\sqrt{n}\times\sqrt{n})$-grids is $\Omega(\log n)$.
\end{theorem}

We will base our proof on the recent result of \citet{chang23_tight_arxiv},
where the authors show a similar lower bound for the deterministic \onlinelocal
model. 
We will start with a brief overview of the techniques used in
\cite{chang23_tight_arxiv}.

\subsection{Relevant ideas from \cite{chang23_tight_arxiv}}

The lower bound in \cite{chang23_tight_arxiv} is based on the idea that any
coloring of a grid $G$ with the three colors $\{1,2,3\}$ can be partitioned into
\emph{regions} with colors $1$ and $2$ that are separated by \emph{boundaries}
of color $3$.
Formally, a region is a maximal connected component that is colored only with
colors $1$ and $2$, while a boundary is a maximal connected component in $G^2$
that is colored only with color $3$.
It is crucial to observe that the boundaries of color $3$ are not always
compatible and that they have parities of their own (see \cref{fig:b-value}).
Indeed, the core idea of the proof is to constrict many incompatible boundaries
of color $3$ to a small space, thus requiring a view of $\Omega(\log n)$ to
resolve them. 

\begin{figure}
  \centering
  \includegraphics{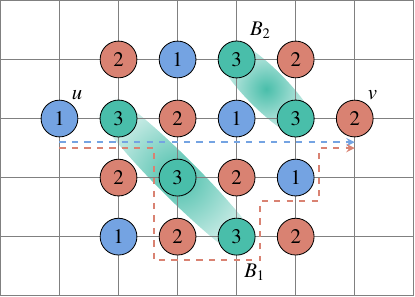}
  \caption{Example of a $3$-colored grid with two boundaries $B_1$ and $B_2$
  that have different parities.
  As one can see, $b(u,v) = 2$:
  No matter if we take the blue or the red path from $u$ to $v$, we always cross
  $B_1$ and then $B_2$.}
  \label{fig:b-value}
\end{figure}

Using the same terminology as \cite{chang23_tight_arxiv}, we count incompatible
boundaries between two points by using so-called \emph{$a$- and $b$-values}.
The $a$-value is defined as an edge weight between any two nodes and captures
the change of colors $1$ and $2$.

\begin{definition}[$a$-value~\cite{chang23_tight_arxiv}]
  Given a directed edge $(u,v)$ and a $3$-coloring of the nodes $c:V\to
  \{1,2,3\}$, we define 
  \[
    a(u,v) = \begin{cases} 
               c(u)-c(v), & \text{if}\ c(u)\neq 3\ \text{and}\ c(v)\neq 3  \\
               0, & \text{otherwise}.
    \end{cases}
  \]
\end{definition}

Observe that the $a$-value of any directed $4$-cycle in a grid is equal to $0$. Using the $a$-value, we define the $b$-value of a path.

\begin{definition}[$b$-value~\cite{chang23_tight_arxiv}]
  For a directed path $P$, its $b$-value is defined as
  \[
    b(P) = \sum_{(u,v)\in P} a(u,v).
  \]
\end{definition}

On a high level, the $b$-value of a path describes the cumulative total of
incompatible boundaries along this path.
Note we say ``cumulative'' because in this count boundaries of the same
parity cancel each other out.
Indeed, if we consider the $b$-value of a simple directed cycle in a grid, then
it must be equal to $0$:

\begin{lemma}[$b$-value of a cycle is zero~\cite{chang23_tight_arxiv}]
  \label{lem:b-value-cycle}
  Let $C$ be a directed cycle in $G$.
  Then $b(C) = 0$.
\end{lemma}

As an example, observe that a path that starts and ends with color $3$ and
otherwise is colored with colors $1$ and $2$ always has a $b$-value of $0$ or
$1$.
In particular, the $b$-value is $1$ if the distance between the nodes of color
$3$ is even (and thus the boundaries are incompatible). 
Meanwhile, a path that goes through two incompatible regions has a total
$b$-value of $2$. 
Nevertheless, a path going through two boundaries that are compatible has a
total $b$-value of $0$.

\begin{lemma}[Parity of the $b$-value~\cite{chang23_tight_arxiv}]
  \label{lem:b-value-parity}
Let $P$ denote any directed path of length $\ell$ that starts in node $u$ and ends in node $v$ in a grid. Then, 
the parity of $b(P)$ is
$$b(P) \equiv \beta(u) + \beta(v) + \ell \Mod 2$$
where $\beta$ is an indicator variable stating whether a node is of color $3$ or
not:
\[
  \beta(u) = \begin{cases} 
  1, & \text{if}\ c(u)= 3  \\
  0, & \text{otherwise}.
  \end{cases}
\]
\end{lemma}

Observe that the parity of the $b$-value of a path is determined by the colors
of the endpoints of this path.

\subsection{\boldmath From deterministic to randomized \onlinelocal}

Suppose that we are given a \rolcl algorithm with visibility radius $T = o(\log
n)$.
Our goal is to prove the algorithm fails to solve the 3-coloring problem.
Using Yao's minimax principle, we show how to construct an (oblivious)
adversarial distribution of inputs so that any deterministic \olcl algorithm
fails with noticeable probability.

As in \cite{chang23_tight_arxiv}, the lower bound consists of two main steps:
\begin{enumerate}
  \item First we show how to force paths to have an arbitrarily large $b$-value.
  Generally we cannot force a large $b$-value between the path's endpoints (or
  in fact between any two fixed nodes), but we do get the guarantee that it
  contains two nodes $v_1$ and $v_2$ where $b(v_1,v_2)$ is large.
  Note the location of $v_1$ and $v_2$ is completely unknown to us since we are
  working with an oblivious adversary.
  The construction is inductive:
  Given a procedure that generates a path with $b$-value $\ge k-1$ with
  probability $\ge 1/2$, we show how to generate a path with $b$-value $\ge k$
  also with probability $\ge 1/2$.
  (Cf.\ the construction in \cite{chang23_tight_arxiv} with an adaptive
  adversary, which succeeds every time.)
  Since we invoke the construction for $k-1$ a constant number of times, we are
  able to obtain any desired $b$-value of $k = o(\log n)$ with high probability.
  Note it is logical that we cannot do much better than this as it would
  contradict the existing $O(\log n)$ \detolcl algorithm.

  \item The second step is to actually obtain a contradiction.
  (See \cref{fig:proof-15}.)
  First we construct a path $P_1$ with large $b$-value, say $\gg 4T$, between
  nodes $u_s$ and $v_t$.
  To get the contradiction, we wish to place two nodes $w_s$ and $w_t$ next to
  $u_s$ and $v_t$, respectively, so that the four node cycle has positive
  $b$-value (which is impossible due to \cref{lem:b-value-cycle}).
  If we had an adaptive adversary as in \cite{chang23_tight_arxiv}, this would
  be relatively simple:
  Since the adaptive adversary knows the location of $u_s$ and $v_t$, it just
  constructs an arbitrary path $P_2$ with the same size as $P_1$, picks any two
  $w_s$ and $w_t$ that are at the same distance from each other as $u_s$ and
  $v_t$, mirrors $P_2$ if needed to obtain a non-negative $b$-value, and then
  places this at minimal distance to $P_1$.
  We show that, allowing for some failure probability, we do not need any
  information about $u_s$ and $v_t$ (i.e., nor their location nor the distance
  between them) in order to obtain the same contradiction.
\end{enumerate}

We now proceed with the proof as outlined above.
Accordingly, the first step is the following:

\begin{lemma}\label{lem:large-b-value}
  Given any $k = o(\log n)$, there is an adversarial strategy to construct a
  directed path of length $n^{o(1)}$ with a $b$-value of at least $k$ against
  any \olcl algorithm with locality $T = o(\log n)$. 
  Moreover, this strategy succeeds with high probability.
\end{lemma}

\begin{proof}
  Consider the following recursive construction:
  \begin{itemize}
    \item If $k = 0$, create a single node with previously unrevealed nodes
    all around it (inside the visibility radius~$T$).
    \item Otherwise, repeat the following steps four times in total:
    \begin{itemize}
      \item Create four distinct paths $P_1,P_2,P_3,P_4$ by following the
      procedure for $k-1$.
      \item Toss independent fair coins $c_1,c_2,c_3 \in \{ 0,1 \}$.
      \item Connect the $P_i$'s horizontally aligned and in order while placing
      $c_i+1$ nodes between paths $P_i$ and $P_{i+1}$.
    \end{itemize}
    Let $Q_1,Q_2,Q_3,Q_4$ be the four paths created by this procedure (each
    having their own four $P_i$'s coming from the procedure in the previous iteration $k-1$). Next, we connect
    the $Q_i$'s horizontally aligned with each other and in arbitrary order.
    Note that we need to place an additional node between each pair of $Q_i$'s. Otherwise, we would have to have revealed the edge
    between the two endpoints too early to the algorithm.
  \end{itemize}
  Thus, for $k \ge 1$ we have $16$ invocations of the procedure for $k-1$ in
  total.
  We argue that, with probability at least $1/2$, the path yielded by this
  procedure contains a segment with $b$-value at least $k$.
  By repeating this procedure $O(\log n)$ times independently, we obtain at least
  one path with the desired property with high probability.

  To prove that the recursive construction works, we proceed by induction.
  Fix $k \ge 1$ and suppose that the procedure for $k-1$ succeeds with
  probability $p \ge 1/2$ at yielding a segment with $b$-value at least $k-1$.
  Let us first consider $Q_1$.
  Let $X_i$ be a random variable that is $1$ if this occurs for path $P_i$ and
  $0$ otherwise.
  Then we have
  \begin{align*}
    \Pr[\sum_{i=1}^4 X_i < 2]
    &= \Pr[\sum_{i=1}^4 X_i = 0] + \Pr[\sum_{i=1}^4 X_i = 1] \\
    &= (1-p)^4 + 4p(1-p)^3 \\
    &= (1-p)^3(1+3p) \\
    &\le \frac{1}{2}.
  \end{align*}
  Hence, with probability at least $1/2$ there are at least two paths $P_i$ and
  $P_j$, $i < j$, for which the construction succeeds.

  Since we are dealing with paths, we may simplify the notation and write $b(u,v)$
  for the $b$-value of the (unique) segment that starts at $u$ and ends at $v$.
  Let thus $u_i,v_i,u_j,v_j$ appear in this order in $Q_1$ and
  $\abs{b(u_i,v_i)}, \abs{b(u_j,v_j)} \ge k-1$.
  Next we will show that, conditioned on this assumption, the probability that
  $Q_1$ contains a segment with $b$-value at least $k$ is at least $1/2$.
  Hence, \emph{a priori}, $Q_1$ contains such a segment with probability at least
  $1/4$.
  Since all $Q_i$'s are constructed in the same way and independently of
  one another, the probability that at least one of the $Q_i$'s contains such a
  segment is at least $1-(1-1/4)^4 > 1-1/e > 1/2$.

  Let us write $\sigma(x)$ for the sign function of $x$ (i.e., $\sigma(x) = 1$
  if $x > 0$, $\sigma(x) = -1$ if $x < 0$, and $\sigma(0) = 0$).
  To see why the above holds for $Q_1$, consider the two following cases:
  \begin{description}
    \item[Case 1: \boldmath$\sigma(b(u_i,v_i)) = \sigma(b(u_j,v_j))$.]
    Using \cref{lem:b-value-parity}, we have that $b(v_i,u_j) \equiv \beta(v_i)
    + \beta(u_j) + c_i + \cdots + c_{j-1} \pmod 2$.
    Since $\beta(v_i)$ and $\beta(u_j)$ are fixed, and the coin tosses are
    independent, we have $\abs{b(v_i,u_j)} \nequiv k-1 \pmod 2$ with probability
    $1/2$.
    Assuming this holds, we have either $\abs{b(v_i,u_j)} \ge k$, in which case
    we are done, or $\abs{b(v_i,u_j)} \le k-2$.
    Since $b(u_i,v_i)$ and $b(u_j,v_j)$ have the same sign, we have 
    \begin{align*}
      \abs{b(u_i,v_j)}
      &= \abs{b(u_i,v_i) + b(v_i,u_j) + b(u_j,v_j)} \\
      &\ge \abs{b(u_i,v_i) + b(u_j,v_j)} - \abs{b(v_i,u_j)} \\
      &\ge 2(k-1) - (k-2) \\
      &=k.
    \end{align*}
    \item[Case 2: \boldmath$\sigma(b(u_i,v_i)) \neq \sigma(b(u_j,v_j))$.]
    Arguing by using \cref{lem:b-value-parity} as before, we obtain that
    $\abs{b(v_i,u_j)} \nequiv 0 \pmod 2$ holds with probability $1/2$.
    Assuming this is the case, we have thus $\abs{b(v_i,u_j)} \ge 1$.
    Without restriction, let $\sigma(b(u_i,v_i)) = \sigma(b(v_i,u_j))$.
    Then 
    \[
      \abs{b(u_i,u_j)}
      = \abs{b(u_i,v_i) + b(v_i,u_j)}
      \ge k-1 + 1
      = k.
    \]
  \end{description}

  Finally, let us confirm that the construction fits into the $(\sqrt{n} \times
  \sqrt{n})$-grid.
  We only reveal at most $2T + 1 = o(\sqrt{n})$ in a column, so we need to only
  consider nodes along a row.
  The initial path contains $m_0 = 2T + 1$ visible nodes and in the $i$-th
  recursive step we have $m_i \le 16p_{i-1} + 27$ visible nodes in total for $i
  \ge 1$.
  (Inside each $Q_i$ we need at most $2(4-1) = 6$ additional nodes to join the
  four $P_i$'s, and to join the four $Q_i$'s we need an additional $3$ nodes.)
  Solving the recursion, in the $k$-th step we have thus
  \[
    m_k = \frac{1}{5} \left( 2^{4k+1}(5T+7) - 9 \right) = n^{o(1)}
  \]
  visible nodes along the path since $k, T = o(\log n)$.
\end{proof}

We now illustrate the idea for getting the contradiction previously described.
(See \cref{fig:proof-15}.)
Invoking \cref{lem:large-b-value}, we obtain a path $P_1$ with a $b$-value of
$4T+4$ between two nodes $u_s$ and $v_t$ (whose positions are unknown to us).
Letting $L$ be the length of $P_1$, we (arbitrarily) create a second path $P_2$
of length $10L$ and then randomly choose how to align $P_1$ and $P_2$.
More specifically, letting $u$ be the first node in $P_1$, we choose some node
$w$ of $P_2$ uniformly at random and align $u$ and $w$; then we mirror $P_2$
with probability $1/2$ and reveal it at distance $2T + 2$ to $P_1$ (which is
consistent with all previously revealed nodes).

The reason why this works is the following:
Since $P_1$ has length $L$ and $P_2$ length $10L$, with at least $4/5$ we align
the paths so that each node in $P_1$ has a matching node underneath it in $P_2$.
Let $w_s$ and $w_t$ be the nodes matching $u_s$ and $v_t$, respectively.
Then because we mirror $P_2$ with probability $1/2$, we get that
$b(w_s,\dots,w_t)$ is at least zero in expectation (conditioned on having
properly aligned the two paths).
Hence, with probability at least $2/5$ we obtain a cycle
$(u_s,\dots,v_t,\dots,w_t,\dots,w_s,\dots,u_s)$ where $b(u_s,\dots,v_t) > 4T+4$,
$b(w_t,\dots,w_s) \ge 0$, and $b(v_t,\dots,w_t), b(w_s,\dots,u_s) \ge -2T-2$
(due to the distance between the two paths).

\begin{proof}[Proof of \cref{thm:three-col-ro-lcl}]
  Observe that, using \cref{lem:large-b-value}, we can construct a path
  $P_1 = (u_0,\dots,u_L)$ of length $L = n^{o(1)}$ where some segment
  $(u_s,\dots,u_t)$ of $P_1$ has a $b$-value of $k>4T+4$.
  Next we construct a path $P_2 = (v_0,\dots,v_{10L})$ of length $10L$ that will
  be placed below $P_1$.
  We assume that the points of the path are revealed to the algorithm in some
  predefined order.

  The position and orientation of $P_2$ are chosen as follows:
  \begin{enumerate}
    \item Choose a node $v_r \in [L,9L]$ of $P_2$ uniformly at random.
    This node is placed below the node $u_0$ of $P_1$ at $2T+2$ distance from
    it.
    \item Throw a fair coin $m \in \{ 0,1 \}$.
    If $m = 1$, mirror $P_2$ along the vertical axis that goes through $u_0$ and
    $v_r$.
  \end{enumerate}
  The nodes of $P_2$ are revealed to the algorithm in the same predefined and
  possibly mirrored order.
  See \cref{fig:proof-15} for an example.

  \begin{figure}
    \centering
    \includegraphics[scale=1.2]{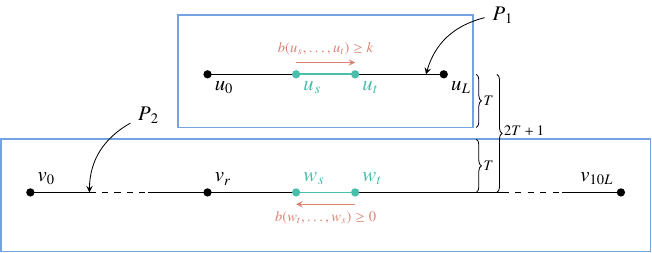}
    \caption{How to turn paths with large $b$-value into a contradiction.
    Here the blue area includes the nodes revealed so far around the path
    segments $(u_0,\ldots,u_L)$ and the corresponding part of $P_2$ underneath
    it. The green segments are the two segments we consider in the proof. A
    cycle going through both paths leads to a contradiction.}
    \label{fig:proof-15}
  \end{figure}

  We next prove that we obtain the desired lower bound.
  First notice that, since we pick $v_r \in [L,9L]$, every node in $P_1$ has a
  counterpart in $P_2$ whether we mirror $P_2$ or not.
  Let $(w_s,\dots,w_t)$ denote the segment matched to $(u_s,\dots,u_t)$.
  Consider the case where either of the following is true:
  \begin{itemize}
    \item $m = 0$ and $r \in [L,9L-s-t]$
    \item $m = 1$ and $r \in [L+s+t,9L]$
  \end{itemize}
  Denoting by $E$ the event in which this occurs, note we have
  \[
    \Pr[E] = \frac{1}{2}\Pr[r \in [L,9L-s-t]] + \frac{1}{2}\Pr[r \in [L+s+t,9L]]
    = \frac{8L - s - t + 1}{9L + 1}
    \ge \frac{6L + 1}{9L + 1}
    > \frac{2}{3}.
  \]
  Now conditioned on $E$, notice that every segment $(v_x,\dots,v_y)$ where $x
  \in [L+s,9L-s]$ and $y \in [L+t,9L-t]$ has equal probability of being matched
  with $(u_s,\dots,u_t)$ either in the same direction (i.e., $w_s = v_x$ and
  $w_t = v_y$) or reversed (i.e., $w_s = v_y$ and $w_t = v_x$).
  Hence,
  \[
    \Pr[b(w_t,\dots,w_s) \ge 0 \mid E] \ge \frac{1}{2}.
  \]
  (In fact, the probability is exactly $1/2$ if $b(w_t,\dots,w_s) > 0$ and $1$
  if $b(w_t,\dots,w_s) = 0$.)
  Thus, the probability that $b(w_t,\dots,w_s) \ge 0$ is $> (2/3) \cdot (1/2) =
  1/3$.

  From here on, we proceed as in~\cite{chang23_tight_arxiv}.
  By \cref{lem:b-value-cycle}, the cycle
  $(u_s,\ldots,u_t,\ldots,w_t,\ldots,w_s,\ldots,u_s)$ must have a $b$-value of
  $0$.
  However, recall that the $b$-value of a path is bounded by its length by
  definition.
  In our case, $-2T-2< b(u_t,\ldots,w_t) <2T+2$, $-2T-2< b(w_s,\ldots,u_s)
  <2T+2$, $b(u_s,\ldots,u_t) \ge k$, and $b(w_t,\ldots,w_s)\ge 0$.
  In order for the $b$-value of the cycle to be $0$, we would need to have
  $2(2T+2) \ge k$, which is a contradiction.
  Since this occurs with noticeable probability (i.e., $> 1/3$), the claim
  follows.
\end{proof}
\section{\boldmath LCL problems in paths and cycles in randomized \onlinelocal}\label{sec:cycles}

In this section, we show that the complexity of LCLs in paths and cycles is either \(\myO{1}\) or \(\myTheta{n}\) in \rolcl (with an oblivious adversary).

  \begin{theorem}
    Let $\Pi$ be an LCL problem on paths and cycles (possibly with inputs).
    If the locality of $\Pi$ is $T$ in the \rolcl model, then its locality is $O(T + \log*n)$ in the \local model.
  \end{theorem}
  This theorem is a simple corollary of the following two lemmas:
  \begin{lemma}
    \label{lemma:rolocal-speedup}
    Let $\Pi$ be an LCL problem on paths and cycles (possibly with inputs), and let $\algoA$ be a \rolcl algorithm solving $\Pi$ with locality $o(n)$.
    Then there exists a \rolcl algorithm $\algoA'$ solving $\Pi$ with locality $O(1)$.
  \end{lemma}
  \begin{lemma}
    \label{lemma:rolocal-simulation}
    Let $\Pi$ be an LCL problem on paths and cycles (possibly with inputs), and let $\algoA$ be a \rolcl algorithm solving $\Pi$ with locality $O(1)$.
    Then there exists a \local algorithm $\algoA'$ solving $\Pi$ with locality $O(\log* n)$.
  \end{lemma}

By previous results~\cite{akbari_et_al:LIPIcs.ICALP.2023.10}, it is known that,
in the case of paths and cycles, any (deterministic) \olcl algorithm with
sublinear locality can be sped up to an \olcl algorithm with constant locality.
We show how to extend \cite[Lemma 5.5]{akbari_et_al:LIPIcs.ICALP.2023.10}.
We start by constructing a large virtual graph $P'$ such that, when the original
algorithm runs on the virtual graph $P'$, the labeling produced by the algorithm
is locally compatible with the labeling in the original graph $P$.

\begin{proof}[Proof of \cref{lemma:rolocal-speedup}]
  Let $\Pi$ be an LCL problem in paths or cycles with checking-radius $r$, let
  $P$ be a path or a cycle with $n$ nodes, and let $\algoA$ be a \rolcl
  algorithm solving $\Pi$ with locality $T(n) = o(n)$.
  In the proof of \cite[Lemma 5.5]{akbari_et_al:LIPIcs.ICALP.2023.10}, the
  authors use three phases to speed up (deterministic) \olcl with sublinear
  locality to only constant locality.
  We use the same strategy to construct a \rolcl algorithm $\algoA'$ for
  solving $\Pi$ with constant locality.
  The phases for constructing \rolcl algorithm $\algoA'$ are as follows:
\begin{enumerate}
  \item In the first phase, the algorithm deterministically creates an $(\alpha,
  \alpha)$-ruling set $R$ for the path $P$, mirroring the approach outlined in
  the proof of \cite[Lemma 5.5]{akbari_et_al:LIPIcs.ICALP.2023.10}.
  \item In the second phase, the algorithm constructs a larger virtual path $P'$
  with $N$ nodes and simulates algorithm $\algoA$ on $P'$ in the neighborhoods
  of nodes in $R$.
  Intuitively this path $P'$ is constructed based on the pumping-lemma-style
  argument on LCL problems presented by \citet{doi:10.1137/17M1157957}, and it
  only relies on (the definition of) problem $\Pi$, not the algorithm $\algoA$.
  During this simulation, the algorithm encounters a failure only if $\algoA$
  fails within specific neighborhoods of $P'$. Consequently, the algorithm's
  success in this phase aligns with the success rate of $\algoA$, ensuring a
  high probability of success.
  \item The third phase is also as in \cite[Lemma
  5.5]{akbari_et_al:LIPIcs.ICALP.2023.10}: The algorithm extends the fixed
  labels around $R$ to the entire path $P$ by applying brute force for nodes
  outside the neighborhood of $R$. Note that this extension is feasible because
  of the way the path $P'$ is created and the pumping-lemma-style argument as
  discussed in \cite{akbari_et_al:LIPIcs.ICALP.2023.10}.
\end{enumerate}
In the end, we compose these three phases together to obtain the \rolcl
algorithm $\algoA'$ for solving $\Pi$ with constant locality.
\end{proof}

\begin{proof}[Proof of \cref{lemma:rolocal-simulation}]
  This proof closely follows the argument presented in \cite[Lemma
  5.6]{akbari_et_al:LIPIcs.ICALP.2023.10}.
  Let $\Pi$ denote an LCL problem with a constant checking-radius $r$, and
  suppose $\algoA$ is a \rolcl algorithm solving $\Pi$ with constant locality
  $T$.
  Define $\beta = T + r + 1$.
  As in \cite[Lemma 5.6]{akbari_et_al:LIPIcs.ICALP.2023.10}, we consider an
  input-labeled graph $P$ formed by many copies of all feasible input
  neighborhoods with a radius of $\beta$.
  We can depict $P$ as a collection of disjoint path fragments that we later
  connect to each other and create a long path with.
  Each of these path fragments has size $2\beta + 1$, so the node in the center
  of each segment (i.e., the ($\beta+1$)-th node in the segment) has the same
  view (up to radius $\beta$) as in the final path.
  Let $V$ be the set of central nodes within these fragments.

  We apply $\algoA$ to each node within the radius-$r$ neighborhood of the nodes
  in $V$ following an arbitrary order and then terminate.
  These nodes have the same view as in the final path because their distance to
  the endpoints is at most $\beta - r = T + 1$.
  We iterate the process multiple times, yielding a distribution of output
  labels around the central nodes.
  Given that the size of $P$ is constant, there exists an output labeling $L$ of
  the nodes occurring with probability $p = \Omega(1)$.
  If needed, we can augment $P$ with (constantly many) additional nodes such
  that $|P| > 1/p$.

  We set this labeling $L$ as the deterministic output for graph $P$ and proceed
  similarly to the deterministic case by constructing the canonical labeling $f$
  from $L$.
  Then, we use the function $f$ to construct a \local algorithm with $O(\log*
  n)$ locality following the same steps as in the proof of
  \cite[Lemma~5.6]{akbari_et_al:LIPIcs.ICALP.2023.10}.
  It is important to highlight that it is feasible to fill any gap of sufficient
  length between parts labeled with the canonical labeling:
  If the algorithm $\algoA$ would never produce a valid labeling for this gap,
  then the algorithm would fail to label $P$ with probability at least $p$ and,
  since $p \geq 1/|P|$, the algorithm $\algoA$ would not succeed with high
  probability.
  \end{proof}

\ifanon
\else
    \section*{Acknowledgments}
    We thank anonymous reviewers for their helpful feedback on previous versions of this work.
    This work was supported in part by the Research Council of Finland, Grants
    333837, 359104 and 363558.
    Xavier Coiteux-Roy acknowledges support from the BMW Endowment Fund, from a
    Postdoc.Mobility fellowship of the Swiss National Science Foundation (SNSF), and from the
    Natural Sciences and Engineering Research Council of Canada (NSERC).
    Francesco d'Amore is supported by the project Decreto MUR n. 47/2025,
    CUP: D13C25000750001.
    Darya Melnyk is supported by the European Research Council (ERC), grant
    agreement No. 864228 (AdjustNet), Horizon 2020, 2020-2025.
    Augusto Modanese is partially supported by the Helsinki Institute for
    Information Technology (HIIT).
    Marc-Olivier Renou received funding from Inria and
    CIEDS through the Action Exploratoire project DEPARTURE, and from the French National
    Research Agency (ANR) through the JCJC grant LINKS (ANR-23-CE47-0003).
    He was additionally supported by the T-ERC QNET project
    (ANR-24-ERCS-0008) and by the European Union's Horizon 2020 Research and Innovation
    Programme through QuantERA Grant Agreements Nos.~731473 and 101017733.

    \paragraph{Errata.}
    Earlier versions of this manuscript (arXiv identifiers \texttt{2403.01903v1}, \texttt{2403.01903v2}, and \texttt{2403.01903v3}) had two additional claims: 
    \begin{enumerate}
        \item In \texttt{v1}, we claimed a result analogous to
        \cref{thm:intro-rolocal-simulation} for \emph{unrooted} trees (which
        would also have implications on the complexity of the sinkless
        orientation problem). 
        \item In \texttt{v1}, \texttt{v2}, and \texttt{v3} we claimed that the
        \dlocal model can derandomize the randomized \local model similarly to
        how the \slocal model can derandomize the \local model
        \cite{dist_deran}.
    \end{enumerate}
    The proofs of these two claims had errors and we do not know whether they
    can be fixed.
    We thank Maxime Flin and Alexandre Nolin for bringing the second of these
    to our attention.
\fi

   \bibliographystyle{plainnat}
  \bibliography{quantum-local.bib}
  \appendix
  \pagebreak
\section{Quantum, bounded-dependence, and non-signaling}\label{app:dc-models}

This appendix introduces both the non-signaling model and the bounded-dependence model, which are the most powerful models that satisfy physical causality, and thus they generalize the \qlocal model.
The distinction between the non-signaling model and the bounded-dependence depends on whether shared states are available (in the non-signaling model) or not (in the bounded-dependence model).
We illustrate the difference between these models by analyzing the problem of $c=2$-coloring a lifebuoy-shaped graph (see \cref{fig:app:toy-grah-blank} and \cref{fig:app:toy-graphs-a}) with locality $T=2$.

We first introduce a circuit formalism which allows us to clarify the early works of \citet{arfaoui2014}, and \citet{gavoille2009}, by re-expressing the randomized \local and \qlocal models in this formalism (see \cref{app:subsec:RandQNSlocal}). 
Second, we introduce the concept of light-cones and the principle of non-signaling, explaining how (together with the symmetries of the graph) they define the non-signaling model (see \cref{app:subsec:causaltheory,app:subsec:causaltheory_ohne}).
Third, we show in \cref{app:sec:G-not-colorable} that lifebuoy-shaped graphs are not $2$-colorable in the non-signaling model with locality $T = 2$ through a reduction from the 2-colorability of the cheating graph $H$ shown in \cref{fig:app:toy-graphs-b}.
At last, we define the bounded-dependence model based on our circuit formalism and on the following three principles: device replication; non-signaling and independence~\cite{gisin2020constraints,coiteux2021no,coiteux2021PRA,beigi2021covariance}; and invariance under symmetries of the graph (see \cref{app:sec:no-shared-resource_mit}). 
We stress its connection to the concept of finitely-dependent distributions~\cite{aaronson1989algebraic,holroyd2016,holroyd2018,spinka2020,holroyd2024}.

\begin{figure}[h]
     \centering
              \includegraphics[width=0.5\textwidth/2]{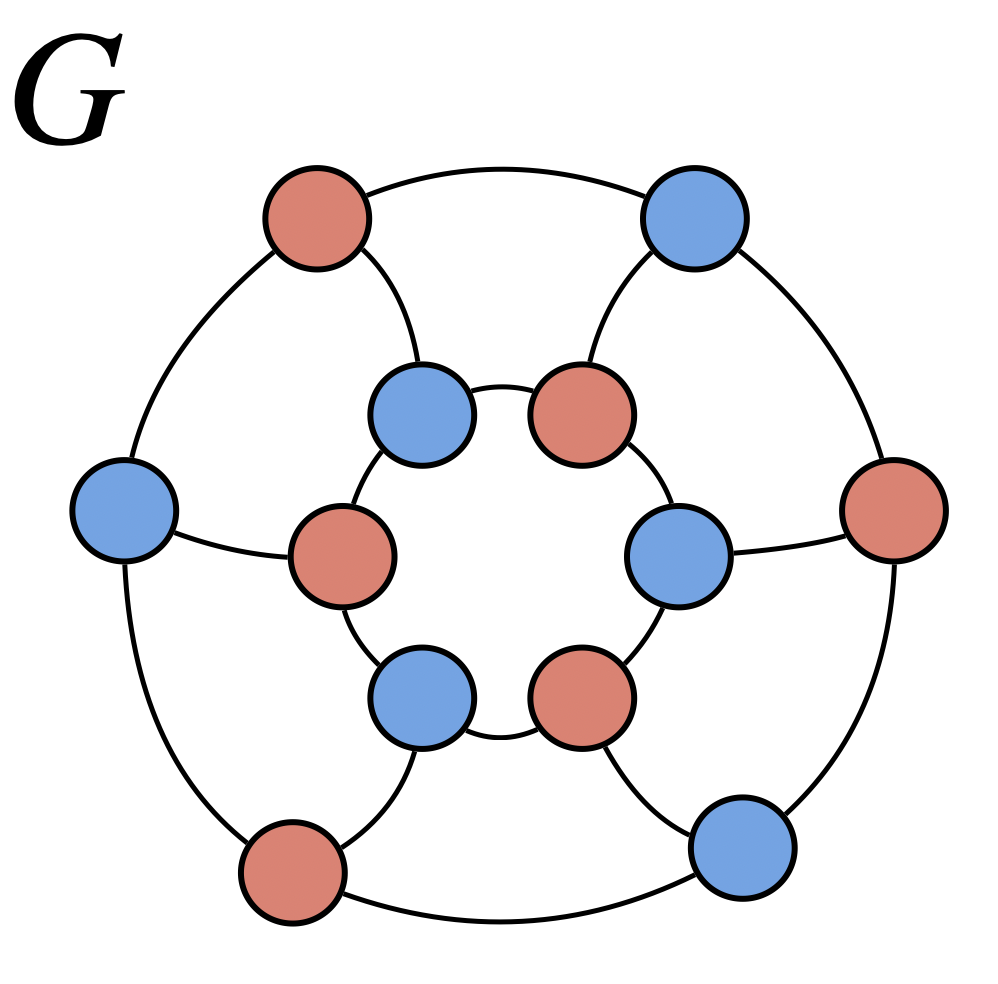}
                       \caption{An example of a \emph{lifebuoy-shaped} graph \(G\). The {lifebuoy-shaped} graphs are defined by the set $\GG$ of graphs that are isomorphic to the above.
    Notice that their chromatic number is 2. 
    In our case, the node labels required by the \local model are unique and go from 1 to 12. Granting this extra power is not restrictive as we are proving a lower bound.}
         \label{fig:app:toy-grah-blank}
\end{figure}
\begin{figure}
     \centering
          \begin{subfigure}[t]{0.39\textwidth}
         \centering
         \includegraphics[width=0.7\textwidth]{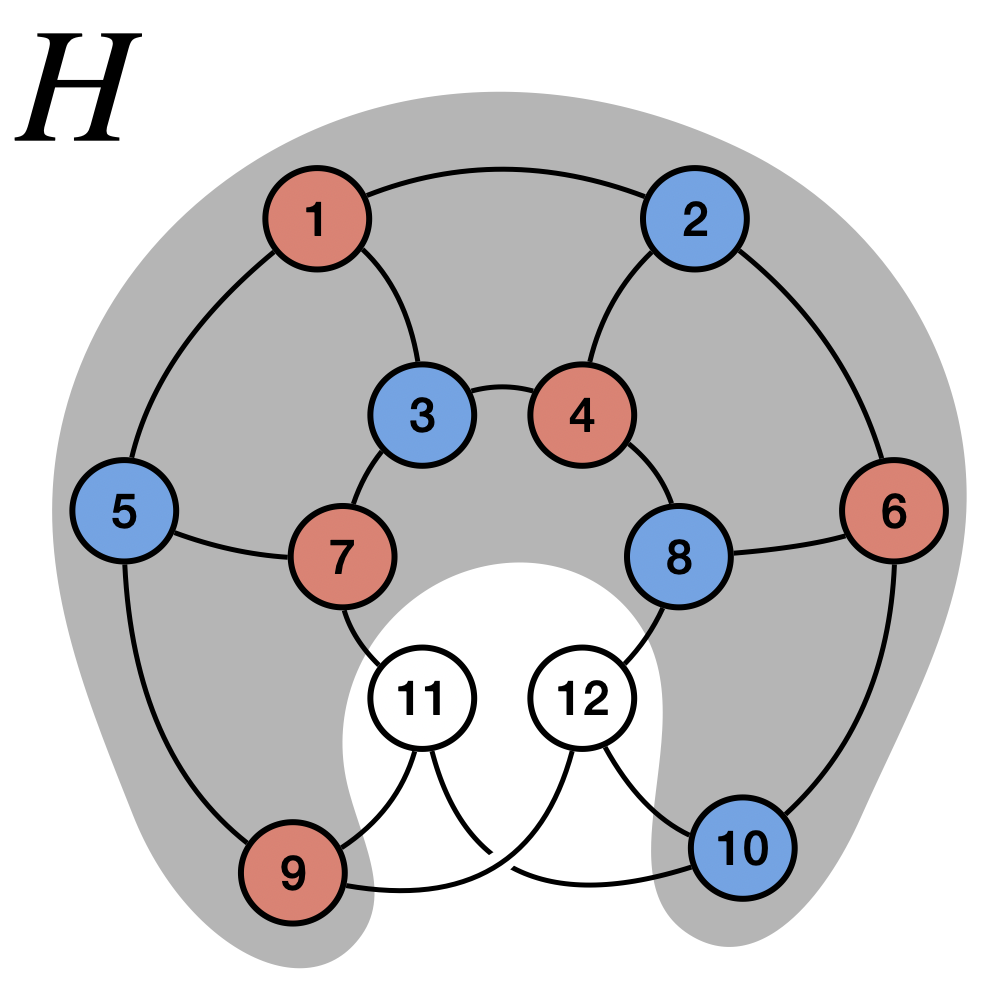}
         \caption{$H$ has chromatic number 3.}
         \label{fig:app:toy-graphs-b}
     \end{subfigure}
          \hfill
     \begin{subfigure}[t]{0.39\textwidth}
         \centering
         \includegraphics[width=0.7\textwidth]{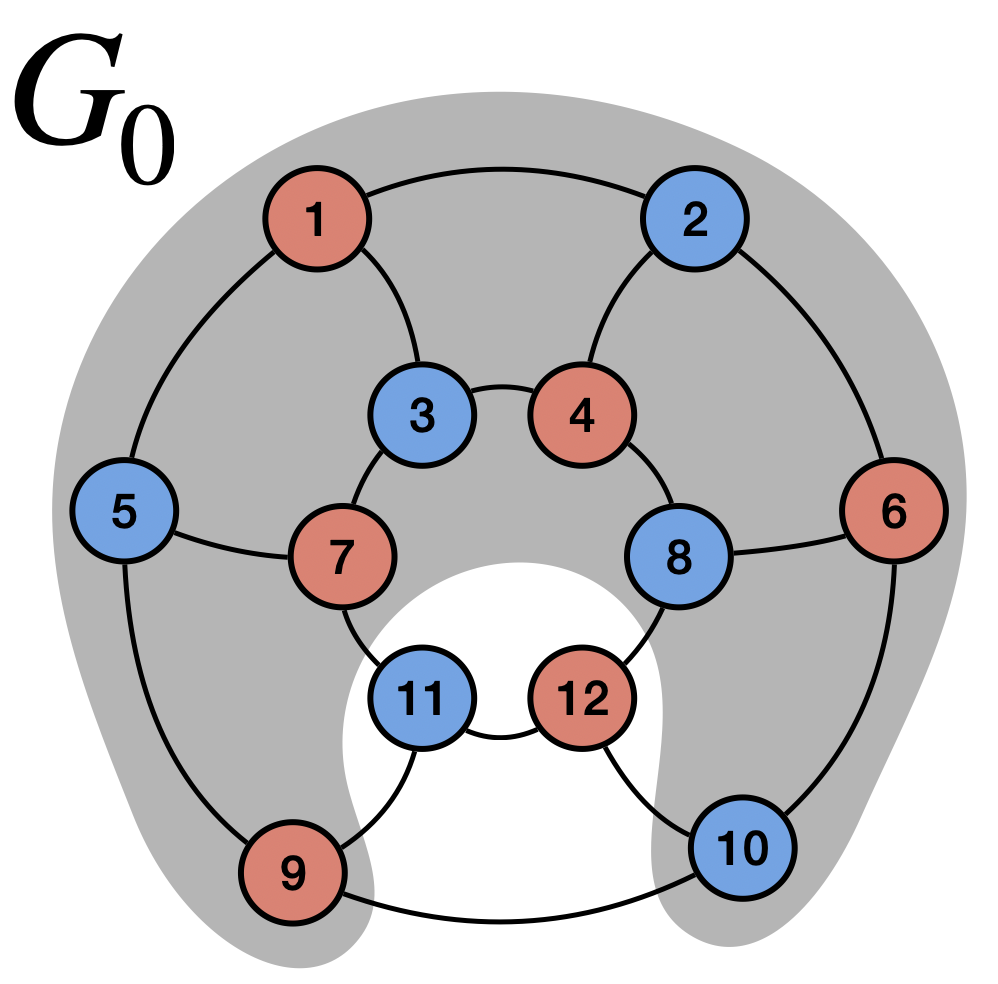}
         \caption{A distributed algorithm that finds a 2-coloring of the lifebuoy-shaped graphs $G\in\GG$ would by definition 2-color the particular instance $G=G_0$.}
         \label{fig:app:toy-graphs-a}
     \end{subfigure}
        \caption{
        The joint view of the couple of nodes $u,v=1,2$ after $T = 2$ rounds of communication is limited to the gray area.  In this region the graphs $G_0$ and $H$ are identical. The non-signaling principle implies that the outputs $(c_1,c_2)$ must therefore be identically distributed in both $G_0$ and~$H$.
}
\label{fig:app:toy-graphs}
\end{figure}

\subsection{\boldmath Randomized \local, \qlocal, and non-signaling models}\label{app:subsec:RandQNSlocal}

In the next two sections, we re-introduce the randomized \local model and the \qlocal model for coloring lifebuoy-shaped graphs in a circuit formalism.
This will later enable us to clarify the definitions of the non-signaling and bounded-dependence models.

\subsubsection{\boldmath Randomized \local model}

We consider twelve nodes with unique identifiers ranging from 1 to 12 and that are connected in a lifebuoy-shaped graph that is a priori unknown to the nodes\footnote{A reader used to standard quantum nonlocality should see the graph --- $H$ or $G\in \GG$ --- as an input of the problem, split and distributed among the local parties.} (an example of such graph is the labeled graph $G_0\in \GG$ illustrated in \cref{fig:app:toy-graphs-a}).
The nodes try to color this graph in the randomized \local model within $T=2$ steps of synchronous communication.
The difficulty is that a given node ignores which $G\in \GG$ connects the set of labelled nodes: it can only discover a local part of the structure of $G$ by communicating to its neighbors.
The most general $(T = 2)$-round strategy for the node $1$ consists of the following procedure, which alternates between randomized processing steps and communication steps (the computational power and size of exchanged messages are unbounded):
\begin{description}
     \item[Processing~0:] Sample a random real number and store it locally.
     \item[Communication~1:] Send all stored information, including the sampled random number, to all neighbors. Receive information from all neighbors and store it for subsequent rounds (in $G_0$, the neighbors are $2, 3, 5$).
     \item[Processing~1:] Process all stored information (possibly in a randomized way) and store the result.\footnote{In classical information theory, it is known that intermediate processing gates can be taken as identity gates, that is, any strategy can be simulated by a two-layer circuit where the parties send their first random number to all parties up to a distance $T$, and then make a unique processing step after all the communication has taken place. In quantum and non-signaling theories, this is not the case anymore~\cite{coiteux2023caterpillars}.}
     \item[Communication~2:] Send all stored information, including all received messages and the outputs of processing steps, to all neighbors. Receive information from neighbors and store it.
     \item[Processing~2:] Process all stored information (possibly in a randomized way) to output a color.
\end{description}

Such $T$-round strategy on the graph $G$ can be represented formally as a circuit $C_{G,T}$, such as in \cref{fig:app:two-circuits-a} where semicircles, line wires, and squares respectively represent the sampling of a random number, the transfer (or storage) of information, and the processing of information.  
Once the concrete operations performed by the nodes (i.e.\ randomness sampling and processing) are made explicit, it is possible to compute (using classical information theory) the exact output distribution of the strategy on graph $G$, that is, the probability distribution $\pr{ c_1,\dots,c_n \ \st \ C_{G,T}}$ of observing that the set of nodes $\{1,\dots,n\}$ outputs the colors $\{c_1,\dots,c_n\}$ when connected as per one of the lifebuoy-shaped graphs $G\in \GG$.
Importantly, in our model, the operations performed by the nodes cannot depend on the connection graph $G$.

A generic classical strategy with shared randomness can be in the same way represented by the general circuit of \cref{fig:app:two-circuits-b}, by initializing the circuit with a source of randomness common to all nodes (the large semicircle).

\begin{figure}
     \centering
     \begin{subfigure}[t]{0.3\textwidth}
         \centering
                  \resizebox{\textwidth}{!}{
\begin{tikzpicture}[square/.style={rectangle}]
\definecolor{myred}{RGB}{255,80,80}
\definecolor{myblue}{RGB}{255,80,80}
\definecolor{mypurple}{RGB}{255,80,80}
\definecolor{mygray}{RGB}{140,140,140}
\definecolor{myyellow}{RGB}{255,200,80}

\def\listnodes{{12,10,8,6,4,2,1,3,5,7,9,11}}
  \foreach \i in {1,...,12}{
   \pgfmathsetmacro{\j}{\listnodes[\i-1]}
     \node at (-0.5,0.5*\i) [circle,fill=white,inner sep=1] (name\j)  {\footnotesize $\j$};
     
    \node at (0,0.5*\i) [semicircle,fill={\ifnum\j=11 black!80\else \ifnum\j=12 black!80\else myred\fi\fi},rotate=90,inner sep=2] (input\j)  {};
    
    \node at (2,0.5*\i) [square,fill={\ifnum\j=11 black!80\else \ifnum\j=12 black!80\else\ifnum\j=9 black!80\else\ifnum\j=7 black!80\else\ifnum\j=8 black!80\else\ifnum\j=10 black!80\else myred\fi\fi\fi\fi\fi\fi},inner sep=2.75] (gate\j)  {};

\node at (4,0.5*\i) [square,fill={\ifnum\j=1 myred\else \ifnum\j=2 myred\else black!80\fi\fi},inner sep=2.5] (output\j)  {};

\node[text={\ifnum\j=1 myred\else \ifnum\j=2 myred\else black\fi\fi}] (out\j) [right of=output\j] {$c_{\j}$};
\draw [->,draw={\ifnum\j=1 myred\else \ifnum\j=2 myred\else black!80\fi\fi},thick] (output\j) -- (out\j);
  }

\node at (-1.25,3.5) [semicircle,rotate=90,fill=white,inner sep=7.25] (res) {};

\node at (4,-0) [square,fill=black!15,inner sep=2.75] (phantomoutput11)  {};
\node at (4,-0.5) [square,fill=black!15,inner sep=2.75] (phantomoutput9)  {};
\node at (4,0.5*13) [square,fill=black!15,inner sep=2.75] (phantomoutput12)  {};
\node at (4,0.5*14) [square,fill=black!15,inner sep=2.75] (phantomoutput10)  {};
    
\node at (2,-0) [square,fill=black!15,inner sep=2.75] (phantomgate11)  {};
\node at (2,-0.5) [square,fill=black!15,inner sep=2.75] (phantomgate9)  {};
\node at (2,0.5*13) [square,fill=black!15,inner sep=2.75] (phantomgate12)  {};
\node at (2,0.5*14) [square,fill=black!15,inner sep=2.75] (phantomgate10)  {};

\node at (0,-0) [semicircle,rotate=90,,fill=black!15,inner sep=2] (phantominput11)  {};
\node at (0,-0.5) [semicircle,rotate=90,,fill=black!15,inner sep=2] (phantominput9)  {};
\node at (0,0.5*13) [semicircle,rotate=90,,fill=black!15,inner sep=2] (phantominput12)  {};
\node at (0,0.5*14) [semicircle,rotate=90,,fill=black!15,inner sep=2] (phantominput10)  {};

\node[text=black!15] at (-0.5,0) [circle,fill=white,inner sep=1] ()  {\footnotesize $11$};
\node[text=black!15] at (-0.5,-0.5) [circle,fill=white,inner sep=1] ()  {\footnotesize $9$};
\node[text=black!15] at (-0.5,0.5*13) [circle,fill=white,inner sep=1] ()  {\footnotesize $12$};
\node[text=black!15] at (-0.5,0.5*14) [circle,fill=white,inner sep=1] ()  {\footnotesize $10$};

\draw[-,black!15,thick]   (phantominput10) -- (gate9);
\draw[-,black!15,thick]   (input9) -- (phantomgate10);

\draw[-,black!15,thick]   (phantominput12) -- (gate11);
\draw[-,black!15,thick]   (input11) -- (phantomgate12);

\draw[-,black!15,thick]   (input10) -- (phantomgate9);
\draw[-,black!15,thick]   (phantominput9) -- (gate10);

\draw[-,black!15,thick]   (input12) -- (phantomgate11);
\draw[-,black!15,thick]   (phantominput11) -- (gate12);

\draw[-,black!15,thick]   (phantomgate10) -- (output9);
\draw[-,black!15,thick]   (gate9) -- (phantomoutput10);

\draw[-,black!15,thick]   (phantomgate12) -- (output11);
\draw[-,black!15,thick]   (gate11) -- (phantomoutput12);

\draw[-,black!15,thick]   (gate10) -- (phantomoutput9);
\draw[-,black!15,thick]   (phantomgate9) -- (output10);

\draw[-,black!15,thick]   (gate12) -- (phantomoutput11);
\draw[-,black!15,thick]   (phantomgate11) -- (output12);

\def\listnodes{{
0,0,0,0,
0,0,0,0,
0,0,0,0,
0,0,0,0,
0,0,0,0,
0,0,0,0,
3,5,7,11,
4,6,8,12,
5,9,9,11,
6,10,10,12,
7,9,11,11,
8,10,12,12
}}
  \foreach \i in {11,9,7,8,10,12}{
  \foreach \j in {0,1,2,3}
   \pgfmathsetmacro{\j}{\listnodes[\i*4+\j-4]}  
\draw[-,mygray,thick]   (input\j) -- (gate\i);
}

  \foreach \i in {11,9,7,8,10,12}{
  \foreach \j in {0,1,2,3}
   \pgfmathsetmacro{\j}{\listnodes[\i*4+\j-4]}  
\draw[-,mygray,thick]   (gate\j) -- (output\i);
}

\draw[-,myred,thick]   (input1) -- (gate3);
\draw[-,myred,thick]   (input3) -- (gate3);
\draw[-,myred,thick]   (input4) -- (gate3);
\draw[-,myred,thick]   (input7) -- (gate3);
\draw[-,mygray,thick]   (gate1) -- (output3);
\draw[-,mygray,thick]   (gate3) -- (output3);
\draw[-,mygray,thick]   (gate4) -- (output3);
\draw[-,mygray,thick]   (gate7) -- (output3);

\draw[-,myred,thick]   (input1) -- (gate5);
\draw[-,myred,thick]   (input5) -- (gate5);
\draw[-,myred,thick]   (input7) -- (gate5);
\draw[-,myred,thick]   (input9) -- (gate5);

\draw[-,mygray,thick]   (gate1) -- (output5);
\draw[-,mygray,thick]   (gate5) -- (output5);
\draw[-,mygray,thick]   (gate7) -- (output5);
\draw[-,mygray,thick]   (gate9) -- (output5);

\draw[-,myblue,thick]   (input2) -- (gate4);
\draw[-,myblue,thick]   (input3) -- (gate4);
\draw[-,myblue,thick]   (input4) -- (gate4);
\draw[-,myblue,thick]   (input8) -- (gate4);
\draw[-,mygray,thick]   (gate2) -- (output4);
\draw[-,mygray,thick]   (gate3) -- (output4);
\draw[-,mygray,thick]   (gate4) -- (output4);
\draw[-,mygray,thick]   (gate8) -- (output4);

\draw[-,myblue,thick]   (input2) -- (gate6);
\draw[-,myblue,thick]   (input6) -- (gate6);
\draw[-,myblue,thick]   (input8) -- (gate6);
\draw[-,myblue,thick]   (input10) -- (gate6);
\draw[-,mygray,thick]   (gate2) -- (output6);
\draw[-,mygray,thick]   (gate6) -- (output6);
\draw[-,mygray,thick]   (gate8) -- (output6);
\draw[-,mygray,thick]   (gate10) -- (output6);

\draw[-,mypurple,thick]   (input1) -- (gate1);
\draw[-,mypurple,thick]   (input2) -- (gate1);
\draw[-,mypurple,thick]   (input3) -- (gate1);
\draw[-,mypurple,thick]   (input5) -- (gate1);

\draw[-,mypurple,thick]   (input1) -- (gate2);
\draw[-,mypurple,thick]   (input2) -- (gate2);
\draw[-,mypurple,thick]   (input4) -- (gate2);
\draw[-,mypurple,thick]   (input6) -- (gate2);

\draw[-,myred,thick]   (gate1) -- (output1);
\draw[-,myred,thick]   (gate2) -- (output1);
\draw[-,myred,thick]   (gate3) -- (output1);
\draw[-,myred,thick]   (gate5) -- (output1);

\draw[-,myblue,thick]   (gate1) -- (output2);
\draw[-,myblue,thick]   (gate2) -- (output2);
\draw[-,myblue,thick]   (gate4) -- (output2);
\draw[-,myblue,thick]   (gate6) -- (output2);

\node[right,rotate=90] at (0,7.25) (legtop1) {\footnotesize Processing 0};
\node[right,rotate=90] at (1,7.25) (legtop2) {\footnotesize Communication 1};
\node[right,rotate=90] at (2,7.25) (legtop3) {\footnotesize Processing 1};
\node[right,rotate=90] at (3,7.25) (legtop4) {\footnotesize Communication 2};
\node[right,rotate=90] at (4,7.25) (legtop5) {\footnotesize Processing 2};
\node[right,rotate=90] at (5,7.25) (legtop6) {\footnotesize Output};
\begin{scope}[on background layer]
\foreach  \i in {1,...,6}{
\coordinate (legbottom\i) at (\i-1,-0.5);
\draw [-,myyellow] (legtop\i) -- (legbottom\i);
}
\end{scope}
\end{tikzpicture} %
}
         \caption{Circuit representation of a non-signaling strategy on graph $G_0$ of \cref{fig:app:toy-graphs-a}, without a shared resource. Note the cyclicity of the circuit. 
         Highlighted in red is the past-light-cone of the joint output $(c_1,c_2)$, that is the set of gates which connects to the output gates producing $(c_1,c_2)$.
         }
         \label{fig:app:two-circuits-a}
     \end{subfigure}
     \hfill
     \begin{subfigure}[t]{0.3\textwidth}
         \centering
         \resizebox{\textwidth}{!}{
\begin{tikzpicture}[square/.style={rectangle}]
\definecolor{myred}{RGB}{255,80,80}
\definecolor{myblue}{RGB}{255,80,80}
\definecolor{mypurple}{RGB}{255,80,80}
\definecolor{mygray}{RGB}{140,140,140}
\definecolor{myyellow}{RGB}{255,200,80}

\def\listnodes{{12,10,8,6,4,2,1,3,5,7,9,11}}
  \foreach \i in {1,...,12}{
   \pgfmathsetmacro{\j}{\listnodes[\i-1]}
     \node at (-0.5,0.5*\i) [circle,fill=white,inner sep=1] (name\j)  {\footnotesize $\j$};
     
    \node at (0,0.5*\i) [semicircle,fill={\ifnum\j=11 black!80\else \ifnum\j=12 black!80\else myred\fi\fi},rotate=90,inner sep=2] (input\j)  {};
    
    \node at (2,0.5*\i) [square,fill={\ifnum\j=11 black!80\else \ifnum\j=12 black!80\else\ifnum\j=9 black!80\else\ifnum\j=7 black!80\else\ifnum\j=8 black!80\else\ifnum\j=10 black!80\else myred\fi\fi\fi\fi\fi\fi},inner sep=2.75] (gate\j)  {};

\node at (4,0.5*\i) [square,fill={\ifnum\j=1 myred\else \ifnum\j=2 myred\else black!80\fi\fi},inner sep=2.5] (output\j)  {};

\node[text={\ifnum\j=1 myred\else \ifnum\j=2 myred\else black\fi\fi}] (out\j) [right of=output\j] {$c_{\j}$};
\draw [->,draw={\ifnum\j=1 myred\else \ifnum\j=2 myred\else black!80\fi\fi},thick] (output\j) -- (out\j);
  }

\node at (-1.25,3.5) [semicircle,rotate=90,fill=myred,inner sep=7.25] (res) {};

\begin{scope}[on background layer]
\draw[-,mygray,thick]   (res) -- (input11);
\draw[-,myred,thick]   (res) -- (input9);
\draw[-,myred,thick]   (res) -- (input7);
\draw[-,mypurple,thick]   (res) -- (input5);
\draw[-,mypurple,thick]   (res) -- (input3);
\draw[-,mypurple,thick]   (res) -- (input1);
\draw[-,mypurple,thick]   (res) -- (input2);
\draw[-,mypurple,thick]   (res) -- (input4);
\draw[-,mypurple,thick]   (res) -- (input6);
\draw[-,myblue,thick]   (res) -- (input8);
\draw[-,myblue,thick]   (res) -- (input10);
\draw[-,mygray,thick]   (res) -- (input12);
\end{scope}

\node at (4,-0) [square,fill=black!15,inner sep=2.75] (phantomoutput11)  {};
\node at (4,-0.5) [square,fill=black!15,inner sep=2.75] (phantomoutput9)  {};
\node at (4,0.5*13) [square,fill=black!15,inner sep=2.75] (phantomoutput12)  {};
\node at (4,0.5*14) [square,fill=black!15,inner sep=2.75] (phantomoutput10)  {};
    
\node at (2,-0) [square,fill=black!15,inner sep=2.75] (phantomgate11)  {};
\node at (2,-0.5) [square,fill=black!15,inner sep=2.75] (phantomgate9)  {};
\node at (2,0.5*13) [square,fill=black!15,inner sep=2.75] (phantomgate12)  {};
\node at (2,0.5*14) [square,fill=black!15,inner sep=2.75] (phantomgate10)  {};

\node at (0,-0) [semicircle,rotate=90,,fill=black!15,inner sep=2] (phantominput11)  {};
\node at (0,-0.5) [semicircle,rotate=90,,fill=black!15,inner sep=2] (phantominput9)  {};
\node at (0,0.5*13) [semicircle,rotate=90,,fill=black!15,inner sep=2] (phantominput12)  {};
\node at (0,0.5*14) [semicircle,rotate=90,,fill=black!15,inner sep=2] (phantominput10)  {};

\node[text=black!15] at (-0.5,0) [circle,fill=white,inner sep=1] ()  {\footnotesize $11$};
\node[text=black!15] at (-0.5,-0.5) [circle,fill=white,inner sep=1] ()  {\footnotesize $9$};
\node[text=black!15] at (-0.5,0.5*13) [circle,fill=white,inner sep=1] ()  {\footnotesize $12$};
\node[text=black!15] at (-0.5,0.5*14) [circle,fill=white,inner sep=1] ()  {\footnotesize $10$};

\draw[-,black!15,thick]   (phantominput10) -- (gate9);
\draw[-,black!15,thick]   (input9) -- (phantomgate10);

\draw[-,black!15,thick]   (phantominput12) -- (gate11);
\draw[-,black!15,thick]   (input11) -- (phantomgate12);

\draw[-,black!15,thick]   (input10) -- (phantomgate9);
\draw[-,black!15,thick]   (phantominput9) -- (gate10);

\draw[-,black!15,thick]   (input12) -- (phantomgate11);
\draw[-,black!15,thick]   (phantominput11) -- (gate12);

\draw[-,black!15,thick]   (phantomgate10) -- (output9);
\draw[-,black!15,thick]   (gate9) -- (phantomoutput10);

\draw[-,black!15,thick]   (phantomgate12) -- (output11);
\draw[-,black!15,thick]   (gate11) -- (phantomoutput12);

\draw[-,black!15,thick]   (gate10) -- (phantomoutput9);
\draw[-,black!15,thick]   (phantomgate9) -- (output10);

\draw[-,black!15,thick]   (gate12) -- (phantomoutput11);
\draw[-,black!15,thick]   (phantomgate11) -- (output12);

\def\listnodes{{
0,0,0,0,
0,0,0,0,
0,0,0,0,
0,0,0,0,
0,0,0,0,
0,0,0,0,
3,5,7,11,
4,6,8,12,
5,9,9,11,
6,10,10,12,
7,9,11,11,
8,10,12,12
}}
  \foreach \i in {11,9,7,8,10,12}{
  \foreach \j in {0,1,2,3}
   \pgfmathsetmacro{\j}{\listnodes[\i*4+\j-4]}  
\draw[-,mygray,thick]   (input\j) -- (gate\i);
}

  \foreach \i in {11,9,7,8,10,12}{
  \foreach \j in {0,1,2,3}
   \pgfmathsetmacro{\j}{\listnodes[\i*4+\j-4]}  
\draw[-,mygray,thick]   (gate\j) -- (output\i);
}

\draw[-,myred,thick]   (input1) -- (gate3);
\draw[-,myred,thick]   (input3) -- (gate3);
\draw[-,myred,thick]   (input4) -- (gate3);
\draw[-,myred,thick]   (input7) -- (gate3);
\draw[-,mygray,thick]   (gate1) -- (output3);
\draw[-,mygray,thick]   (gate3) -- (output3);
\draw[-,mygray,thick]   (gate4) -- (output3);
\draw[-,mygray,thick]   (gate7) -- (output3);

\draw[-,myred,thick]   (input1) -- (gate5);
\draw[-,myred,thick]   (input5) -- (gate5);
\draw[-,myred,thick]   (input7) -- (gate5);
\draw[-,myred,thick]   (input9) -- (gate5);

\draw[-,mygray,thick]   (gate1) -- (output5);
\draw[-,mygray,thick]   (gate5) -- (output5);
\draw[-,mygray,thick]   (gate7) -- (output5);
\draw[-,mygray,thick]   (gate9) -- (output5);

\draw[-,myblue,thick]   (input2) -- (gate4);
\draw[-,myblue,thick]   (input3) -- (gate4);
\draw[-,myblue,thick]   (input4) -- (gate4);
\draw[-,myblue,thick]   (input8) -- (gate4);
\draw[-,mygray,thick]   (gate2) -- (output4);
\draw[-,mygray,thick]   (gate3) -- (output4);
\draw[-,mygray,thick]   (gate4) -- (output4);
\draw[-,mygray,thick]   (gate8) -- (output4);

\draw[-,myblue,thick]   (input2) -- (gate6);
\draw[-,myblue,thick]   (input6) -- (gate6);
\draw[-,myblue,thick]   (input8) -- (gate6);
\draw[-,myblue,thick]   (input10) -- (gate6);
\draw[-,mygray,thick]   (gate2) -- (output6);
\draw[-,mygray,thick]   (gate6) -- (output6);
\draw[-,mygray,thick]   (gate8) -- (output6);
\draw[-,mygray,thick]   (gate10) -- (output6);

\draw[-,mypurple,thick]   (input1) -- (gate1);
\draw[-,mypurple,thick]   (input2) -- (gate1);
\draw[-,mypurple,thick]   (input3) -- (gate1);
\draw[-,mypurple,thick]   (input5) -- (gate1);

\draw[-,mypurple,thick]   (input1) -- (gate2);
\draw[-,mypurple,thick]   (input2) -- (gate2);
\draw[-,mypurple,thick]   (input4) -- (gate2);
\draw[-,mypurple,thick]   (input6) -- (gate2);

\draw[-,myred,thick]   (gate1) -- (output1);
\draw[-,myred,thick]   (gate2) -- (output1);
\draw[-,myred,thick]   (gate3) -- (output1);
\draw[-,myred,thick]   (gate5) -- (output1);

\draw[-,myblue,thick]   (gate1) -- (output2);
\draw[-,myblue,thick]   (gate2) -- (output2);
\draw[-,myblue,thick]   (gate4) -- (output2);
\draw[-,myblue,thick]   (gate6) -- (output2);

\node[right,rotate=90] at (0,7.25) (legtop1) {\footnotesize Processing 0};
\node[right,rotate=90] at (1,7.25) (legtop2) {\footnotesize Communication 1};
\node[right,rotate=90] at (2,7.25) (legtop3) {\footnotesize Processing 1};
\node[right,rotate=90] at (3,7.25) (legtop4) {\footnotesize Communication 2};
\node[right,rotate=90] at (4,7.25) (legtop5) {\footnotesize Processing 2};
\node[right,rotate=90] at (5,7.25) (legtop6) {\footnotesize Output};
\begin{scope}[on background layer]
\foreach  \i in {1,...,6}{
\coordinate (legbottom\i) at (\i-1,-0.5);
\draw [-,myyellow] (legtop\i) -- (legbottom\i);
}
\end{scope}
\end{tikzpicture} %
}
         \caption{Circuit representation of a non-signaling strategy on graph $G_0$ of \cref{fig:app:toy-graphs-a}, with an arbitrary shared resource. The joint output $(c_1,c_2)$ remain, after $T = 2$ rounds of communication, independent of the graph structure around nodes $11$ and $12$, because the difference lies outside their joint past light-cones.}
         \label{fig:app:two-circuits-b}
     \end{subfigure}
          \hfill
     \begin{subfigure}[t]{0.3\textwidth}
         \centering
        \resizebox{\textwidth}{!}{
\begin{tikzpicture}[square/.style={rectangle}]
\definecolor{myred}{RGB}{255,80,80}
\definecolor{myblue}{RGB}{255,80,80}
\definecolor{mypurple}{RGB}{255,80,80}
\definecolor{mygray}{RGB}{140,140,140}
\definecolor{myyellow}{RGB}{255,200,80}

\def\listnodes{{12,10,8,6,4,2,1,3,5,7,9,11}}
  \foreach \i in {1,...,12}{
   \pgfmathsetmacro{\j}{\listnodes[\i-1]}
     \node at (-0.5,0.5*\i) [circle,fill=white,inner sep=1] (name\j)  {\footnotesize $\j$};
     
    \node at (0,0.5*\i) [semicircle,fill={\ifnum\j=11 black!80\else \ifnum\j=12 black!80\else myred\fi\fi},rotate=90,inner sep=2] (input\j)  {};
    
    \node at (2,0.5*\i) [square,fill={\ifnum\j=11 black!80\else \ifnum\j=12 black!80\else\ifnum\j=9 black!80\else\ifnum\j=7 black!80\else\ifnum\j=8 black!80\else\ifnum\j=10 black!80\else myred\fi\fi\fi\fi\fi\fi},inner sep=2.75] (gate\j)  {};

\node at (4,0.5*\i) [square,fill={\ifnum\j=1 myred\else \ifnum\j=2 myred\else black!80\fi\fi},inner sep=2.5] (output\j)  {};

\node[text={\ifnum\j=1 myred\else \ifnum\j=2 myred\else black\fi\fi}] (out\j) [right of=output\j] {$c_{\j}$};
\draw [->,draw={\ifnum\j=1 myred\else \ifnum\j=2 myred\else black!80\fi\fi},thick] (output\j) -- (out\j);
  }

\node at (-1.25,3.5) [semicircle,rotate=90,fill=myred,inner sep=7.25] (res) {};

\begin{scope}[on background layer]
\draw[-,mygray,thick]   (res) -- (input11);
\draw[-,myred,thick]   (res) -- (input9);
\draw[-,myred,thick]   (res) -- (input7);
\draw[-,mypurple,thick]   (res) -- (input5);
\draw[-,mypurple,thick]   (res) -- (input3);
\draw[-,mypurple,thick]   (res) -- (input1);
\draw[-,mypurple,thick]   (res) -- (input2);
\draw[-,mypurple,thick]   (res) -- (input4);
\draw[-,mypurple,thick]   (res) -- (input6);
\draw[-,myblue,thick]   (res) -- (input8);
\draw[-,myblue,thick]   (res) -- (input10);
\draw[-,mygray,thick]   (res) -- (input12);
\end{scope}

\node at (4,-0) [square,fill=black!15,inner sep=2.75] (phantomoutput11)  {};
\node at (4,-0.5) [square,fill=black!15,inner sep=2.75] (phantomoutput9)  {};
\node at (4,0.5*13) [square,fill=black!15,inner sep=2.75] (phantomoutput12)  {};
\node at (4,0.5*14) [square,fill=black!15,inner sep=2.75] (phantomoutput10)  {};
    
\node at (2,-0) [square,fill=black!15,inner sep=2.75] (phantomgate11)  {};
\node at (2,-0.5) [square,fill=black!15,inner sep=2.75] (phantomgate9)  {};
\node at (2,0.5*13) [square,fill=black!15,inner sep=2.75] (phantomgate12)  {};
\node at (2,0.5*14) [square,fill=black!15,inner sep=2.75] (phantomgate10)  {};

\node at (0,-0) [semicircle,rotate=90,,fill=black!15,inner sep=2] (phantominput11)  {};
\node at (0,-0.5) [semicircle,rotate=90,,fill=black!15,inner sep=2] (phantominput9)  {};
\node at (0,0.5*13) [semicircle,rotate=90,,fill=black!15,inner sep=2] (phantominput12)  {};
\node at (0,0.5*14) [semicircle,rotate=90,,fill=black!15,inner sep=2] (phantominput10)  {};

\node[text=black!15] at (-0.5,0) [circle,fill=white,inner sep=1] ()  {\footnotesize $11$};
\node[text=black!15] at (-0.5,-0.5) [circle,fill=white,inner sep=1] ()  {\footnotesize $9$};
\node[text=black!15] at (-0.5,0.5*13) [circle,fill=white,inner sep=1] ()  {\footnotesize $12$};
\node[text=black!15] at (-0.5,0.5*14) [circle,fill=white,inner sep=1] ()  {\footnotesize $10$};

\draw[-,black!15,thick]   (phantominput12) -- (gate9);
\draw[-,black!15,thick]   (input9) -- (phantomgate12);

\draw[-,black!15,thick]   (phantominput10) -- (gate11);
\draw[-,black!15,thick]   (input11) -- (phantomgate10);

\draw[-,black!15,thick]   (input12) -- (phantomgate9);
\draw[-,black!15,thick]   (phantominput9) -- (gate12);

\draw[-,black!15,thick]   (input10) -- (phantomgate11);
\draw[-,black!15,thick]   (phantominput11) -- (gate10);

\draw[-,black!15,thick]   (phantomgate12) -- (output9);
\draw[-,black!15,thick]   (gate9) -- (phantomoutput12);

\draw[-,black!15,thick]   (phantomgate10) -- (output11);
\draw[-,black!15,thick]   (gate11) -- (phantomoutput10);

\draw[-,black!15,thick]   (gate12) -- (phantomoutput9);
\draw[-,black!15,thick]   (phantomgate9) -- (output12);

\draw[-,black!15,thick]   (gate10) -- (phantomoutput11);
\draw[-,black!15,thick]   (phantomgate11) -- (output10);

\def\listnodes{{
0,0,0,0,
0,0,0,0,
0,0,0,0,
0,0,0,0,
0,0,0,0,
0,0,0,0,
3,5,7,11,
4,6,8,12,
5,9,9,11,
6,10,10,12,
7,9,11,11,
8,10,12,12
}}
  \foreach \i in {11,9,7,8,10,12}{
  \foreach \j in {0,1,2,3}
   \pgfmathsetmacro{\j}{\listnodes[\i*4+\j-4]}  
\draw[-,mygray,thick]   (input\j) -- (gate\i);
}

  \foreach \i in {11,9,7,8,10,12}{
  \foreach \j in {0,1,2,3}
   \pgfmathsetmacro{\j}{\listnodes[\i*4+\j-4]}  
\draw[-,mygray,thick]   (gate\j) -- (output\i);
}

\draw[-,myred,thick]   (input1) -- (gate3);
\draw[-,myred,thick]   (input3) -- (gate3);
\draw[-,myred,thick]   (input4) -- (gate3);
\draw[-,myred,thick]   (input7) -- (gate3);
\draw[-,mygray,thick]   (gate1) -- (output3);
\draw[-,mygray,thick]   (gate3) -- (output3);
\draw[-,mygray,thick]   (gate4) -- (output3);
\draw[-,mygray,thick]   (gate7) -- (output3);

\draw[-,myred,thick]   (input1) -- (gate5);
\draw[-,myred,thick]   (input5) -- (gate5);
\draw[-,myred,thick]   (input7) -- (gate5);
\draw[-,myred,thick]   (input9) -- (gate5);

\draw[-,mygray,thick]   (gate1) -- (output5);
\draw[-,mygray,thick]   (gate5) -- (output5);
\draw[-,mygray,thick]   (gate7) -- (output5);
\draw[-,mygray,thick]   (gate9) -- (output5);

\draw[-,myblue,thick]   (input2) -- (gate4);
\draw[-,myblue,thick]   (input3) -- (gate4);
\draw[-,myblue,thick]   (input4) -- (gate4);
\draw[-,myblue,thick]   (input8) -- (gate4);
\draw[-,mygray,thick]   (gate2) -- (output4);
\draw[-,mygray,thick]   (gate3) -- (output4);
\draw[-,mygray,thick]   (gate4) -- (output4);
\draw[-,mygray,thick]   (gate8) -- (output4);

\draw[-,myblue,thick]   (input2) -- (gate6);
\draw[-,myblue,thick]   (input6) -- (gate6);
\draw[-,myblue,thick]   (input8) -- (gate6);
\draw[-,myblue,thick]   (input10) -- (gate6);
\draw[-,mygray,thick]   (gate2) -- (output6);
\draw[-,mygray,thick]   (gate6) -- (output6);
\draw[-,mygray,thick]   (gate8) -- (output6);
\draw[-,mygray,thick]   (gate10) -- (output6);

\draw[-,mypurple,thick]   (input1) -- (gate1);
\draw[-,mypurple,thick]   (input2) -- (gate1);
\draw[-,mypurple,thick]   (input3) -- (gate1);
\draw[-,mypurple,thick]   (input5) -- (gate1);

\draw[-,mypurple,thick]   (input1) -- (gate2);
\draw[-,mypurple,thick]   (input2) -- (gate2);
\draw[-,mypurple,thick]   (input4) -- (gate2);
\draw[-,mypurple,thick]   (input6) -- (gate2);

\draw[-,myred,thick]   (gate1) -- (output1);
\draw[-,myred,thick]   (gate2) -- (output1);
\draw[-,myred,thick]   (gate3) -- (output1);
\draw[-,myred,thick]   (gate5) -- (output1);

\draw[-,myblue,thick]   (gate1) -- (output2);
\draw[-,myblue,thick]   (gate2) -- (output2);
\draw[-,myblue,thick]   (gate4) -- (output2);
\draw[-,myblue,thick]   (gate6) -- (output2);

\node[right,rotate=90] at (0,7.25) (legtop1) {\footnotesize Processing 0};
\node[right,rotate=90] at (1,7.25) (legtop2) {\footnotesize Communication 1};
\node[right,rotate=90] at (2,7.25) (legtop3) {\footnotesize Processing 1};
\node[right,rotate=90] at (3,7.25) (legtop4) {\footnotesize Communication 2};
\node[right,rotate=90] at (4,7.25) (legtop5) {\footnotesize Processing 2};
\node[right,rotate=90] at (5,7.25) (legtop6) {\footnotesize Output};
\begin{scope}[on background layer]
\foreach  \i in {1,...,6}{
\coordinate (legbottom\i) at (\i-1,-0.5);
\draw [-,myyellow] (legtop\i) -- (legbottom\i);
}
\end{scope}
\end{tikzpicture} %
}
         \caption{Circuit representation of a non-signaling strategy on graph $H$ of \cref{fig:app:toy-graphs-b}, with an arbitrary shared resource. Note the difference between this circuit and the one of \cref{fig:app:two-circuits-b} (namely, the different connections between the top layer and bottom layer of the circuit) is only manifested outside the joint past light-cone of the outputs $(c_1,c_2)$.}
         \label{fig:app:two-circuitsc}
     \end{subfigure}
        \caption{The above circuits represent non-signaling strategies (it includes as special cases the classical strategies and quantum strategies) executed by the nodes $1$ and $2$ in various scenarios.
        The semicircles represent private (\cref{fig:app:two-circuits-a}) or shared (\cref{fig:app:two-circuits-b} and \cref{fig:app:two-circuitsc}) arbitrary-but-non-signaling resources; the wires depict communication (or storage), and the squares are local operations (using possibly private resources). The last layer of gates (or measurements) outputs the individual colors of the nodes (i.e.\ classical variables).       
        Since the nodes start with no knowledge about the identity of their neighbors, the operations of the gates are a priori independent of the graph structure (as long as the nodes have the right degree).        
 For the special cases of classical and quantum strategies, the output distribution can be computed directly from those circuits, which define it uniquely.}
\label{fig:app:two-circuits}
\end{figure}
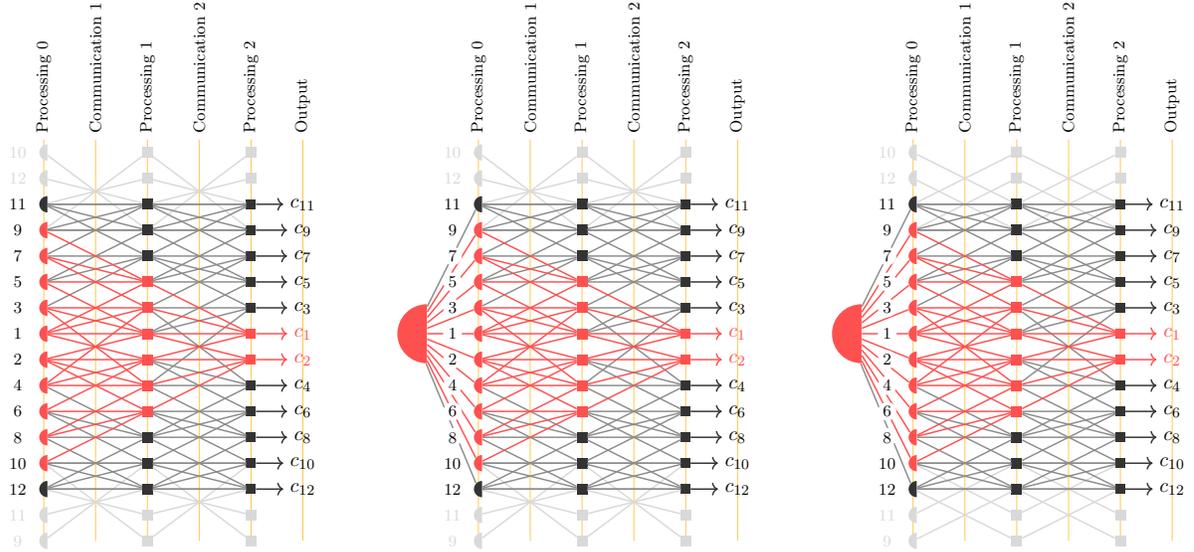

\subsubsection{\boldmath \Qlocal model}
Any quantum strategy can also be represented by the circuit formalism of~\cref{fig:app:two-circuits} by using resources and processing operations that are quantum rather than classical.
Quantum information theory provides a concrete mathematical formalism (based on the tensor product of Hilbert spaces) with rules to represent the potential operations performed by each of the gates of the circuit. 
Semicircles, line wires, and squares now respectively represent the creation of a quantum state (a density matrix), the transfer (or storage) of quantum information, and the processing of quantum information with quantum channels (or quantum measurements in the last layer) represented by completely positive operators.
Once the quantum states and channels are made explicit, the Born rule of quantum information theory~\cite{born1926quantenmechanik, NielsenChuang2011Book} allows computing $\pr{c_1,\dots,c_n \ \st C_{G,T}}$, the output distribution of the strategy on graph $G$.

However, analyzing the \qlocal model directly is complicated.
Instead, we employ the fact that it is possible to bound the time complexity of the \qlocal model by analyzing models that do not depend on the mathematical formalism of quantum information theory (and which are, therefore, arguably simpler).
In the present appendix, we utilize the fact that \qlocal is sandwiched between the randomized \local model (less powerful) and the non-signaling model (the most powerful model satisfying the information-theoretic principles of non-signaling, independence, and device replication we introduce below).
We next introduce the non-signaling model.

\subsubsection{The non-signaling model (with unique identifiers)}\label{app:subsec:causaltheory}

The non-signaling model is the most powerful model of synchronous distributed computing that does not violate physical causality, when a pre-established shared resource exists.

The formulation of the non-signaling model is radically different from the randomized \local and \qlocal models.
The randomized \local and \qlocal models are based on classical (randomized) and quantum information theory which both provide a mathematical formalism to describe information processing gates and to compute the probability distribution of outputs when these gates are placed in a circuit.
To show that some probability distribution is feasible in these models, one needs to propose a valid strategy in these mathematical formalisms.
The non-signaling model does not rely on any underlying mathematical formalism, but on some information-theoretical principles which should not be violated.
More precisely, to show that some probability distribution is feasible in the non-signaling model, one needs to explain how to obtain it, but only to make sure that this probability distribution is compatible with the \emph{non-signaling} principle in a way we now explain.

Before introducing this principle, we state the following definition of light-cones: 

\begin{definition}[Past light-cone]
Consider a circuit.
Consider a subset of gates $R$, and another gate $s \notin R$ in that circuit. 
We say that $s$ is in the past light-cone of $R$ if starting from $s$, one can reach a gate in $R$ by traveling down the circuit.
For instance, in \cref{fig:app:two-circuits-b}, the past light-cone of the second processing gates of nodes $\{1,2\}$ is composed of all gates in red.
\end{definition}
\begin{remark}
     Light-cones allow one to formalize the fact that, in a circuit, the precise time ordering in which the gates process information has no influence over the result, as long as physical causality (that is, information can be transferred only through communication) is preserved. For instance, in \cref{fig:app:two-circuits-a}, our drawing of the circuit represents the processing-1 gate of node~8 as being in the past of the processing-2 gate of node~2. However, by stretching the communication lines, another drawing in which this time ordering is exchanged exists, and this is possible as long as one node is not in the past light-cone of another. The name ``light-cone'' refers to relativity, where physical causality is bounded by the speed of light. Circuits provide an abstracted representation of physical causality that is not based on any notion of spacetime.
\end{remark}

We now state the non-signaling principle\footnote{\citet{chiribella2016} call this principle \emph{no-signaling from the future}.}, which is essentially the only constraint in the non-signaling model.

\begin{definition}[Non-signaling principle] Consider a set of gates, two different circuits $C, C'$ connecting them, and corresponding distributions $\pr{c_1, \dotsc, c_n\ \st \ C}$, $\pr{c_1, \dotsc, c_n \ \st \ C'}$.
Let $U$ be a subset of measurement gates whose past light-cones in $C$ and $C'$ coincide (i.e.\ they are composed of the same gates which are connected in the same way). Then $\pr{\{c_i\}_{i \in U} \ \st \ C}=\pr{\{c_i\}_{i \in U} \ \st C'}$.
\end{definition}

Implicitly, a non-signaling theory assumes that, given a set of gates and a valid way to connect them in a circuit $C$, one obtains a corresponding distribution $\pr{c_1, \dotsc, c_n \ \st \ C}$\footnote{Note that the gates have ``types'' and can only be connected to other gates of matching types: For instance, in our case, the inputs of gates from the second processing step must correspond to the outputs of gates from the first processing step. The theory of circuits can be formalized using category theory --- see, e.g.\ \cite{d2017quantum}.} (here the alphabet of the outputs $c_i$ is not limited: e.g., in case of unexpected circuits, the measurement gates can produce failure outputs).

We are now ready to define the non-signaling model, which is associated to circuits with a pre-shared non-signaling resource such as in \cref{fig:app:two-circuits-b}:
\begin{definition}[Non-signaling model, with unique identifiers]
The distribution
\[\pr{c_1, \dotsc, c_n \ \st \ C_{G,T}}\]
is feasible in the non-signaling model (with unique identifiers) with locality~$T$ on graph $G$ if and only if, for all possible alternative connecting graphs $H$ of the nodes, there exists a distribution $P=\pr {c_1, \dotsc, c_n\ \st \ C_{H,T}}$ such that the no-signaling principle is respected.\footnote{Note that considering scenarios with more than one copy of each node does not allow us to derive additional constraints, since the pre-shared non-signaling resource cannot by definition be cloned and extended to more than one copy of each of its original recipient.}
\end{definition}

Note that $H$ might not be a lifebuoy-shaped graph: in case a gate `discovers' this unexpected situation, it has the possibility to produce a new output from a set of error outputs (i.e.\ it crashes and no further useful constraint can be derived).

\subsubsection{The non-signaling model without unique identifiers}\label{app:subsec:causaltheory_ohne}

We have up until now considered $N=12$ nodes with unique identifiers ranging from 1 to $N=12$. When the nodes do not start with unique identifiers, the situation is slightly more complicated.
In that case, the resulting distribution should be invariant under subgraph isomorphism, as all nodes are running identical programs and the circuit is thus completely symmetric under permutations of nodes\footnote{Note that while a shared classical resource is inherently symmetric, because it can be without a loss of generality taken to be distributed identically to each party, a non-signaling resource is not a priori symmetric under permutation of the parties that it connects. The absence of identifiers thus forces the non-signaling model to pre-share only a subset of all possible non-signaling resources, namely the subset that respects such symmetry.}. For instance, if the nodes were to color the graph of \cref{fig:app:toy-graphs-a} without being a priori assigned unique identifiers, the resulting circuit $C_{G,T}$ would have some symmetries (as all processing gates would be the same in any layer of the respective steps 0, 1 and 2), implying that the distribution should be invariant under several non-trivial graph isomorphisms cyclically permuting the nodes, or inverting the inner and outer cycle of 6 nodes.

\begin{definition}[Non-signaling model, without unique identifiers]
The distribution \[\pr{c_1, \dotsc, c_n \ \st \ C_{G,T}}\] is feasible in the non-signaling model (without unique identifiers) in $T$ communication steps on graph $G$ if and only if the following two conditions are respected: for all possible alternative connecting graphs $H$ of the nodes, there exists a distribution $P=\pr {c_1, \dotsc, c_n\ \st \ C_{H,T}}$ such that the no-signaling principle is respected; and the distributions in $G$ and $H$ are invariant under subgraph isomorphisms.
\end{definition}

This definition can also be generalized to the case where several nodes with the same identifiers are present in $G$, or where the number of identifiers ranges from 1 to $M \neq N$.

\subsection{\texorpdfstring{\boldmath Graph $\GG$ is not 2-colorable in $T = 2$ rounds}{Graph G is not 2-colorable in T = 2 rounds}}\label{app:sec:G-not-colorable}

We now show by contradiction that the nodes $1, \dots, 12$ cannot 2-color the set $\GG$ of lifebuoy-shaped graphs with a non-signaling strategy (including pre-shared non-signaling resources) in $T = 2$ rounds of communication.
More precisely, we prove that if there existed such a non-signaling algorithm able to color any $G\in\GG$, then the same algorithm would also color the cheating graph $H$ with 2 colors, which is impossible.

\begin{proof}
Assume by contradiction that there exists such a non-signaling algorithm, that is, there exist gates as in \cref{fig:app:two-circuits-b} such that the resulting probability distribution $\pr{c_1,\dots,c_{12} \ \st \ C_{G,T}}$ is a proper coloring for all lifebuoy-shaped graphs $G\in \GG$. Such proper coloring means that for all $G\in \GG$, and for all nodes $u,v$ that are neighbors in $G$, it holds that $\pr{c_u\neq c_v \ \st \ C_{G,T}}=1$.

We consider the same nodes performing the same gates, but we change the edges  so that the connected nodes now form the graph $H$ in \cref{fig:app:toy-graphs-a}.
Let $\pr{c_1,\dots, c_{12}\ \st \ C_{H,T}}$ be the distribution of output colors in that new configuration\footnote{While we are deceiving the nodes by promising a lifebuoy-shaped connecting graph but imposing instead $H$, the nodes cannot locally detect the fraud in $T = 2$ communication steps or less. (More formally, the past light-cone in $H$ of any node is then compatible with a lifebuoy-shaped graph --- an individual node cannot detect the difference and must therefore, according to the non-signaling principle, output a color as if it were in a lifebuoy-shaped graph.)}.
Consider two nodes connected by an edge in $H$: by symmetry of $H$, we can without a loss of generality assume that these are nodes $1$ and $2$.
We introduce the lifebuoy-shaped graph $G_0$ of \cref{fig:app:toy-graphs-a} ($G_0\in G$ depends on our choice of nodes, here $1,2$, in $H$).
Then, we observe that the common past light-cones of the two nodes is the same when connected through $H$ or through $G_0$ (compare the circuit of \cref{fig:app:two-circuits-b} with the circuit of \cref{fig:app:two-circuitsc}).  It hence holds that $\pr{c_1\neq c_2 \ \st \ C_{H,T}} = 1 = \pr{c_1\neq c_2  \ \st \ C_{G_0,T}}$ due to the non-signaling principle.
We conclude that the non-signaling distributed algorithm outputs a 2-coloring for $H$, which is impossible.
\end{proof}

\subsection{Bounded-dependence model}\label{app:sec:no-shared-resource_mit}

The bounded-dependence model is similar to the non-signaling model, but the former does not allow pre-shared non-signaling resources between the nodes. The following two new principles are needed to define feasible distributions in the bounded-dependence model. 

\begin{definition}[Independence principle]
Consider a set of gates connected in a circuit $C$, and the corresponding distribution $\pr{ c_1, \dotsc, c_n \ \st \ C}$.
Let $U,V$ be two subsets of measurement gates producing the outputs $\{c_u\}_{u\in U}$ and $\{c_v\}_{v\in V}$, respectively. If the past light-cones of $U$ and $V$ do not intersect, then $\pr{ \{c_w\}_{w \in U\cup V} \ \st \ C}=\pr{ \{c_u\}_{u \in U} \ \st \ C}\cdot \pr{ \{c_v\}_{v \in V} \ \st \ C}$, that is, their output distributions are independent.
\end{definition}

\begin{definition}[Device-replication principle]
Identical and independent copies of non-signaling gates and non-signaling resources can be prepared\footnote{Note that this principle does not imply that one can duplicate \emph{unknown} non-signaling resources: device-replication is compatible with the quantum no-cloning theorem.}.
\end{definition}

Then, we obtain the following definition:
\begin{definition}[Bounded-dependence model, with unique identifiers]
The distribution \[\pr{ c_1, \dotsc, c_n \ \st \ C_{G,T}}\] with unique identifiers has \emph{bounded dependence} with locality~$T$ on graph $G$ (without pre-shared non-signaling resources) if and only if, for all possible alternative connecting graphs $H$ of the nodes and their replicates, there exists a distribution $\pr{ c_1, \dotsc, c_m \ \st \ C_{H,T}}$ such that the non-signaling and independence principles are respected, and such that the distribution is invariant under subgraph isomorphisms.
\end{definition}

Note a subtlety related to subgraph isomorphisms in the bounded-dependence model. There is the variant where the nodes have identifiers in $G$, and the one where they do not. When the nodes do not have any identifiers, the class of subgraph isomorphisms of all alternative graphs $H$ created out of the nodes in $G$ and their replicates is obviously larger than with identifiers, because the subgraph isomorphisms must respect the identifiers.
However, even if the nodes of $G$ do have distinct identifiers, as $H$ is created out of possibly many copies of the original nodes of $G$, $H$ might contain several nodes with the same identifiers.
Hence, the group of subgraph isomorphisms of $H$ might be nontrivial even if all nodes in $G$ have distinct identifiers. 
For instance, in lifebuoy-shaped graphs, one could consider the case represented in \cref{fig:app:toy-graph-bounded}, which starts from the graph $G_0$ in \cref{fig:app:toy-graphs-a}, duplicates all nodes, and constructs a new graph $H$ of 24 nodes with identifiers ranging from 1 to 12 with one non-trivial graph isomorphism cyclically permuting the nodes.

\begin{figure}
     \centering
              \includegraphics[width=1\textwidth/2]{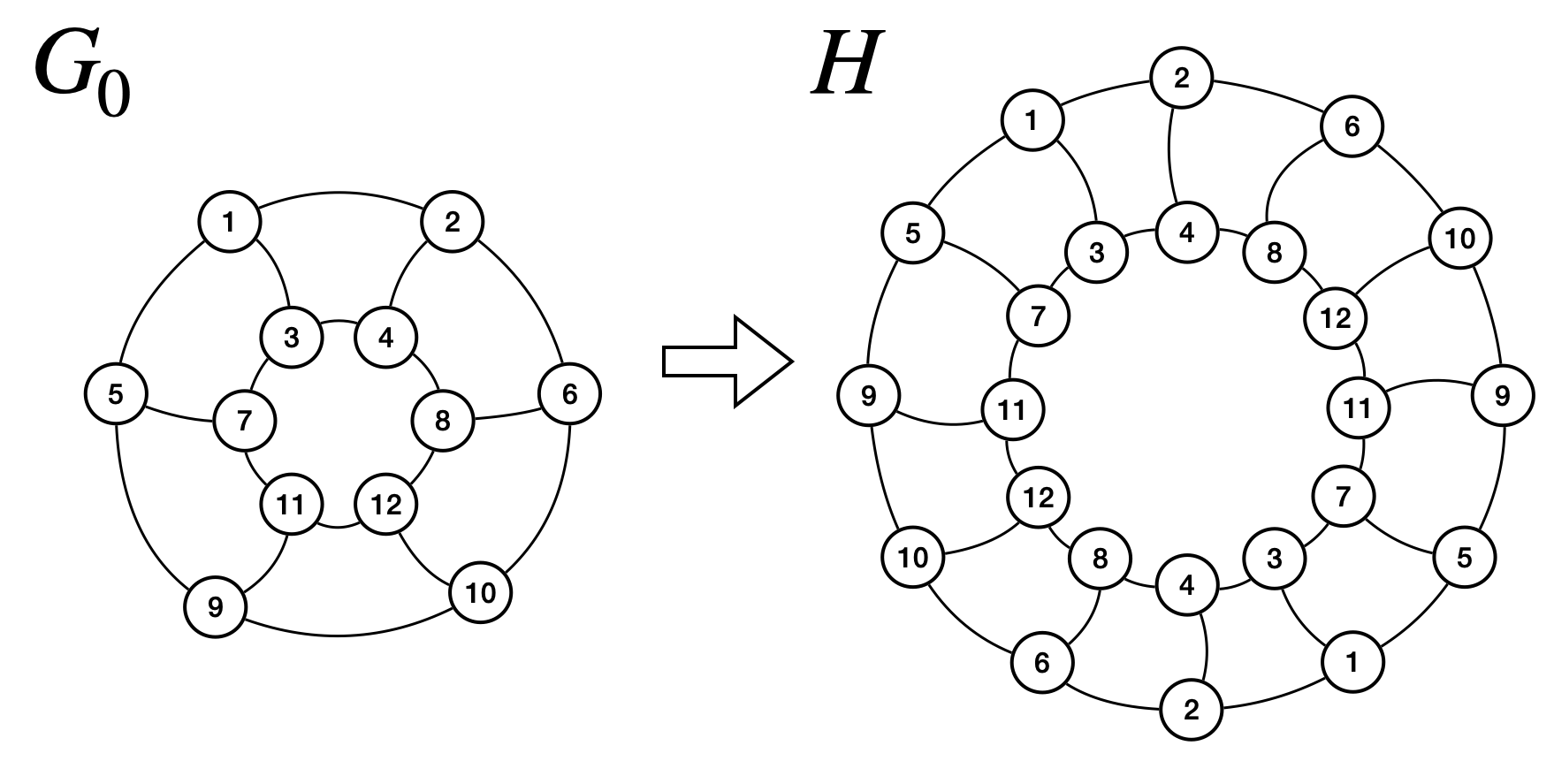}
                       \caption{In the bounded-dependence model, one can find non-trivial subgraph isomorphisms even when the nodes are provided with unique identifiers.}
         \label{fig:app:toy-graph-bounded}
\end{figure}

\subsubsection{Relation with finitely-dependent distributions}

Our bounded-dependence model is directly connected to the concept of finitely-dependent distributions introduced in mathematics~\cite{aaronson1989algebraic, holroyd2016, holroyd-2017-finitary-coloring, holroyd2018, spinka2020, holroyd2024}. In this framework, in a graph $H$, the color $c_i$ produced by each node $i$ is seen as the result of a random process $C_i$. The set of all random processes $\{C_i\}_i$ is said to be $k$-dependent in graph $H$ if, for any two subsets $U,V$ of the nodes of $H$ which are at least at distance $k+1$, the two sets of processes $\{C_u\}_{u\in U}$ and ${C_v}_{v\in V}$ are independent. In our notation, assuming even $k$ and taking $T=k/2$, this is equivalent to asking that $\pr{ \{c_w\}_{w \in U\cup V} \ \st \ C_{H,T}}=\pr{ \{c_u\}_{u \in U} \ \st \ C}\cdot \pr{ \{c_v\}_{v \in V} \ \st \ C_{H,T}}$.

Hence, the question of the existence of distributions with bounded dependence that solve some problem can directly be formulated in terms of the existence of finitely-dependent distributions over a family of graphs that satisfy additional constraints of compatibility (to satisfy the non-signaling principle) and symmetries (to cope with subgraph isomorphisms).
\begin{figure}[p]
    \centering
    \includegraphics[page=1]{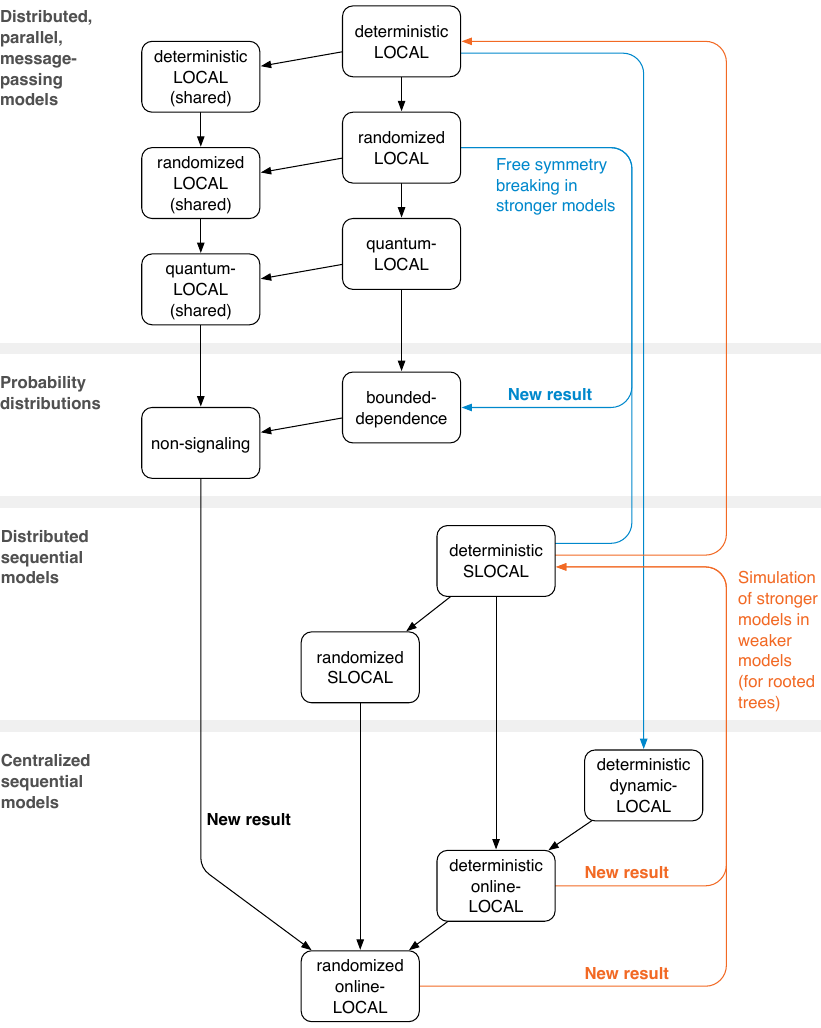}
    \caption{Landscape of models and their relations for LCL problems. A black arrow $X \to Y$ indicates that model $Y$ is at least as strong as model $X$ and can simulate any algorithm designed there. A blue arrow $X {\color{myblue} {}\to{}} Y$ indicates that we additionally get symmetry-breaking for free: locality $O(\log* n)$ in model $X$ implies locality $O(1)$ in model $Y$. An orange arrow $X {\color{myorange} {}\to{}} Y$ indicates that \emph{at least in rooted trees} we can simulate highly-localized algorithms in model $X$ using the weaker model $Y$ so that e.g.\ locality $O(\log* n)$ in model $X$ implies locality $O(\log* n)$ in model $Y$ (see \cref{thm:olocal-slocal-simulation,thm:olocal-sim:reduction-slocal} and \cite{ghaffari2017} for details).}\label{fig:landscape}
\end{figure}
 \end{document}